\documentclass[prc,floatfix,groupedaddress,nofootinbib,showpacs,preprintnumbers,
amsmath,amssymb,amsfonts,superscriptaddress,widetable] {revtex4-1}

\usepackage{graphicx}
\usepackage{dcolumn}
\usepackage{mathrsfs,amsmath,amssymb,color}
\usepackage{bm}
\usepackage{bbm, dsfont}
\usepackage{relsize}

\usepackage[normalem]{ulem}

\newcommand{\GEsq}{G_{\!E}^{2}}
\newcommand{\GMsq}{G_{\!M}^{2}}
\newcommand{\rp}{r_{\!p}}

\def\qed{\hfill$\diamondsuit$}

\newcommand{\Rmnum}[1]{\expandafter\romannumeral #1}

\newtheorem{proposition}{{Proposition}}
\newtheorem{proof}{{Proof}}

\newcommand{\RNum}[1]{\uppercase\expandafter{\romannumeral #1\relax}}

\def\m{\mathcal}

\def\mr{\mathrm}

\newcommand{\ind}{\mathbbm{1}}

\def\T{{ \mathrm{\scriptscriptstyle T} }}

\begin{document}

\title{Revisiting the proton-radius problem using constrained Gaussian processes }

\author{Shuang Zhou}
\affiliation{Department of Statistics, Texas A\&M University, College Station, TX 77843}
\author{P. Giuliani}
\author{J. Piekarewicz}
\affiliation{Department of Physics, Florida State University, Tallahassee, FL 32306} 
\author{Anirban Bhattacharya}
\affiliation{Department of Statistics, Texas A\&M University, College Station, TX 77843}
\author{Debdeep Pati}
\affiliation{Department of Statistics, Texas A\&M University, College Station, TX 77843}

\begin{abstract}

\textbf{Background:} The ``proton radius puzzle'' refers to an eight-year old problem 
that highlights major inconsistencies in the extraction of the charge radius of the proton 
from muonic Lamb-shift experiments as compared against experiments using elastic
electron scattering. For the latter, the determination of the charge radius involves an 
extrapolation of the experimental form factor to zero momentum transfer.

\textbf{Purpose:} To estimate the proton radius by introducing a novel and powerful 
non-parametric model based on a constrained Gaussian process to model the electric form factor of the proton. 

\textbf{Methods:} Within a Bayesian paradigm, we develop a model flexible enough 
to fit the data without any parametric assumptions on the form factor. The Bayesian
estimation is guided by imposing only two physical constraints on the form factor: 
(a) its value at zero momentum transfer (normalization) and (b) its overall shape, 
assumed to be a monotonically decreasing function of the momentum transfer. 
Variants of these assumptions are explored to assess the impact of these constraints.

\textbf{Results:} So far our results are inconclusive in regard to the proton puzzle, as
they depend on both, the assumed constrains and the range of experimental 
data used to fit the Gaussian process. For example, if only low momentum-transfer 
data is used, adopting only the normalization constraint provides a value compatible 
with the smaller muonic result, while imposing only the  shape constraint favors the 
larger electronic value.

\textbf{Conclusions:} We have presented a novel technique to estimate the proton 
radius from electron scattering data based on a non-parametric Gaussian process.
We have shown the major impact of the physical constraints imposed on the form
factor and of the range of experimental data used to implement the extrapolation.
In this regard, we are hopeful that as this technique is refined and with the anticipated 
new results from the PRad experiment, we will get closer to resolve of the puzzle.  
\end{abstract}

\maketitle

\section{Introduction}
\label{Sec:Introduction}

Nuclear Physics is an extremely broad field of science whose mission is to understand 
all manifestations of nuclear phenomena\,\cite{LongRangePlan}. Regardless of whether 
probing individual nucleons, atomic nuclei, or neutron stars, a common theme across this 
vast landscape is the characterization of these objects in terms of their mass and radius.
Indeed, shortly after the discovery of the neutron in 1932, Gamow, Weizs\"acker, Bethe, 
and Bacher formulated the ``liquid-drop'' model to estimate the masses of atomic 
nuclei\,\cite{Weizsacker:1935,Bethe:1936}. Since then, remarkable advances in 
experimental techniques have been exploited to determine nucleon and nuclear masses 
with unprecedented precision; for example, the rest mass of the proton is known to a 
few parts part in a billion\,\cite{Mohr:2015ccw}. Similarly, starting with the pioneering work 
of Hofstadter in the late 1950's \cite{Hofstadter:1956qs} and continuing to this 
day\,\cite{DeJager:1987qc,Fricke:1995,Angeli:2013}, elastic electron scattering has provided 
the most accurate and detailed picture of the distribution of charge in nuclear systems. Although 
not as impressive as in the case of nuclear masses, the charge radii of atomic nuclei has 
nevertheless been determined with extreme precision; for example, the charge radius of 
${}^{208}$Pb is known to about two parts in 10,000\,\cite{Angeli:2013} (or 
$R_{\rm ch}^{208}\!=\!5.5012(13)\,{\rm fm}$). Given such an impressive track record, 
it came as a shocking surprise that the accepted 2010-CODATA value for the charge 
radius of the proton obtained from electronic hydrogen and electron scattering was in 
stark disagreement with a new result obtained from the Lamb shift in muonic 
hydrogen\,\cite{Pohl:2010zza}. This unforeseen conflict with the structure of the 
proton has given rise to the ``proton radius 
puzzle"\,\cite{Pohl:2013yb,Bernauer:2014cwa,Carlson:2015jba}, 

The value of the charge radius of the proton $\rp\!=\!0.84087(39)\,{\rm fm}$ determined 
from muonic hydrogen\,\cite{Pohl:2010zza,Pohl:2013yb} differs significantly (by $\sim$4\% or 
nearly 7$\sigma$) from the recommended CODATA value of $\rp\!=\!0.8775(51)\,{\rm fm}$.
Note that the CODATA value is obtained by combining the results from both electron scattering 
and atomic spectroscopy\,\cite{Pohl:2013yb,Mohr:2015ccw,Carlson:2015jba}. The muonic
measurement is so remarkably precise because the muon---with a mass that is more than 200 
times larger than the electron mass and thus a Bohr radius 200 times smaller---is a much more 
sensitive probe of the internal structure of the proton. Of great relevance to the proton puzzle is 
the recent measurement of the 2S-4P transition frequency in electronic hydrogen that suggests
a smaller proton radius of $\rp\!=\!0.8335(95)\,{\rm fm}$---in agreement with the result from 
muonic hydrogen\,\cite{Beyer79}. Although significant, it remains to be understood why the present 
extraction differs from the large number of spectroscopic measurements carried out in electronic 
hydrogen throughout the years.

As in the case of earlier physics puzzles---notably the ``solar neutrino problem"---one attempts to 
explain the discrepancy by exploring three non-mutually-exclusive options: (a) the experiment 
(at least one of them) is in error, (b) theoretical models used in the extraction of the proton radius 
are the culprit (see for example\,\cite{Robson:2013nwa} and references contained therein), 
or (c) there is new physics that affects the muon differently than the electron. Indeed, hints of 
possible violations to lepton universality are manifested in the anomalous magnetic moment 
($g\!-\!2$) of the muon\,\cite{Bennett:2006fi} and in certain decays of the B-meson into either a 
pair of electrons or a pair of muons\,\cite{Aaij:2014}. 

In an effort to resolve the ``proton radius puzzle" a suite of experiments in both spectroscopy and
lepton-proton scattering are being commissioned. Spectroscopy of both electronic and muonic atoms, 
as already initiated by Beyer {\sl et al.}\,\cite{Beyer79}, will continue with a measurement of a variety 
of transitions to improve both the value of the Rydberg constant and the charge radius of the proton;
note that the Rydberg constant and $\rp$ are known to be highly correlated. Lepton scattering 
experiments are planned at both the Thomas Jefferson National Accelerator Facility (JLab) and at 
the Paul Scherrer Institute (PSI). The proton radius experiment (PRad) at JLab has already collected 
data in the momentum-transfer range of 
$Q^{2}\!=\!(10^{-4}\,$--$\,10^{-1})\,{\rm GeV}^{2}$\,\cite{Gasparian:2017cgp}, a wide-enough region
to allow for comparisons against the most recent Mainz data\,\cite{Bernauer:2010wm}, but also to 
extend the Mainz data to significantly lower values of $Q^{2}$. Finally, the Muon Proton Scattering 
Experiment (MUSE) will fill a much-needed gap by determining $\rp$ from the scattering of both 
positive and negative muons of the proton. These experiments will be conducted concurrently with 
electron scattering measurements in an effort to minimize systematic uncertainties\,\cite{Gilman:2013eiv}. 
 
Within this broad context our contribution is rather modest, as our main goal is to address how best to 
extract the charge radius of the proton form electron scattering data. The view adopted here is that the 
puzzle lays not in the experimental data, but rather in the extraction of the proton radius from the 
scattering data. The proton charge radius is related to the slope of the electric form factor of the proton 
$G_{\rm E}(Q^{2})$ at the origin, {\sl i.e.,} at $Q^{2}\!=\!0$ (see Sec. \ref{Sec:Formalism}). Despite 
heroic efforts at both Mainz\,\cite{Bernauer:2010wm} and JLab\,\cite{Gasparian:2017cgp} to determine 
$G_{\rm E}(Q^{2})$ at extremely low values of $Q^{2}$, a subtle \textit{extrapolation} to $Q^{2}\!=\!0$ 
is unavoidable. Given the current data available, the value one can obtain for the proton radius from the 
extrapolation is quite sensitive to the model used to describe the form factor. In a first attempt at mitigating 
such uncontrolled 
extrapolations, Higinbotham and collaborators have brought to bear the power of statistical methods 
into the solution of the problem\,\cite{Higinbotham:2015rja}; see also\,\cite{Yan:2018bez}. They have 
concluded that ``statistically justified linear extrapolations of the extremely-low-$Q^{2}$ data produce a 
proton charge radius which is consistent with the muonic results and is systematically smaller than the 
one extracted using higher-order extrapolation functions". However, recent analyses of electron scattering 
data that suggest smaller proton radii consistent with the muonic Lamb shift have been called into 
question\,\cite{Bernauer:2016ziz}. Moreover, much controversy has been generated around the optimal 
(``parametric") model that should be used to fit the electric charge form factor of the proton---ranging from 
monopole, to dipole, to polynomial fits, to Pade' approximants, among many others. In an effort to eliminate 
the reliance on specific functional forms, we introduce a method that does not assume a particular 
parametric form for the form factor. Such a nonparametric approach aims to ``let data speak for itself''
without introducing any preconceived biases. Although the nonparametric approach does not assume a 
particular form for the form factor, several constraints justified by physical considerations are imposed. In 
essence, a \textit{nonparametric Bayesian} curve fitting procedure that incorporates various physical 
constraints is used to provide robust predictions and uncertainty estimates for the charge radius of the 
proton. In our analysis we use the 1422 data points from the Mainz collaboration \cite{bernauer2011high,bernauer2014electric,bernauer2011bernauer}.

The paper has been organized as follows. In Sec.~\ref{Sec:Formalism}, we introduce some of the
basic concepts necessary to understand the measurement of the electric form factor of the proton.
After such brief introduction, we explain the critical concepts behind our nonparametric approach,
including the selection of the basis functions and the Gaussian process used for their calibration. 
A synthetic data example is presented in Sec.~\ref{sec:sims} and the electron-scattering data analysis 
is presented Sec.~\ref{sec:real}. We offer our conclusions and some perspective for future
improvements in Sec.\,\ref{Sec:Conclusions}. Finally, several details about the implementation of 
the model and on the analysis on both synthetic and real data are presented in the various Appendices.

\section{Formalism}
\label{Sec:Formalism}

We start this section with a brief introduction to elastic electron scattering with particular emphasis 
on the determination of the electric form factor of the proton from the scattering data. Then, we 
proceed in significant more detail to describe the formalism associated with the determination of
the charge radius of the proton by extrapolating the experimental data to zero momentum transfer.

\subsection{Electron scattering}
\label{Sec:Electron scattering}

In the one-photon exchange approximation, the most general expression for the elastic cross section 
consistent with Lorentz and parity invariance is encoded in two Lorentz-scalar functions: the electric 
$G_E$ and magnetic $G_M$ form factors of the proton. That is,
\begin{equation}
 \frac{d\sigma}{d\Omega} = \left(\frac{d\sigma}{d\Omega}\right)_{\!\rm Mott}
 \left(\frac{\GEsq(Q^{2}) + \tau \GMsq(Q^{2})}{1+\tau} +
 2\tau \GMsq(Q^{2})\tan^{2}(\theta/2)\right) \,,
 \label{CrossSection}
\end{equation}
where the square of the four-momentum transfer is given by:
\begin{equation}
 Q^{2} \equiv -(p'-p)^{2} = 4EE'\sin^{2}(\theta/2).
 \label{Q2}
\end{equation}
Note that $E$ $(E')$ is the initial (final) energy of the electron, $\theta$ is the scattering angle (all in the laboratory 
frame), $\tau\!\equiv\!Q^{2}/4M^{2}$, and $M$ is the mass of the proton. The internal structure of the proton is 
imprinted in the two form factors, with the electric one describing (in a non-relativistic picture) the distribution 
of charge and the magnetic one the distribution of magnetization. Finally, the Mott cross section introduced in 
Eq.(\ref{CrossSection}) represents the scattering of a massless electron from a spinless and structureless point 
charge. That is,
\begin{equation}
 \left(\frac{d\sigma}{d\Omega}\right)_{\!\rm Mott} = 
  \frac{4\alpha^{2}}{Q^{4}}\frac{E^{\prime3}}{E}\cos^{2}(\theta/2) =
  \frac{\alpha^{2}}{4E^{2}\sin^{4}(\theta/2)}\frac{E^{\prime}}{E}\cos^{2}(\theta/2), 
 \label{Mott}
\end{equation}
where $\alpha$ is the fine structure constant. 

In a non-relativistic picture, the electric form factor of the proton is related to the Fourier transform of its spatial 
distribution of charge as follows:
\begin{equation}
 G_{E}(Q^{2}) = \int\!\rho_{{}_{\!E}}(r)\mathlarger{e}^{i{\bf Q}\cdot{\bf r}} d^{3}r =
    \int\!\rho_{{}_{\!E}}(r)\left(1-\frac{Q^{2}}{3!} r^{2}+\frac{Q^{4}}{5!} r^{4}+\ldots\right)d^{3}r = 
    1 - \frac{Q^{2}}{6}\langle r_{\!{}_{E}}^{2}\rangle + \frac{Q^{4}}{120}\langle r_{\!{}_{E}}^{4}\rangle + \ldots 
 \label{GE}
\end{equation}
This equation suggests that low-energy --or long wavelength-- electrons are unable to resolve the internal structure 
of the proton and are therefore only sensitive to its entire charge. As the momentum transfer increases and the 
wavelength becomes commensurate with the proton size, finer details may now be resolved. In particular, the 
charge radius of the proton is defined as: 
\begin{equation}
 \rp^{2} \equiv \langle r_{\!{}_{E}}^{2}\rangle = -6 \frac{dG_{E}}{dQ^{2}}\bigg|_{Q^{2}=0}.
\label{PRadius}
\end{equation}
Although the above expression for $\rp$ was motivated using non-relativistic arguments, its connection to 
the derivative of the electric form factor has been universally adopted as the definition of the proton radius. Based on this description we introduce the following expressions that are the cornerstone of the nonparametric approach.
\begin{align}
 & G_{E}(Q^{2}\!=\!0) =1, \label{GE0}\\
 & G_{\!E}^{\prime}(Q^2) \equiv \frac{dG_{E}}{dQ^{2}} <0, \label{GE1}\\
 & G_{\!E}^{\prime\prime}(Q^2) \equiv \frac{d^{2}G_{E}}{d(Q^{2})^{2}} >0.\label{GE2} 
\end{align}\label{GEs}


The first equation \eqref{GE0} is model independent since it is directly related to the charge of the proton. The other two equations \eqref{GE1}-\eqref{GE2}, which we will call the shape constraints, are not directed guaranteed by the above definitions, but rather, are deduced from the analytic properties of the form factor, see for example \cite{alarcon2018accurate,alarcon2018nucleon}.

\subsection{Modeling the electric form factor of the proton}
\label{Sec:Modeling}

Having introduced the electric form factor of the proton we now proceed to build a flexible nonparametric model 
that will allow us to extrapolate $G_{E}(Q^{2})$ to $Q^2\!=\!0$. We are interested in studying the impact of the
different constraints displayed in Eq.\eqref{GEs} on the estimation of $\rp$. Hence, we define four model variants 
that will take into account the different combinations of the constraints: 1) cGP: fully constrained model 
$\big($Eq.\eqref{GE0} and \eqref{GE1}-\eqref{GE2}$\big)$; 2) c0GP: constraint at zero $\big($Eq. \eqref{GE0}$\big)$; 3) c1GP: 
shape constraints $\big($Eq. \eqref{GE1}-\eqref{GE2}$\big)$; 4) uGP: unconstrained model $\big($none of the equations in 
Eq. \eqref{GEs} are taken into account$\big)$. 

Our main goal is to incorporate the general constraints given in Eq.\,(\ref{GEs}) into the estimation procedure without making parametric assumptions on the functional form of $G_{E}(Q^{2})$. The available experimental data will guide the shape of such nonparametric curve, ultimately allowing us to estimate $\rp$. To facilitate the implementation of the nonparametric approach, we assume without loss of generality that the ``basis functions" (see \ref{basis_subs}) employed to model the $G_E(Q^2)$ curve are defined in the closed interval $[0,1]$. We select a maximum value of $Q^2$, $Q^2_{\rm max}$, up to where the analysis is performed, a selection that has been shown to impact the estimation of $r_p$. Once the momentum-transfer range has been selected, $0\!\le\!Q^{2}\!\le\!Q_{\rm max}^{2}$, we define the dimensionless scaled variable $x$ as $x\!=\!Q^{2}/Q_{\rm max}^{2}$.

We note that although the condition $G_E(0)=1$ \eqref{GE0} is ultimately related to the charge of the proton, experimental systematic errors can have an appreciable impact on the fulfillment of this constraint in the obtained data. It has become a customary practice (see for example \cite{Yan:2018bez}) to represent the observed values as $f(Q^2)=n_0G_E(Q^2)$, where $n_0$ is a floating normalization parameter, $f(Q^2)$ are the observed values and $G_E(Q^2)$ is the true proton form factor. We can identify in our framework the choice $n_0=1$ with the requirement that our model estimate for the form factor has the fixed value of $1$ at $Q^2=0$ (cGP and c0GP). 
Instead, leaving $n_0$ as an adjustable parameter corresponds to the cases c1GP and uGP.

In the following sections we describe in detail the construction of the fully constrained model cGP, pointing out the possible differences that might be taken into account for the construction of the other three. Most of the details regarding this matter are shown in the Appendix \ref{ssec:variant}.

\subsection{Approximating $G_E$: basis construction}\label{basis_subs}
We start be defining a working grid formed by a collection of $N\!+\!1$ equally spaced points $x_j\!=\!j/N$ in the closed 
interval $[0, 1]$. We adopt the notation of \cite{maatouk2017gaussian} to define a set of basis functions:

\begin{equation}\label{eq:basis}
 h_j(x) = 
  \begin{cases} 
    1-N|x-x_j|, & \text{if }  |x-x_j|\leq 1/N;\\
   0,       & \text{otherwise}.
  \end{cases}
\end{equation}

It is particularly convenient to also define the corresponding integrals of $h_{j}(x)$ as follows:

\begin{align}
 \psi_j (x) & = \int_0^x h_j(t)\,dt, \label{eq:psi} \\ 
 \phi_j (x) & = \int_0^x dt \int_0^t h_j(s)\,ds.
 \label{PsiPhi}
\end{align}

Although analytic expressions for both $\psi_j (x)$ and $\phi_j (x)$ are readily available, it is more illuminating to 
display their behavior in pictorial form, as in Fig.\,\ref{BasisFunction}(a). 
The basis functions $h_j(x)$ can be used to approximate any continuous function $f(x)$ by linearly interpolating between 
the grid points. That is,

\begin{equation}
 f(x) \approx \sum_{j=0}^N f(x_j) h_j(x).
\end{equation}

To illustrate the quality of the approximation, we used a grid of size $N\!=\!10$ to display in Fig.\,\ref{BasisFunction}(b) 
the results for a dipole function of the form:

\begin{equation}\label{Dipole_Def}
 f(x) = \left(1+\frac{\bar{r}_{p}^{2}x}{12}\right)^{-2}, 
\end{equation}

where $x\!=\!Q^{2}/Q_{\rm max}^{2}$, $Q_{\rm max}\!=\!25.01\,{\rm fm}^{-1}$, and $\bar{r}_{p}\!=\!r_{p}Q_{\rm max}\!=\!4.21$. The apex of each triangle, namely, the scale factor multiplying each basis function $h_j(x)$, is the value of the dipole function 
at the j$_{\rm th}$ grid point, or $f(x_j)$. The approximation is so accurate that the underlying exact dipole function (shown in
red) is difficult to discern. As we show next, for the purpose of extracting the proton radius it is better not to approximate directly 
the electric form factor $G_{E}(Q^{2})$ using the basis functions $h_{j}(x)$, but rather the smoother set of related functions 
$\phi_{j}(x)$ defined in Eq.\,(\ref{PsiPhi}). To do so, we invoke the fundamental theorem of calculus for any twice differentiable 
function $f(x)$ defined on the closed interval $[0, 1]$. That is,
\begin{align}
 f(x) = f(0) + x f'(0) +\int_0^x dt \int_0^t f''(s) ds.
\end{align} 
If we now approximate $f''\!(s)$ under the integral sign using the basis functions $h_j(x)$ we obtain:
\begin{align} 
 f(x) \approx f(0) + x f'(0) + \sum_{j=0}^{N} f{''}(x_j) \int_0^x dt \int_0^t  h_j(s)ds =
 f(0) + x f'(0) + \sum_{j=0}^{N} f''(x_j) \, \phi_j(x).
 \label{m:approx}
\end{align}

\begin{figure}
 \includegraphics[width=0.8 \textwidth]{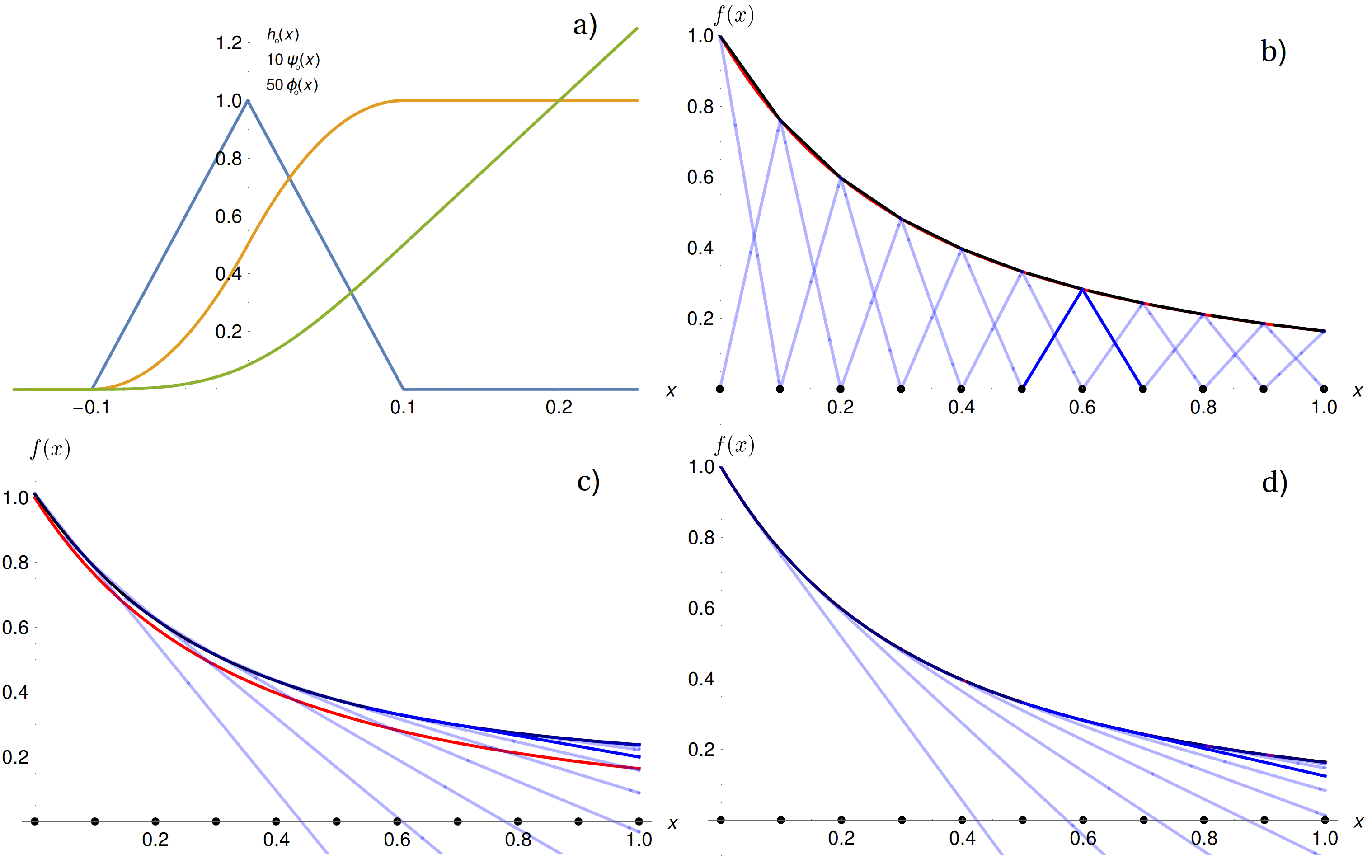}
  \centering
  \caption{(a) Functions $h_0(x)$, $\psi_0 (x)$ and $\phi_0 (x)$ for $N=10$. The functions $\psi_0$ and $\phi_0$ have been rescaled by a factor of 10 and 50 respectively. (b) Approximation (black) of the dipole function (Red) using the basis functions $h_j(x)$ (blue) on 11 gridpoints between 0 and 1 (black dots). The function $h_7(x)$, which is centered at $x=0.6$, is highlighted to illustrate its ``spike'' form. The approximation matches the function so well that the true red curve is hard to see. (c)-(d) Approximation (black) of the same dipole function (red) using the basis functions $\phi_j$ as in Eq.\,\eqref{m:approx} (c), and Eq.\,\eqref{eq:model1} (d). In both cases, the functions $\phi_j$ are plotted starting from a neighboring of their respective grid point and matching their value and their slope at the grid point with the complete approximation (black curve). In both cases the function $\phi_7(x)$ is highlighted. When the coefficients of the $\phi_j$ functions are fitted instead of matched to second derivatives (d) the red curve is hard to see again.}
  \label{BasisFunction}
\end{figure}

This approximation to the exact dipole is shown in Fig.\,\ref{BasisFunction}(c) together with the underlying behavior of
$\phi_j(x)$. In this case the approximation to the exact dipole is not as accurate as in Fig.\,\ref{BasisFunction}(b). However, 
in a regression problem neither the function, nor its derivative at $x\!=\!0$, nor the values of all second derivatives may be
known. Hence, we characterize our regression model in terms of $(N\!+\!3)$ free parameters $\xi_{j}$ that will be obtained 
from a suitable fit to the experimental data. That is,
\begin{align}
 f(x) \equiv f_{\xi}(x) \approx \xi_1 + \xi_2 \, x + \sum_{j=0}^N \xi_{j+3} \, \phi_j(x).
\label{eq:model1}
\end{align}
As displayed in Fig.\,\ref{BasisFunction}(d), once this scheme is adopted, the agreement with the real dipole function 
becomes excellent. Clearly, one great advantage of Eq.\,(\ref{eq:model1}) is that values for the floating normalization
and mean-square radius are directly encoded in $\xi_{1}$ and $\xi_{2}$. Moreover, this approximation has a nice 
physical underpinning. If we regard $f(t)$ as the one-dimensional trajectory of a particle as a function of time $t$, then 
the approximation:
\begin{align}\label{model_eq2}
f(t) \approx f(0) + t f'(0) + \sum_{j=0}^{N} f''(t_j) \, \phi_j(t),
\end{align}
may be explained as follows. At time $t=0$ the particle starts at a position $f(0)$ with an initial velocity $f'(0)$. As time evolves, corrections to the straight-line 
trajectory are implemented by the different $\phi_j$ in proportion to $f''(t_j)$, that can be thought as ``acceleration spikes'' 
that stir the particle into the correct trajectory. We now proceed to discuss how the various constraints are incorporated 
into our modeling framework.

\subsection{Incorporating full constraints}

The great virtue of the non-parametric approach adopted here is that no assumption is made about the functional form 
of the electric form factor. However, if the calibration parameters $\xi_{j}$ defined in Eq.\,(\ref{eq:model1}) are left 
unrestricted, the resulting model for $f_{\xi}(x)$ is likely to violate the physical constraints outlined in Eq.\,(\ref{GEs}). 
In the notation assumed in this section these constraints are given by: (a) $f_{\xi}(0)\!=\!1$, (b) $f'_{\xi}(x)\!<\!0$, and
(c) $f''_{\xi}(x)\!>\!0$. In this section we discuss the model formulation with all the constraints. We have shown in Appendix\,\ref{ssec:theory} that in order to satisfy these constraints the model parameters must obey the following linear relations:
\begin{align}
 &\xi_1 =1,  \label{cons1} \\
 & \xi_2 + \sum_{j=0}^N c_j \, \xi_{j+3} \le0,  \label{cons2} \\
 &\xi_{j+3} \ge 0, \ \text{for}\ j = 0, 1, \ldots, N. \label{cons3}
\end{align}\label{eq:cons_set}

where $c_j\!=\!\psi_j (1)$ is the area under the triangle formed by the function $h_j(x)$, except for the first one $c_{{}_{0}}$
and last one $c_{{}_{N}}$ which are equal to half the area of the triangle. In order to incorporate the constraints in Eq.\,\eqref{eq:cons_set} we define the following set:

\begin{align}
\m C_{\xi} \equiv \bigg\{\xi\in \mathbb{R}^{N+2}: ~ \xi_2 + \sum_{j=0}^N c_j \ \xi_{j+3}\le0, \ \xi_{j+3} \ge 0, \ j = 0, \dots, N \bigg\},
\end{align}\label{cset}

in other words, the list $\xi =\{\xi_{2},\ \xi_{3},\ldots ,\ \xi_{N+3}\}$ belongs to $\m C_{\xi}$ if the $\xi_j$ satisfies the relationships \eqref{eq:cons_set}. 

~The proton radius introduced in Eq.\,(\ref{PRadius}) is expressed directly in terms of $\xi_{2}$ as:
\begin{align*}
r_{p} = \frac{\sqrt{{-6\xi_{2}}}}{Q_{\max}}, 
\end{align*}
where $Q_{\max}$ enters to account for the rescaling of $Q^{2}$ into the dimensionless variable $x\!=\!Q^{2}/Q_{\rm max}^{2}$. 

Note that the value of $\xi_1$ is fixed at $1$ and $r_p$ only depends on the value of $\xi_2$ in the constraint set $\m C_{\xi}$. We provided a detailed discussion on a partially constrained model with the condition $\xi_1=1$ removed in Appendix.\,\ref{ssec:variant}. The rest of the discussion in the following sections obeys a fully constrained model.

\subsection{Probabilistic model for fully constrained function estimation}
The observed experimental data consists of $n$ pairs of the form $(x_{i},g_{i})$, where $x_{i}\!=\!Q_{i}^{2}/Q_{\rm max}^{2}$ 
and $g_{i}$ is equal to the form factor $G_E(Q_i^2)$ up to some experimental noise. Specifically, one assumes 
that the $n$ experimental measurements $g_{i}$ have normally distributed experimental errors $\epsilon_i$. That is, 
$g_{i}\!=\!G_E(Q^2_i)\!+\!\epsilon_i$, where we assume that each $\epsilon_i$ is a normally distributed variable with zero mean and standard 
deviation $\sigma$.

Let $Y = (y_1, \ldots, y_n)^\T$ with $y_i :\,= g_i - \xi_1 =g_i - 1$ (the subtraction of the independent term $\xi_1$ is made in order to build an homogeneous matrix equation), 
and set $\varepsilon = (\epsilon_1, \ldots, \epsilon_n)^\T$. Also, define a basis matrix $\Phi$ (a $n \times (N+2)$ matrix) with $i$th row $(x_i, \phi_0(x_i), \ldots, \phi_N(x_i))$. With these ingredients, we express our model in vectorized notation as:
\begin{align}\label{eq:model_vect}
Y = \Phi \xi + \varepsilon, \quad \varepsilon \sim \m N_n(0, \sigma^2 \mr I_n), \quad \xi \in \m C_{\xi},
\end{align}
where $\m C_{\xi}$ is defined in Eq.\,\eqref{cset}. The notation $v \sim \m N_n(\mu, \Sigma)$ means that the random variable $v$ follows a multivariate Gaussian distribution with mean $\mu$ and covariance matrix $\Sigma$.

We operate in a Bayesian framework \cite{gelman2014bayesian} and express pre-experimental uncertainty in $\xi$ through a prior distribution $P(\xi)$. The prior for $\xi$ is combined with the data likelihood $P(Y|\xi)$ to obtain the posterior distribution for $\xi$ given the observed values $Y$:

\begin{equation}
 P(\xi|Y) = \frac{P(Y|\xi) P(\xi)}{P(Y)}.
\end{equation}
This posterior distribution of the parameters $P(\xi|Y)$ can then be used to make inference on $r_p$ including point estimates and uncertainty quantification through credible intervals. Since we assume Gaussian distributed noise $\varepsilon_i$ for the observational points $y_i$, our likelihood term $P(Y|\xi)$ will be of the form $Y \sim \m N_n(\Phi \xi, \sigma^2 \mr I_n)$, which represents an exponential decay in the square of the difference between our observed data and our model prediction, usually denoted by $\chi^2$ and defined as: $\chi^2=\sum_1^n (Y_i-f_\xi(x_i))^2$. The choice of a suitable prior $P(\xi)$ is critical for a valid inference on $r_p$. It is evident from Eq.\,\eqref{m:approx} and Eq.\,\eqref{eq:model1} that a flexible representation for $f$ can be reproduced through the coefficients $\xi$ which is in turn relatable to $f$ through its derivatives. In the unconstrained setting, a natural choice of prior for $\xi$ can be induced through a Gaussian process prior on $f$. On the other hand, the prior for $\xi$ should be supported on the restricted space $\m C_\xi$ so that any prior draw obeys the constraints for $\xi$. We combine these two features to propose a flexible constrained Gaussian prior for $\xi$ and describe this procedure in the following subsection.

\subsection{Prior specification: Gaussian Process}
A Gaussian process (GP)\,\cite{rasmussen2004gaussian} is a distribution of functions on the functions space such that the collection of random 
variables obtained by evaluating the random function at a finite set of points is multivariate Gaussian. A GP is completely defined by 
a mean function $\mu(x)$ and a covariance function $K(x,x')$. Therefore, any finite collection of points 
$y_1(x_1),...y_N(x_N)$ at locations $x_1,...,x_N$ has a joint Gaussian distribution given by:
\begin{equation}
 \Big(y_1(x_1),\dots,y_N(x_N)\Big)\sim \m N (\mu,\Sigma),
 \label{eq:distr}
\end{equation}
where $\mu=\big(\mu(x_1),\ldots,\mu(x_N)\big)$ and $\Sigma_{ij}=\tau^2 K(x_i,x_j)$. Intuitively, one can think that the mean function 
represents a central value at each $x$ around which we expect our observations to be. The deviation of these observations from the 
mean function is controlled by the parameter $\tau$. In turn, the covariance function $K(x_i,x_j)$ controls the correlation between the 
observed deviations at different points $x_i$ and $x_j$. We use the notation $f|X \!\sim\!\mbox{GP}(\mu(X), \tau^2 K(X, X'))$ to denote that 
the function $f$ follows a Gaussian process with mean function $\mu$ and covariance function $\tau^2K$. As is commonly 
done\,\cite{jeffreys1946invariant} we have placed an (improper) objective prior on $\tau^2$. For a more detailed explanation 
on Gaussian Processes see \,\cite{rasmussen2004gaussian}.

The model parameters $\xi_{j}$ are related to first and second derivatives of the form factor $G_E$, or equivalently to its rescaled version 
$f$ at the various grid points $x_j$. Since Gaussian processes are closed under linear operations, such as taking 
derivatives\,\cite{rasmussen2004gaussian}, they represent an optimal choice in estimating the form factor. If 
$f\!\sim\!\mbox{GP}(0, \tau^2 K)$\footnote{The selection $\mu(x)=0$ is done to avoid centering the GP around any parametric 
form.}, then any finite number of observations $f(x_1),...,f(x_N)$ follow the distribution specified by Eq.\,(\ref{eq:distr}). Therefore, a 
collection of random variables that involves derivatives $f'(0),f''(x_0)...,f''(x_N)$ also follow a Gaussian distribution with a covariance 
matrix $\Gamma$ that involves up to four mixed partial derivatives of the covariance function $K(x,x')$; 
see Theorem 2.2.2 in\,\cite{adler1981geometry}. That is,
\begin{eqnarray}\label{eq:cov}
\Gamma = \begin{bmatrix}
 \frac{\partial^2 K}{\partial x \partial x'}(0,0)&\frac{\partial^3 K}{\partial x \partial {x'}^2}(0,x_0) & \cdots & \frac{\partial^3 K}{\partial x \partial {x'}^2}(0,x_N)\\[1.5ex]
\frac{\partial^3 K}{\partial x^2 \partial x'}(x_0,0)& \frac{\partial^4 K}{\partial x^2 \partial {x'}^2}(x_0,x_0) & \cdots &\frac{\partial^4 K}{\partial x^2 \partial {x'}^2}(x_0,x_N)\\
\vdots &\vdots &\ddots&\vdots \\
\frac{\partial^3 K}{\partial x^2 \partial x'}(x_N,0)&\frac{\partial^4 K}{\partial x^2 \partial {x'}^2}(x_N,x_0)&\cdots & \frac{\partial^4 K}{\partial x^2 \partial {x'}^2}(x_N,x_N)\\
\end{bmatrix}_{(N+2)\times(N+2)}.
\end{eqnarray}
For illustration purposes consider the first row of the matrix $\Gamma$. It specifies how the derivative of the function at zero, $\xi_2$, correlates with all the other $\xi_j$. The correlation between $\xi_2$ and the other $\xi_j$ for $j\!>\!2$ is controlled by the mixed partial third derivative of $K$ at $x_j$. 

If the model parameters $\xi_{j}$ are left unconstrained, then a natural prior, induced from a GP prior on the unknown function $f$, would be a 
finite-dimensional Gaussian prior $\xi\!\sim\!\m N_{N+2}(0, \tau^2 \, \Gamma)$ with $\Gamma$ as in Eq.\,\eqref{eq:cov}. However, since the 
various shape constraints on the function impose a corresponding set of constraints on the model parameters, we adopted a truncated Gaussian prior on $\xi$: 
\begin{align*}
p(\xi) = \frac{1}{M_\xi} \, (2\pi)^{-(N+2)/2} \, |\Gamma|^{-1/2} \, (\tau^2)^{-(N+2)/2} \, e^{- \frac{\xi^T \Gamma^{-1} \xi}{2 \tau^2}} \, \ind_{\m C_{\xi}}(\xi),
\end{align*}
where the ``indicator function'' $\ind_{\m C_{\xi}}(\xi)$ filters the $\xi_j$ such that only the allowed combinations are those that satisfy the constraints listed in Eq.\,(\ref{eq:cons_set}): $\ind_{\m C_{\xi}}(\xi)=1$ if $\xi \in C_{\xi}$, and $\ind_{\m C_{\xi}}(\xi)=0$ otherwise.
In the above expression $M_\xi$ is a constant of proportionality required to make $p(\xi)$ a density distribution, {\sl i.e.,} $p(\xi)$ must integrate to 
one. We shall denote $p(\xi)$ by $\m{N}_{N+2}(0, \tau^2 \, \Gamma)\ind_{\m C_{\xi}}(\xi)$ and refer to it as the constrained Gaussian Process (cGP) prior for $\xi$.

To fully specify the cGP prior 
we still need to define the covariance function $K(x,x')$ that determines the matrix $\Gamma$. Following common practice, 
we chose $K$ to be a stationary Mat{\'e}rn kernel with smoothness parameter $\nu\!=\!5/2$ and length-scale $\ell\!>\!0$. Such a kernel only depends on the relative distance between the 
coordinates $r\!\equiv\!|x-x'|$ and can be written in closed form as follows:
\begin{align*}
\mathrm{K}(x,x') \equiv \mathrm{k}_{\nu=5/2, \ell}(r) = \bigg(1 + \frac{\sqrt{5} \, r}{\ell} + \frac{5 r^2}{3 \ell^2} \bigg) \, \exp\bigg(- \frac{\sqrt{5} \, r}{\ell} \bigg). 
\end{align*}
In our analysis we also explored the values $\nu=3$ and $\nu=7/2$. The more general definition for the Mat{\'e}rn kernel is shown in the Appendix \ref{ssec:hyp}. The optimal value for the correlation length $\ell$ is chosen by a cross-validation scheme outlined also in the Appendix.\,\ref{ssec:hyp}.

\subsection{Posterior sampling and inference}
Given the complex nature of the model space associated with the allowed values of $\xi$, an analytic expression of $M_{\xi}$ is not available. 
However, we show in Appendix\,\ref{ssec:theory} that $M_{\xi}$ does not depend on the unknown parameter $\tau$. Hence, provided $\Gamma$ is fixed, one can exploit this fact and use a Markov Chain Monte Carlo (MCMC) algorithm to sample the posterior distribution. The model along with priors on various components are represented in a hierarchical fashion as follows:
\begin{align*}
& Y \mid \xi, \sigma^2, \tau^2 \sim \m N_n(\Phi \xi, \sigma^2 \mr I_n), \\
& \xi \sim \m N_{N+2}(\xi; 0, \tau^2 \, \Gamma) \, \ind_{\m C_{\xi}}(\xi), \quad p(\tau^2) \propto \frac{1}{\tau^2}, \quad p(\sigma^2) \propto \frac{1}{\sigma^2},
\end{align*}
in which we have made the common non-informative prior choice for the observational noise standard deviation $\sigma^2$. 
For the hierarchical model above, the joint posterior distribution of the model parameters is given by: 
\begin{align}\label{eq:posterior}
P(\xi, \tau^2, \sigma^2 \mid Y) \propto \bigg\{ (\sigma^2)^{-n/2} \, e^{-\frac{ \|Y - \Phi \xi \|^2}{2 \sigma^2}} \bigg\} \ \bigg\{ (\tau^2)^{-(N+2)/2} e^{- \xi^\T \Gamma^{-1} \xi/(2 \tau^2)} \, \ind_{\m C_{\xi}}(\xi) \bigg\} \ (\tau^2)^{-1} \, (\sigma^2)^{-1}.
\end{align} 

The final normalizing constant of the posterior distribution is intractable and hence we resort to MCMC algorithm \cite{gelman2014bayesian} to sample from the posterior distribution of the model parameters. More specifically, we use Gibbs sampling to iteratively sample from the full conditional distribution of (i) $\xi \mid \tau^2, \sigma^2, Y$ \footnote{Recall that in Bayesian notation $\xi \mid \tau^2, \sigma^2, Y$ means the posterior distribution of $\xi$ given $ \tau^2, \sigma^2,$ and $Y$.}, (ii) $\tau^2 \mid \xi, \sigma^2, Y$, and (iii) $\sigma^2 \mid \xi, \tau^2, Y$. The conditional posterior of $\xi$ in (i) is a truncated multivariate normal distribution which is sampled using the method proposed in \cite{botev2017normal}.
~The conditional posteriors of $\sigma^2$ and $\tau^2$ in (ii) and (iii) are inverse-gamma distributions (IG) and hence easy to sample from. The details of the algorithm are provided in the Appendix \ref{ssec:hyp}. 

After discarding initial burn-in samples, let $\xi_j^{(1)}, \ldots, \xi_j^{(T)}$ be $T$ successive iterate values of $\xi_j$ from the Gibbs sampling algorithm, for $j = 2,\dots, N+3$. Our point estimates for $r_p$ based on the posterior samples are:
\begin{align}\label{eq:est1}
\widehat{r}_p = T^{-1} \sum_{t=1}^T \frac{\sqrt{- 6 \xi_2^{(t)}}}{Q_{\max}}. 
\end{align}
The confidence interval of $95 \%$ for $r_p$ is also computable from our sampling algorithm. We shall denote the lower bound of this interval by CI$_\mathrm{l}$ (the 2.5 \% quantile) and the upper bound of this interval by CI$_\mathrm{u}$ (the 97.5 \% quantile).

\section{Pseudo-data analysis}\label{sec:sims}
Before analyzing the real data, we test the GP methods on synthetically generated datasets. The details of this analysis can be found in Appendix \ref{DetailsPseudo}. A general guidance on the prior and hyperparameter choices can be found in Appendix\, \ref{ssec:hyp}. Recall that we defined four variants of the method proposed to understand the role of each constraint, which can be described as follows in terms of the $\xi$:
\begin{enumerate}
\item {\bf cGP:} denotes the proposed constrained GP model as described in Eq.\,\eqref{eq:model_vect}. The curve is restricted to be convex and the value at $Q^2=0$ is fixed at $1$ ($\xi_1 =1$). 
\item {\bf c$_0$GP:} denotes the model in Eq.\,\eqref{eq:model_vect} with the only constraint being \eqref{GE0}, the value at zero ($\xi_1 =1$). The parameters $\xi_2, \ldots, \xi_{N+3}$ are left unconstrained in this model and therefore the curve is not necessarily monotonic and convex. 
\item {\bf c$_1$GP:} denotes the model with only shape constraints \eqref{GE1} and \eqref{GE2}, which implies that the function is non-increasing and convex, but the value at zero is not fixed ($\xi_1$ is left unconstrained). 
\item {\bf uGP:} denotes the completely unconstrained GP, all the parameters $\xi_1,\xi_2, \ldots, \xi_{N+3}$ are free. 
\end{enumerate}
Note that since for cGP and c$_0$GP we fix $\xi_1\!=\!1$, we use $\widehat{r}_p$ in Eq.\,\eqref{eq:est1} to estimate the proton radius, while for c$_1$GP and uGP we use $\widetilde{r}_p$ defined in Eq.\,\eqref{eq:est_c1} as our estimator. 

To mimic the real dataset, we use the $Q^2$ from the electron-proton scattering data obtained from Mainz \cite{bernauer2011high,bernauer2014electric,bernauer2011bernauer}, and generate the pseudo $G_E$ data using the ``Dipole function'' given by:
\begin{eqnarray*}
G_E(Q^2) = \bigg( 1+ \frac{r_p^2Q^2}{12}\bigg)^{-2},
\end{eqnarray*}
$r_p$ being the pseudo-radius of the proton. The ``Dipole function'' is a good proxy for the electric form factor equation and would serve as the ground truth for conducting the simulation study. In the following numerical examples we set $r_p = 0.84$ fm. We extract $n = 500$ sample points of $Q^2$ from the Mainz dataset in three regimes: i) low $Q^2 (\leq 1.36$ fm$^{-2})$, ii) medium $Q^2 (\leq 4.85$ fm$^{-2})$ and iii) high $Q^2 (\leq 25.12$ fm$^{-2})$. To generate noisy observations of $G_E$ we add independent and identically distributed zero-mean Gaussian noise with standard deviations in the set $\{0$, $0.002$, $0.005$, $0.01\}$, where $\sigma = 0$ means no noise is added to $G_E$. The interval for $\sigma$ is chosen to contain the typical observed errors in the Mainz data \cite{bernauer2011high,bernauer2014electric,bernauer2011bernauer}. The four models are used with smoothness parameter $\nu = 2.5$. Following \cite{maatouk2017gaussian}, the number of basis functions is set to $N = \{n/4, n/2, n\}$. To select the optimal \textit{length-scale} parameter $\ell$ we developed a cross validation procedure by analyzing the MSE (Mean Squared Error) as a function of $\ell$, see Appendix \ref{DetailsPseudo} for the implementation details. The selected optimal values for $l$ are $l=20$ for both low and medium regime, and $l=1$ for the high regime. Since the scale for $\ell$ is on the re-scaled variable $x$, $\ell \gg 1$ can be interpreted as an indicative that the whole range of $Q^2$ considered is highly correlated.

Tables \ref{tab:dipole_I}-\ref{tab:dipole_III} in the Appendix \ref{DetailsPseudo} show the posterior summaries of the estimates of the radius $r_p$ and $95\%$ credible intervals, the lower and upper bound denoted respectively by CI$_\mathrm{l}$ and CI$_\mathrm{u}$. Since we know the generated $G_E$ values as well as the generated radius $r_p$, we are able to evaluate the results of the GP methods with different constraints and in different regimes. We observed that in all three regimes using a smaller number of the basis functions ($N$) lead to smaller values of MSEs on the 20\% held out $Q^2$ values. 

We found that in the presence of noise, imposing all the constraints (cGP) reduces the uncertainty in the estimation while maintaining accuracy, while only imposing the constraint at zero (c0GP), tends to give accurate results but with wider credible intervals. If we only consider the shape constraints (c1GP) the estimates of the radius are somewhat variable as the noise level increases, becoming more biased for the higher $Q^2$ regimes. The unconstrained model (uGP) leads to the widest credible intervals in general and reasonably good estimates when the noise level is small. Comparing results across different regimes, we found that in the medium and high $Q^2$ regimes our methods tend to give lower estimates at the origin as the noise level increases. The trend of obtaining lower estimates of $r_p$ in higher regimes could be caused by many reasons. One could be that the model is able only to borrow information from one side when estimating over the boundary (the origin), but the model hyperparameters are selected according to the overall model fitting. We give a more detail explanation in favor of this argument in the Appendix \ref{DetailsPseudo} and we shall explore this trend in a future work.

We conclude that all the physical constraints are necessary for providing a realistic estimate of the radius. It is also evident that the low $Q^2$ regime data informs about the radius more reliably than the high $Q^2$ regime data, at least under the assumption of additive independent and identically distributed errors and the fulfillment of all the constrains \eqref{GE0}, \eqref{GE1} and \eqref{GE2} by the data. However in the real-data scenario with unknown errors, and possibly with some violation of the constraints, specially of the first one \eqref{GE0}, it might be important to consider the full dataset to take into account all sources of variation in the analysis.

\section{Electron-scattering data analysis}\label{sec:real}

In this section, we re-analyze the electron-proton scattering data obtained from Mainz \cite{bernauer2011high,bernauer2014electric,bernauer2011bernauer}. We conducted the analysis in two regimes: low $Q^2<1.36$ fm$^{-2}$ (the first 500 data points) and high $Q^2<25.12$ fm$^{-2}$ (the full data set). The low regime was chosen based on the results in the pseudo-data analysis in which we observed that in this range the models gave a more accurate estimate of the slope of the assumed Dipole function. On the other hand, even though in the high regime we observed some biasing toward lower estimates of the slope, we considered also the full data analysis. It is well known that due to the difficulty of measuring the form factor for smaller values of the momentum, the experimental data might be significantly biased for $Q^2 \approx 0$ and also the noise structure could not satisfy the assumptions we made on the pseudo data analysis: it could not be independent and identically distributed and all the points might not share the same variance. Thus, incorporating the whole range of values could help the analysis to overcome that experimental bias. Finally, having the two extremum (low and high regimes) is beneficial for comparison. 

The analysis started with conducting pilot experiments with subsets of the data of size $n = 250$ randomly selected from the range of the potential values ($Q^2$) for the high regime, and with the full 500 points in the low regime. The pilot experiments provided us with a better idea of the roles of the different hyperparameters of our model, $N, \ell$ and $\nu$, before eventually analyzing the full dataset. Recall that the $Q^2$ values are rescaled to $[0,1]$ before the analysis. Overall we used $500$ MCMC iterations after discarding a burning of $100$ samples to form the posterior summary estimates of the radius.

\begin{table}[htbp!]
 \caption{High regime posterior estimates of the radius and credible interval for cGP, c$_0$GP, c$_1$GP and uGP with $N =\{n/4, n\}$ and $\nu = \{2.5,3,3.5\}$ and $n = 1422$.} \small
 \begin{center}
  \begin{tabular}{cclcccccc} \hline
  $\bf{ \nu}$ & & {$2.5$} &{$2.5$} & {$3$} & {$3$} &{$3.5$} & {$3.5$} \\ 

   \text{$N$} & & n/4 & n  & n/4 & n & n/4 & n \\ \hline

  cGP  & $r_p$ &0.8435 & 0.8452 & 0.8413 & 0.8431 & 0.8425 & 0.8408 \\
   & CI$_\mathrm{l}$ &0.8396 &0.8426 &0.8265 & 0.8406 &0.8301 &0.8383 \\
   & CI$_\mathrm{u}$ &0.8481 & 0.8476 &0.8524& 0.8457 & 0.8511 & 0.8435 \\ \hline
 
   c$_0$GP & $r_p$ &0.8355 & 0.8448 &0.8226 &0.8431 & 0.8319 &0.8406 \\
  & CI$_\mathrm{l}$ &0.8254 &0.8373 &0.8045 & 0.8328 &0.8167  &0.8305  \\
   & CI$_\mathrm{u}$ &0.8467 &0.8519 &0.8415 & 0.8527 &0.8467 &0.8497  \\ \hline

   c$_1$GP & $r_p$ &0.8423 & 0.8311& 0.8295 & 0.8259 &0.8347 & 0.8225 \\
  & CI$_\mathrm{l}$ &0.8346 & 0.8266 & 0.7993 &0.8217 & 0.8111& 0.8190  \\
   & CI$_\mathrm{u}$ &0.8507 & 0.8369 & 0.8447 & 0.8303 & 0.8461 & 0.8266 \\ \hline
 
 uGP  & $r_p$ & 0.8474 & 0.8577 & 0.7665 & 0.8563 & 0.8253 & 0.8505 \\ 
 & CI$_\mathrm{l}$ &0.8256 & 0.8419 & 0.7374 &0.8398 & 0.7969 & 0.8338 \\ 
 & CI$_\mathrm{u}$ &0.8683 & 0.8742 & 0.7938 & 0.8744 & 0.8530 & 0.8680 \\ \hline 

\end{tabular}%
\end{center}
\label{tab:full}%
\end{table}%

\begin{table}[htbp]
 \caption{Low regime posterior estimates of the radius and credible interval for cGP, c$_0$GP, c$_1$GP and uGP with $N =\{n/4, n\}$ and $\nu = \{2.5,3,3.5\}$ and $n = 500$.} \small
 \begin{center}
  \begin{tabular}{cclcccccc} \hline
    $\bf{ \nu}$ & & {$2.5$} &{$2.5$} & {$3$} & {$3$} &{$3.5$} & {$3.5$} \\ 
   \text{$N$} & & n/4 & n  & n/4 & n & n/4 & n \\ \hline
  cGP & $r_p$ & 0.8529 & 0.8550 & 0.8543 & 0.8561 & 0.8550 & 0.8570 \\
  & CI$_\mathrm{l}$& 0.8488 & 0.8514 & 0.8503 & 0.8529 & 0.8511 & 0.8539 \\
  & CI$_\mathrm{u}$ & 0.8576 & 0.8587 & 0.8591 & 0.8593 & 0.8597 & 0.8601 \\ \hline
  c$_0$GP & $r_p$ & 0.8399 & 0.8408 & 0.8411 & 0.8432 & 0.8399 & 0.8458 \\
  & CI$_\mathrm{l}$ & 0.8213 & 0.8269 & 0.8143 & 0.8309 & 0.8168 & 0.8346 \\
  & CI$_\mathrm{u}$ & 0.8533 & 0.8516 & 0.8598 & 0.8547 & 0.8584 & 0.8556 \\ \hline
   c$_1$GP & $r_p$ & 0.8725 & 0.8719 & 0.8721 & 0.8731 & 0.8735 & 0.8739 \\
   & CI$_\mathrm{l}$ & 0.8613 & 0.8626 & 0.8628 & 0.8660 & 0.8640 & 0.8664 \\
   & CI$_\mathrm{u}$ & 0.8857 & 0.8815 & 0.8820 & 0.8799 & 0.8836 & 0.8820 \\ \hline
  uGP  & $r_p$ & 0.8573 & 0.8618 & 0.8612 & 0.8667 & 0.8593 & 0.8654 \\
  & CI$_\mathrm{l}$ & 0.8212 & 0.8424 & 0.8321 & 0.8467 & 0.8249 & 0.8449 \\
  & CI$_\mathrm{l}$ & 0.8898 & 0.8830 & 0.8897 & 0.8851 & 0.8899 & 0.8845 \\ \hline
\end{tabular}%
\end{center}
\label{tab:low}%
\end{table}%

\begin{figure}[htbp!]
\begin{center}
   \includegraphics[scale = 0.7]{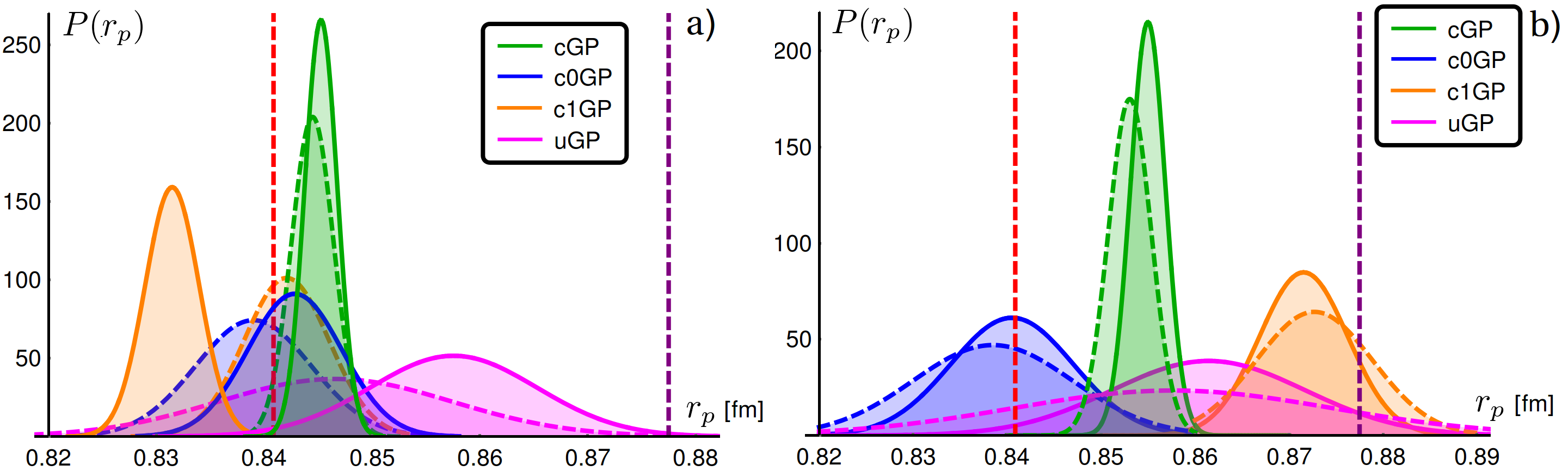}
 \end{center}

\caption{Estimated density plots of MCMC samples of radius $r_p$ for cGP, c$_0$GP, c$_1$GP and uGP with $N =n/4$ (dotted line), $n$ (solid line) and $\nu = 2.5$ for the high $Q^2$ regime (a) and for the low $Q^2$ regime (b). The vertical dashed lines stand for the muonic result of 0.84 fm (Red) and the recommended CODATA value of 0.88 fm (Purple). }
\label{fig:hist_all}
\end{figure}

Similar to the pseudo data analysis, we conducted a cross validation procedure to select the optimal scale-length parameter $\ell$ for each regime, the details of which are shown in the Appendix~\ref{DetailsElec}. Our analysis guided us to choose $\ell_\mathrm{opt}=0.5$ for the full data set and $\ell_\mathrm{opt}=10$ on the low $Q^2$ set. 

Having chosen the correlation length we performed the MCMC iterations for the four models, selecting the number of grid points $N=n/4$ and $N=n$ in order to compare results. Tables \ref{tab:full} and \ref{tab:low} show the posterior medians of $r_p$ of the four models and the $95\%$ credible intervals in the high and low regime respectively. Fig.\,\ref{fig:hist_all} shows the density plots (posterior distribution $P(r_p)$) for all discussed GP models with $\nu=2.5$ and both $N = n/4$ and $N=n$ in the high regime (a) and low regime (b). The detailed histograms for each model in both regimes are shown in the Appendix (Fig.\,\ref{fig:highhist_n4}, \ref{fig:highhist_n} for high regime and Fig.\,\ref{fig:lowhist_n4}, \ref{fig:lowhist_n} for low regime). The function fits for the different models are shown in Fig. \ref{fig:fit_high} and Fig. \ref{fig:fit_low}.

For the high regime we see that the estimates became more sensitive to the choice of the hyperparameters $\nu$ and $N$ as the constrains were removed: cGP estimations of the radius are in all cases around $0.843$ fm, while on the other extremum the unconstrained model uGP estimations range between $0.76$ and $0.85$ fm. Incorporating constraints also strongly affects the credible intervals of each model: cGP credible intervals are between 0.005 and 0.02 fm wide, while uGP intervals can be as wide as 0.06 fm. In respect to the influence of $N$, it seems that for all the models a lower number of grid points produces a lower estimate of the radius, with the exception of c1GP in which $N$ has a reversed effect. Also, the credible intervals tend to get wider for all models when $N$ decreases. Finally, the influence of $\nu$ does not seem to have a clear tendency on the estimation, but its effect get suppressed when the constraints are present. 

For the low regime we can see that overall all the models seem to be more stable for changes in the parameters $\nu$ and $N$ when comparing with the high regime. Also, as expected from the pseudo data analysis, overall all models gave a larger estimates of the radius than the estimates obtained in the full data case, being c1GP the model with the biggest increase. c$_1$GP gave estimates for $r_p$ around $0.87$ fm and the credible intervals included $0.88$ fm in the low regime, a huge difference in comparison with its performance in the high regime in which its estimates were consistently below $0.842$ fm. uGP models gave also slightly larger estimated $r_p$ around $0.86$ fm, however the credible intervals were wide and include both $0.84$ fm and $0.88$ fm in most of the cases. In contrast, it seems c$_0$GP was not affected too much by the change of range in $Q^2$. Among all models, we can see that cGP and c$_0$GP are the most robust to the range of $Q^2$ used.

\begin{figure}[htbp!]
\begin{center}
  \includegraphics[scale = 1]{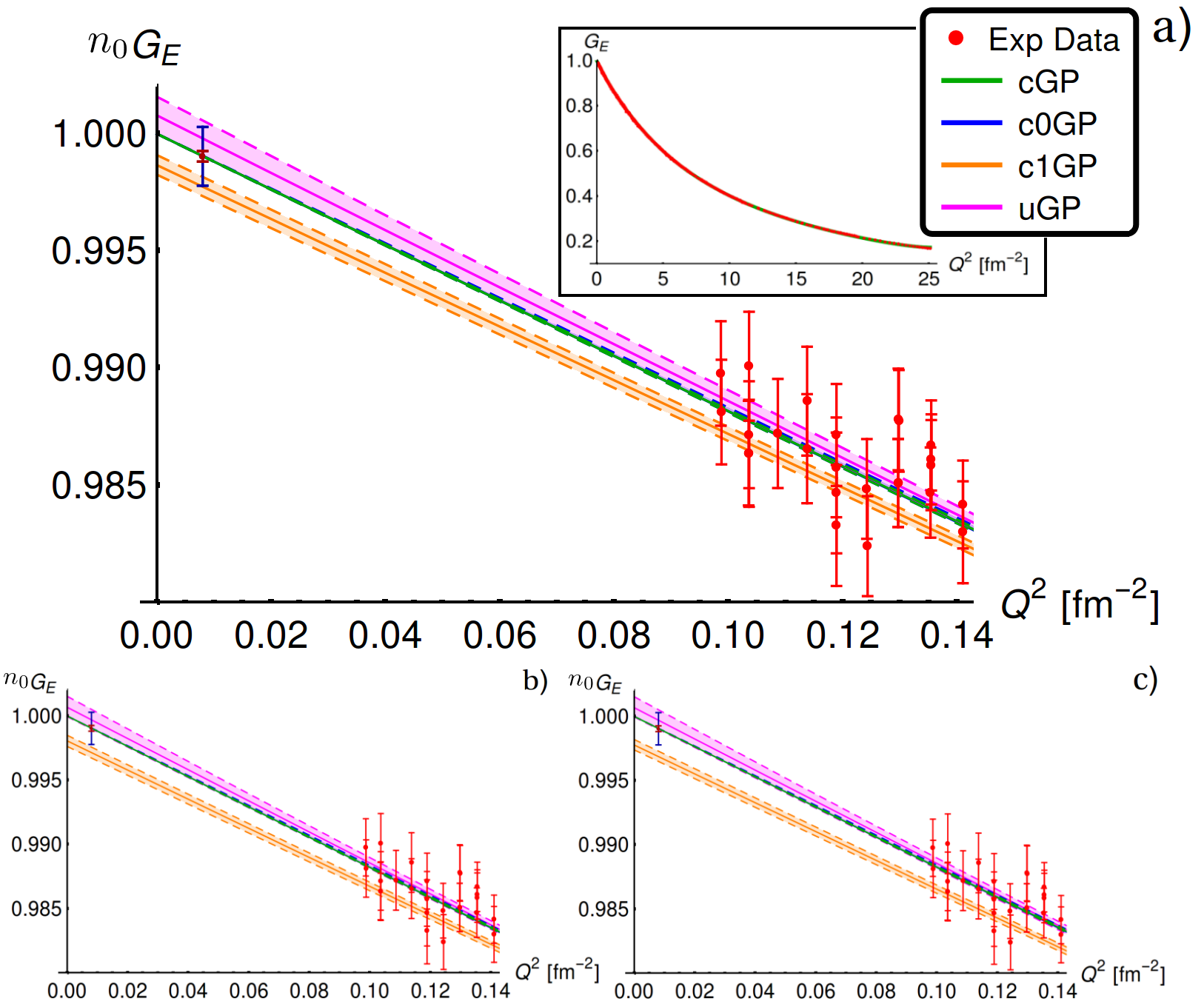}

  \end{center}
\caption{Function fit with $\nu = 2.5$ (a), $3$ (b), $3.5$ (c) and $N=n$ in the high regime. The inset plot in (a) shows the overall fit of the models for $\nu=2.5$ to the entire data range. The solid curves denote the model predictions while the shaded intervals bounded by dotted lines represent the $95\%$ confidence intervals for the predictions. The red dots denote the experimental data obtained from Mainz with its respective error bars. The red and blue points near the origin at $Q^2=0.008 fm^{-2}$ represent the lower value the new PRad experiment will be able to measure, with two different estimates for the projected uncertainty \cite{Gasparian:2017cgp} and arbitrary $G_E(Q^2)$ value.} 
\label{fig:fit_high}
\end{figure}

\begin{figure}[htbp!]
\begin{center}
  \includegraphics[scale = 1]{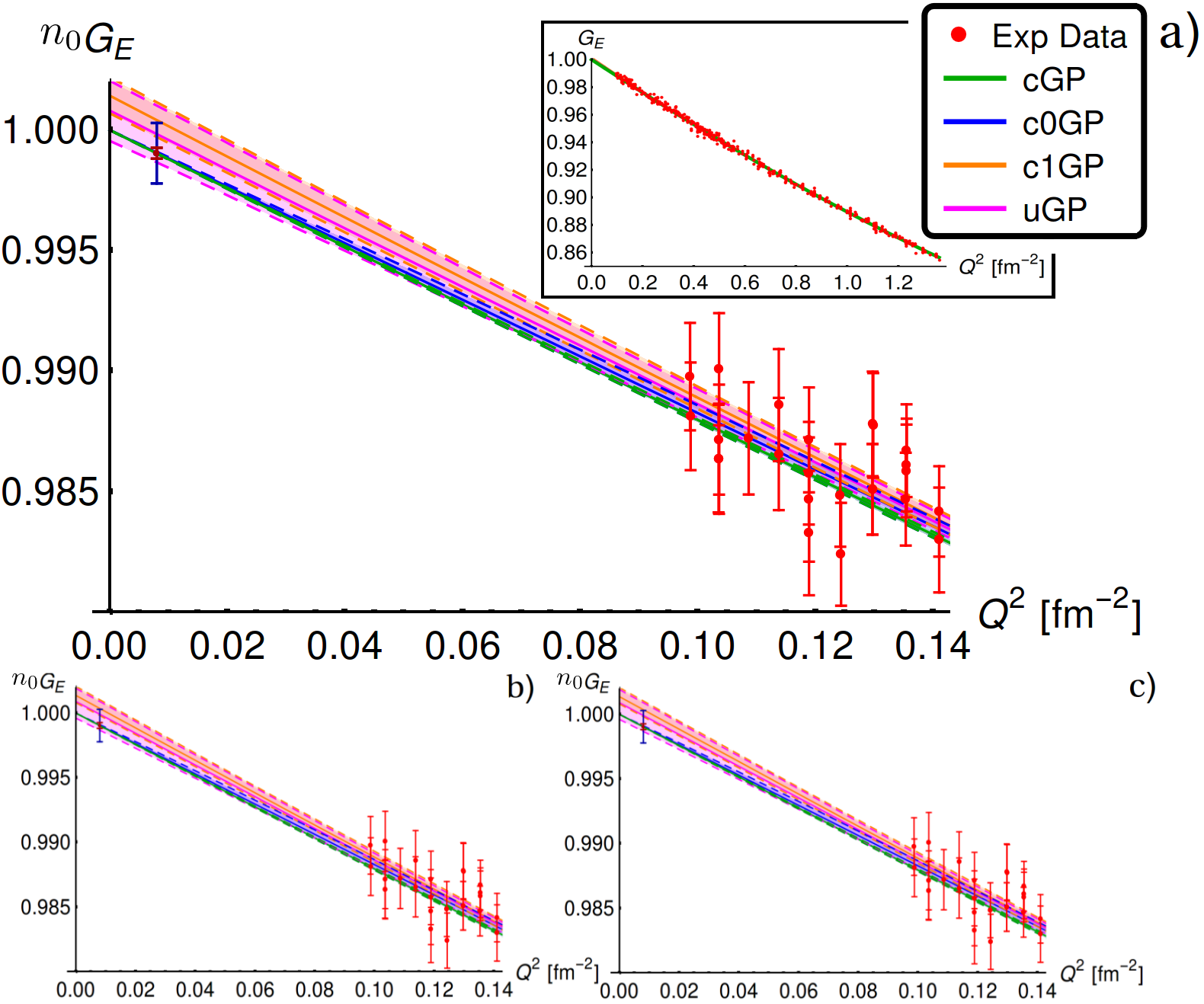}
  \end{center}
\caption{Function fit with $\nu = 2.5$ (a), $3$ (b) $3.5$ (c) and $N=n$ in the low regime. The inset plot of (a) shows the overall fit of the models for $\nu=2.5$ to the entire data range. The solid curves denote the model predictions while the shaded intervals bounded by dotted lines represent the $95\%$ confidence intervals for the predictions. The red dots denote the experimental data obtained from Mainz with its respective error bars. The red and blue points near the origin at $Q^2=0.008 fm^{-2}$ represent the lower value the new PRad experiment will be able to measure, with two different estimates for the projected uncertainty \cite{Gasparian:2017cgp} and arbitrary $G_E(Q^2)$ value.} 
\label{fig:fit_low}
\end{figure}

For the particular choice $\nu=2.5$ we show in Fig.\,\ref{fig:hist_all} the final posterior distribution of $r_p$ of all the models on both regimes and both choices of $N$, and we denote by $P(r_p)$ the posterior density function of $r_p$. As a general trend we can see that as the number of grid points $N$ increases the estimate of c1GP moves to a lower value of $r_p$ while the estimates of all the other models increase to a higher value of $r_p$. This effect is less prominent in the low regime and overall cGP is the most robust with respect to changing the number of grid points. In all the cases, as $N$ increases the variability in the estimation reduces (the estimated $\sigma$ is slightly smaller than those for $N = n/4$), giving more precise results. As we observed in Tables \ref{tab:full} and \ref{tab:low}, in going from high regime to low regime all the models, with the exception of c0GP, gave a larger estimate of the radius, being c1GP the one that showed the biggest change. uGP is the only model that includes both $0.84$ and $0.88$ fm in its support in both regimes. 

Fig.\,\ref{fig:highhist_n4} and \ref{fig:highhist_n} (high regime), and Fig.\,\ref{fig:lowhist_n4} and \ref{fig:lowhist_n} (low regime) in the Appendix \,\ref{DetailsElec} show more in detail each individual posterior histogram of the MCMC samples from GP models for both $N = n/4$ and $N =n$. Fig.\,\ref{fig:highhist_f0} (high regime) and \ref{fig:lowhist_f0} (low regime) in the Appendix \,\ref{DetailsElec} show the MCMC samples of $n_0G_E(0)$ from c$_1$GP and uGP in the cases where $N = n/4$ and $N=n$. Recall that $n_0$ is defined as a floating normalization factor, while $G_E(0)$ is a guaranteed property by the definition of $G_E$. The sample centers of $n_0G_E(0)$ deviate from 1 by a very small amount ($|n_0G_E(0)-1| \lesssim 0.0014$) for both models in both regimes. It is remarkable how such a small deviation in the case of c1GP can make such drastic changes when $r_p$ results are compared with the fully constrained model cGP. For example, in the low regime for $N=n/4$ cGP estimates $r_p=0.853$ fm while c1GP, having a value of $1.0014$ at zero, estimates $r_p=0.873$, a result that highlights the impact that a floating normalization can have on the extraction of the radius.

Fig.\,\ref{fig:fit_high} and \ref{fig:fit_low} show the function fits for the high and low regime respectively, with $\nu\!=\!2.5$, $3$, $3.5$ for $N=n$. The overall fit is good for all the methods in both regimes, the real differences appear as $Q^2\rightarrow0$. For this reason we show the full fit in each regime only for $\nu=2.5$ in the inset of the respective top plot, being the full fits for the other values of $\nu$ visually indistinguishable. 

Overall we found relatively small variability in the function fits across different values of $\nu$ in both regimes, not enough to change the estimation of the radius by more than $0.01$ fm within any of the models. Due to the constraint at the origin, both posterior medians of cGP and c0GP agree as $Q^2\rightarrow0$ with very narrow credible intervals, while c1GP and uGP are either below or above and start going close to the other GP models estimates as $Q^2$ grows. As expected, the shape constraints help reduce the variability of the models, which is evidenced by the smaller credible intervals of c1GP in comparison with uGP, specially in the low regime. In the low regime, it seems that without the location restriction the extrapolations are likely to attain values at $Q^2=0$ larger than $1$, which in turn pushes the estimate of the radius to larger values, as can be also seen in Fig.\,\ref{fig:hist_all}. In the low regime, as a general trend, we observed wider credible intervals for all the models.

The blue and red points near $Q^2=0.008 fm^{-2}$ displayed in Fig.\,\ref{fig:fit_high} and \ref{fig:fit_low} for an arbitrary $G_E(Q^2)$ value represent the lowest momentum that will be measured by the new PRad experiment \cite{Gasparian:2017cgp}. The blue and red error bars are two different estimates of the projected uncertainty the measurement will have. In the case of our proposed model, it seems that the blue uncertainty could allow us to discard either c1GP or uGP, while the red uncertainty would allow us to discard up to three of the model selected, clearly imposing a defined constraint in the final estimation of the radius.

As we have shown by our analysis the extrapolation near the boundary can be subtle and highly subjective to the data. Obviously the constraint at $Q^2=0$ can reduce the influence from the data range and also the model hyperparameters, however the question is how much we can trust on the constraints, and if without the constraints how much we can trust on the estimation procedure near $Q^2=0$, we leave this issue for the future work.

\vfill\eject

\section{Conclusions}
\label{Sec:Conclusions}

The charge radius of the proton is a fundamental parameter that has attracted enormous attention during the last decade because of a discrepancy
between two experimental methods. The value of the charge radius of the proton $\rp\!=\!0.84087(39)\,{\rm fm}$ determined from muonic 
hydrogen\,\cite{Pohl:2010zza,Pohl:2013yb} differs significantly (by $\sim$4\% or nearly 7$\sigma$) from the recommended CODATA value of 
$\rp\!=\!0.8775(51)\,{\rm fm}$ obtained from decades of experiments in electron scattering and atomic spectroscopy. Many possible solutions to the 
``proton puzzle'' have been proposed ranging from errors in the experimental data or in its interpretation all the way to new physics associated to a 
violation in lepton universality. There is even a recent publication that questions whether muonic hydrogen and electron scattering experiments measure 
the same observable\,\cite{Donnelly:2018uft}. Within this wide context our contribution is rather modest. In our view, the proton puzzle lays not in the 
experimental data, but rather in the extraction of the proton radius from the scattering data. To extract the charge radius from the electron scattering 
data set, one must extrapolate from the measured values of the electric form factor at a finite momentum transfer $Q^{2}$ all the way to $Q^2=0$. 
How to properly extrapolate to $Q^2\!=\!0$ has been the source of much controversy and innumerable debates. Many of these debates center around 
the optimal functional form ({\sl e.g.,} monopole, dipole, polynomial, Pad\'e, {\sl etc.}) that should be adopted to carry out the extrapolation and on how
best to determine the parameters associated to such functions. In this paper we also seek for an optimal extraction of the proton radius from the
scattering data. However, in contrast to most of these approaches and in an effort to eliminate any reliance on specific functional forms, we have 
introduced a non-parametric method that does not assume any particular functional form for the form factor. Rather, we adopt a method that is 
flexible enough to ``let the data speak for itself" and that solely relies on two physical constraints inherent to the form factor: (a) $G_E(Q^2\!=\!0)\!=\!1$ 
and (b) $G_E(Q^2)$ is a monotonically decreasing function of the momentum transfer. Note that this last constraint implies that $G'_E(Q^2)\!<\!0$ and 
$G''_E(Q^2)\!>\!0$ for all values of $Q^{2}$. These shape constraints are adopted in our study and their individual effects on the estimation of $r_p$ are explored.

The modeled form factor was expanded in terms of a suitable set of basis functions with coefficients restricted exclusively by the shape constraints.
To determine the optimal coefficients, the experimental data was divided into two $Q^{2}$ regions: (i) low $Q^2\!\leq\!1.36$ fm$^{-2}$ and (ii) high $Q^2\!\leq\!25.12$ fm$^{-2}$. For each of these regions, the optimal hyperparameters --the correlation 
length $\ell$, the smoothness parameter $\nu$, and the number of grid points $N$-- were obtained by monitoring the performance of the algorithm 
against the 20\% of the data that was left out from the calibration. The actual implementation of the algorithm was carried out via MCMC sampling of 
the posterior distribution using Bayesian inference.

To test the robustness and reliability of the approach we started by confronting our results against 
(known) synthetically-generated data with random Gaussian errors in low, medium and high regime. For the case in which both shape constraints were incorporated (labeled in the 
main text as cGP) we obtained an accurate and precise determination of the proton radius in both the low and medium $Q^2$ regions. In the high 
$Q^2$ region where the entire synthetic data set was used, we observed a systematic shift towards lower values of the (known) radius. We believe
that this problem may be associated to the method chosen to determine the hyperparameters. We plan to devote more attention to this matter in a 
future work.

In the case of the real experimental data from Mainz, we also found that the extraction of the proton radius is sensitive to the range of $Q^2$ 
values considered in the analysis. In the case of the high $Q^{2}$ region where the entire experimental data set is incorporated, the CODATA value 
of $r_{p}\!=\!0.878\,{\rm fm}$ is disfavored regardless of the adopted constraints. If both constraints are incorporated (cGP) we extract a charge 
radius of $r_{p}\!=\!0.8452^{+0.0024}_{-0.0026}\,{\rm fm}$. The value is even lower if we assume a floating normalization (c$_{1}$GP):
$r_{p}\!=\!0.8311^{+0.0058}_{-0.0045}\,{\rm fm}$. We note that we also considered a scenario of largely academic interest in which no constraints 
were incorporated. As expected, the unconstrained model (uGP) returned posterior distributions that were wide enough to be consistent with both
the muonic hydrogen and CODATA values. We conclude that if the entire Mainz data set is included, our analysis favors the smaller value of the
proton radius, as suggested by the muonic Lamb shift. 

However, if the low $Q^2$ region is used to inform the posterior distribution, we obtained mixed results. First, when both shape constraints are
included, we obtain a proton radius of $r_{p}\!=\!0.8550^{+0.0037}_{-0.0036}\,{\rm fm}$---that falls almost in the middle of the two experimental 
values. If now one of the constraints is removed the behavior is radically different. Removing the normalization constraint in favor 
of a floating normalization (c$_{1}$GP) shifts the posterior distribution to a large enough value of $r_{p}$ to make it consistent with the CODATA
estimate. Note that the value at zero of c1GP is $1.0014$, not far away from $1$, and yet that is enough to produce a radius $0.02$ fm bigger than the fully constrained model cGP. In contrast, leaving the normalization fixed at $G_E(Q^2\!=\!0)\!=\!1$ but relaxing the demand for $G_E(Q^2)$ to be a monotonically 
decreasing function of $Q^{2}$ results in a value for $r_{p}$ consistent with muonic result. In this regard, we anticipate that the PRad analysis 
will play a critical role in helping resolve this ambiguity. However, based solely on the present analysis focused on the low $Q^2$ region (where 
the behavior of the form factor is nearly linear) our results are inconclusive as far as resolving the proton puzzle.

 In the future, we propose to improve our model in order to overcome a possible bias in the analysis of the high $Q^2$ region, an objective that 
 could be accomplished by developing a better procedure for estimating the hyperparameters. As this technique is still in development, we would 
 like to test it on more synthetic data sets, similar in spirit to the framework developed by Yan et al \cite{Yan:2018bez}. We trust that
 lessons learned from their project will help us improve the robustness of our non-parametric model.

 Yet, even if the resolution of the proton
 puzzle is found elsewhere, the advances along this direction would have not been in vain. The proton puzzle as well as many other 
 developments have allowed us to realize the importance of {\sl enhancing the interaction between nuclear experiment and theory through 
 information and statistics}\,\cite{Ireland:2015}. We are entering into a new era in which statistical insights will become essential and uncertainty
 quantification will be demanded. 
 
\begin{acknowledgments}
 We are enormously grateful to Prof. Douglas Higinbotham for his unconditional help, guidance and lightning fast email responses. 
 This material is based upon work supported by the U.S. Department
 of Energy Office of Science, Office of Nuclear Physics Awards
 Number DE-FG02-92ER40750. Dr. Bhattacharya acknowledges NSF CAREER (DMS 1653404), NSF DMS 1613156 and National Cancer Institute's R01 CA 158113, and Dr. Pati acknowledges NSF DMS 1613156 for supporting this research.
\end{acknowledgments}

\newpage
\appendix

\section{Appendix}\label{App}

\subsection{Theoretical guarantees for the constraints on $f_\xi$}\label{ssec:theory}
  Denote by $\m C_f$ the function subspace of all the $f_{\xi}$ defined in Eq \eqref{eq:model1} that obey the constraints \eqref{GEs}. We show below that the constraints that define $\m C_f$ can be \textit{equivalently} represented as linear restrictions on $\xi$. We state Proposition \ref{prop:const} which provides an explicit characterization of the stated linear constraints. 
\begin{proposition}\label{prop:const}
$f_{\xi} \in \m C_f $ \textit{if and only if} $\xi \in \m C_{\xi}$, recall $\m C_{\xi}$ is defined in Eq.\,\eqref{cset}.

\end{proposition}

\begin{proof} 
We first check the convexity constraint, by taking second order derivative we have $f_{\xi}''(x) = \sum_{j=0}^N \xi_{j+3}h_j(x)$, by the non-negativity of $h_j$ for all $x \in [0,1]$ and any $j = 0,\dots, N$, the set $\{f_{\xi}''(x)\ge 0, \forall x\in[0,1]\}$ is equivalent to $\{\xi_{j+3} \ge 0, j =0, \dots, N\}$. 
 To impose the non-increasing constraint, we need to check the following:
\begin{eqnarray*}
f_{\xi}'(x) = \xi_2 + \sum_{j=0}^N \xi_{j+3} \psi_j (x) \le 0, \forall x\in[0,1]. 
\end{eqnarray*}
Observe that this is equivalent to 
\begin{eqnarray}\label{eq:app1}
 \xi_2 \le - \max_{x\in[0,1]}\bigg(\sum_{j=0}^N \xi_{j+3} \psi_j (x)\bigg) = -\sum_{j=0}^N c_j \xi_{j+3}.
\end{eqnarray}
\eqref{eq:app1} follows since $\psi_j$ defined in (\ref{eq:psi}) is a non-decreasing function of $x$ and $\max_{x\in[0,1]} \psi_j(x) = \psi_j(1) =: c_j $ for $j=0,\dots, N$. This concludes the proof of the proposition. \qed 
\end{proof} 
In Proposition \ref{prop:normalizing}, we provide a detailed discussion on why the normalizing constant $M_{\xi}$ of the truncated prior distribution of $\xi$ is independent of $\tau$. 

\begin{proposition}\label{prop:normalizing}
The normalizing constant $M_{\xi}$ associated with the truncated prior distribution of $\xi$ is a constant in $[0,1]$ that does not depend on $\tau^2$. 
\end{proposition}
\begin{proof}
By definition
\begin{eqnarray*}
M_{\xi} = \int_{C_{\xi}} (\tau^2)^{-(N+2)/2}(|\Gamma|)^{(-1/2)} e^{-\frac{1}{2\tau^2}{\xi}^T \Gamma^{-1} \xi} d\xi.
\end{eqnarray*}
By change of variable ${\xi}' =\xi/\tau$, observe that the truncated region $C_{{\xi}'}$ is the same as $C_{\xi}$ as long as $\tau >0$. Hence, $M_{\xi} \in [0,1]$ does not depend on $\tau$. 

\end{proof}

\newpage
\subsection{Details on the constrained model without the constraint $\xi_1 = 1$ (c$_1$GP)}\label{ssec:variant}
As mentioned before, in order to account for a possible systematic error in the experimental data one could consider adding an unknown multiplicative parameter $n_0$ to $G_E$, a normalization constant. Assuming $f(x) = n_0 G_E(x)$ and expanding as in Eq. \eqref{eq:model1}, we get:
\begin{align}
n_0\,G_E(x) &\approx  n_0\,G_E(0)+ x\, n_0\,G'_E(0)  + \sum_{j=0}^{N} n_0\,G''_E(x_{j+3})\,\phi_j(x),\\
&  \tilde{\xi}_1 + \tilde{\xi}_2\, x + \sum_{j=0}^N \tilde{\xi}_{j+3}\phi_j (x). \label{eq:model2}
\end{align}
With the assumption $G_E (0)=1$, $\tilde{\xi}_1$ can capture all the information about $n_0$. Consider the constraint set
\begin{align}\label{cset2}
\m C_{\tilde{\xi}} \equiv \bigg\{\tilde{\xi} \in \mathbb{R}^{N+3}: ~ \tilde{\xi}_1 \in \mathbb{R}, \ \tilde{\xi}_2 + \sum_{j=0}^N c_j \ \tilde{\xi}_{j+3}\le0, \ \tilde{\xi}_{j+3} \ge 0, \ j = 0, \dots, N \bigg\}
\end{align}
where $\tilde{\xi} = \{\tilde{\xi}_j,j=1,\dots, N+3\}$. Then the proton radius introduced in Eq.\,(\ref{PRadius}) is expressed in terms of both $\xi_{1}$ and $\xi_{2}$ as:
\begin{align*} 
 r_{p} = \frac{1}{Q_{\max}} \sqrt{-6\frac{\tilde{\xi}_{2}}{\tilde{\xi}_{1}}}. 
\end{align*}
By dividing by $\xi_1$ we are able to take out the effect on the radius estimation from the floating systematic error term. Following the same line as in section \ref{Sec:Formalism}, now we discuss the partially constrained model that only incorporates constraints (b) and (c) in Eq.\,\eqref{eq:cons_set} (refer to c$_1$GP model). Let $\tilde{Y} = (\tilde{y}_1, \ldots, \tilde{y}_n)^\T$ with $\tilde{y}_i \equiv g_i$, and define the corresponding basis matrix $\tilde{\Phi}$ (a $n \times (N+3)$ matrix) with $i$th row $(1, x_i, \phi_0(x_i), \ldots, \phi_N(x_i))$. Similar to the model in Eq.\,\eqref{eq:model_vect}, now we have:
\begin{align}\label{eq:model_vect2}
\tilde{Y} =  \tilde{\Phi} \tilde{\xi} + \varepsilon, \quad \varepsilon \sim \m N_n(0, \sigma^2 \mr I_n), \quad  \tilde{\xi} \in \m C_{\tilde{\xi}}.
\end{align}
Again, the random variables $f(0), f'(0),f''(x_0)...,f''(x_N)$ follow a Gaussian distribution, with the following covariance matrix:

\begin{eqnarray}\label{eq:cov2}
\widetilde{\Gamma} = \begin{bmatrix}
K(0,0) & \frac{\partial K}{\partial x'}(0,0) & \frac{\partial^2 K}{\partial {x'}^2}(0,x_0) & \cdots & \frac{\partial^2 K}{\partial {x'}^2}(0,x_N)\\[1.5ex]
\frac{\partial K}{\partial x}(0,0) & \frac{\partial^2 K}{\partial x \partial x'}(0,0)&\frac{\partial^3 K}{\partial x \partial {x'}^2}(0,x_0) & \cdots & \frac{\partial^3 K}{\partial x \partial {x'}^2}(0,x_N)\\[1.5ex]
\frac{\partial^2 K} {\partial x^2} (x_0, 0)&\frac{\partial^3 K}{\partial x^2 \partial x'}(x_0,0)& \frac{\partial^4 K}{\partial x^2 \partial {x'}^2}(x_0,x_0) & \cdots &\frac{\partial^4 K}{\partial x^2 \partial {x'}^2}(x_0,x_N)\\
\vdots & \vdots &\vdots &\ddots&\vdots \\
\frac{\partial^2 K}{\partial x^2} (x_N, 0) & \frac{\partial^3 K}{\partial x^2 \partial x'}(x_N,0)&\frac{\partial^4 K}{\partial x^2 \partial {x'}^2}(x_N,x_0)&\cdots & \frac{\partial^4 K}{\partial x^2 \partial {x'}^2}(x_N,x_N)\\
\end{bmatrix}_{(N+3)\times(N+3)}.
\end{eqnarray}

Similar to Eq.\,\eqref{eq:posterior}, the joint posterior distribution of the model parameters with partial constraints is: 
\begin{align}\label{eq:posterior_c1}
P(\tilde{\xi}, \tau^2, \sigma^2 \mid \tilde{Y}) \propto \bigg\{ (\sigma^2)^{-n/2} \, e^{-\frac{ \|\tilde{Y} - \tilde{\Phi} \tilde{\xi} \|^2}{2 \sigma^2}} \bigg\} \ \bigg\{ (\tau^2)^{-(N+3)/2} e^{- \tilde{\xi}^\T {\widetilde{\Gamma}}^{-1} \tilde{\xi}/(2 \tau^2)} \, \ind_{\m C_{\tilde{\xi}}}(\tilde{\xi}) \bigg\} \ (\tau^2)^{-1} \, (\sigma^2)^{-1}.
\end{align} 
Therefore the estimation of the proton radius based on the posterior samples of ${\tilde{\xi}}^{(t)}_1$ and ${\tilde{\xi}}^{(t)}_2$, with $t =1,\dots,T$ is: 
\begin{align}\label{eq:est_c1}
\widetilde{r}_p = T^{-1} \sum_{t=1}^T \frac{\sqrt{- 6 {\tilde{\xi}}_2^{(t)}/{\tilde{\xi}}_1^{(t)}}}{Q_{\max}}. 
\end{align}
Note that Proposition\,\ref{prop:const} and \ref{prop:normalizing} in Appendix\,\ref{ssec:theory} still hold for the c$_1$GP model. To see how Proposition \ref{prop:normalizing} holds, the normalizing constant is:
\begin{eqnarray*}
M_{\tilde{\xi}} = \int_{C_{\tilde{\xi}}} (\tau^2)^{-(N+3)/2}(|\widetilde{\Gamma}|)^{(-1/2)} e^{-\frac{1}{2\tau^2}{\tilde{\xi}}^T {\widetilde{\Gamma}}^{-1} \tilde{\xi}} d\tilde{\xi},
\end{eqnarray*}
and since ${\tilde{\xi}}_1 \in \mathbb{R}$, by the change of variable $\tilde{\xi}' = \tilde{\xi}/\tau$, it is easy to see $\m C_{\tilde{\xi}'} = \m C_{\tilde{\xi}}$, thus the integration does not depend on $\tau$ as well.

\newpage
\subsection{Details on the choices of priors and hyperparameters, and on the Gibbs sampling steps} \label{ssec:hyp}

\noindent {\bf \underline{Choice of $\nu$:}}
Assuming $f$ to be smooth in addition to being convex, the minimum possible smoothness required is twice differentiability. In an unconstrained Gaussian process regression, Corollary 3.1 and 3.2 of \cite{yang2017frequentist} show that the point-wise posterior credible intervals contain the true function with at least the nominal coverage probability provided that the prior smoothness is set to be less than or equal to the smoothness of the underlying function. We conjecture that this will continue to hold in the case of function estimation using a constrained Gaussian process, which motivated the following choice of $\nu$. It is well-known that the reproducing kernel Hilbert space of Gaussian process endowed with Mat\'{e}rn covariance kernel with smoothness $\nu$ consists of H\"{o}lder class of smoothness $\nu + 0.5$. Hence, the choice of $\nu = 2.5$ (corresponding to twice-differentiable functions) ensures that the posterior credible intervals will not underestimate the uncertainty in estimating $r_p$. The choices of $\nu = 3$ and $\nu = 3.5$ were made in order to assets the impact of this hyperparameter on the estimation of $r_p$. The following is the most general definition of the Mat\'{e}rn covariance kernel:

\begin{equation}
 k_\nu(r) \equiv \frac{2^{1-\nu}}{\Gamma(\nu)}\bigg(\frac{\sqrt{2\nu}r}{\ell} \bigg) ^\nu K_\nu \bigg( \frac{\sqrt{2\nu}r}{\ell} \bigg),
\end{equation}

where $K_\nu$ is the modified Bessel function of the second kind.\\
{\bf \underline{Choice of $\tau$:}} $\tau$ controls the prior signal to noise ratio. An objective choice is the non-informative prior $p(\tau^2) \propto 1/\tau^2$.\\ 
{\bf \underline{Choice of $\ell$:}} The parameter $\ell$ is typically called the \textit{length-scale} parameter of a Gaussian process. It controls the rate of decay of the covariance kernel with the inter-site distances. Typically one chooses $\ell$ so that the correlation between two points far apart in the covariate space is very small. Empirically, one can use a variogram plot of the data to estimate the value of $\ell$. Instead, we used a cross-validation approach to estimate the value of $\ell$. We varied $\ell$ in the range $ \{0.05, 0.1, 0.5,1,2, \dots, 20\}$, where the overall scale is the same after scaling $Q^2$ to $[0,1]$, and used 5-fold cross-validation (repeating the cross validation five times) to chose the optimal value of $\ell$ that minimizes the predictive mean squared error. \\
{\bf \underline{Choice of $\sigma$:}} Based on the error values in the experimental data, we noticed that the estimated error (standard deviation) is no larger than $0.01$, which motivated our choice of values for $\sigma$ for the pseudo data analysis, since is computationally less expensive than adding a prior distribution to it. For the real data analysis we allowed $\sigma$ to vary by putting an objective prior such that $p(\sigma^2)\propto 1/\sigma^2$. \\ 
{\bf \underline{Choice of $N$:}} The number of grid points $N$ in Eq.\,\eqref{eq:basis} directly influences the approximation power of the function estimation method. The role of the grid points is to project a Gaussian process onto a regular grid. The function at any intermediate value is then obtained using linear interpolation (do not confuse with the linear interpolation made by the functions $h_{j}(x)$ in Sec. \ref{Sec:Formalism}). It is unreasonable to set $N$ to a value larger than the sample size $n$, since that may lead to overfitting. In order to conduct a thorough analysis on the full dataset by constrained GP and following \cite{maatouk2017gaussian} we considered $N=\{n/4,n/2,n\}$ in our data-analysis. \\
\noindent {\bf \underline{Gibbs sampling:}} Given the above choice of hyperparameters, the joint posterior distribution in Eq.\,\eqref{eq:posterior} can be updated by Gibbs sampling,
\begin{itemize}
\item Update $[\xi \mid \tau^2, \sigma^2, \Phi, Y] \sim N(\mu_{\xi}, \Sigma_{\xi}) \ind _{{\mathcal{C}_{\xi}}}(\xi)$, with $\Sigma_{\xi} = (\Phi^T \Phi/\sigma^2 + \Gamma^{-1}/\tau^2)^{-1}$ and $\mu_{\xi} = \Sigma_{\xi}^{-1} \Phi^T Y/\sigma^2$,
\item Update $[\tau^2 \mid \xi, \sigma^2, \Phi, Y] \sim \text{IG}( a_{\tau}, b_{\tau})$, with $a_{\tau} = (N+2)/2$ and $b_{\tau} = \xi^T \Gamma^{-1}\xi/2$,
\item Update $[\sigma^2 \mid \xi, \tau^2, \Phi, Y] \sim \text{IG}( a_{\sigma}, b_{\sigma})$, with $a_{\sigma} = n/2$ and $b_{\sigma} = \|Y-\Phi\xi \|^2/2$,
\end{itemize}
where IG denotes an Inverse Gamma distribution. Note that the above Gibbs sampling procedure is applicable to all proposed GP models associated with different constraint sets, and $\Phi, Y, N, \xi, \Gamma$ vary in different cases. 

\subsection{Details on the Pseudo-Data Analysis}\label{DetailsPseudo}

In this section we give details on the implementation and results of the pseudo data analysis presented in Sec. \ref{sec:sims} for the subsets of the data in the three $Q^2$ regimes: i) low $Q^2 (\leq 1.36$ fm$^{-2})$, ii) medium $Q^2 (\leq 4.85$ fm$^{-2})$ and iii) high $Q^2 (\leq 25.12$ fm$^{-2})$. Recall that the data was generated using the Dipole function defined in Eq.~\eqref{Dipole_Def} with an input radius of $0.84$ fm. Fig.\,\ref{fig:dipole_mse} shows the cross validation results for selecting the optimal correlation length. Tables \ref{tab:dipole_I}-\ref{tab:dipole_III} show the result summaries of the estimates of the radius $r_p$ and $95\%$ credible intervals for the three regimes with varying level of noise. Fig.\,\ref{fig:dipole_fit} shows the model fitting in all regimes as $Q^2 \rightarrow 0$, and the inset plot at the right upper corner shows that model fit in its respective entire range.

By dividing the total data into 80\% training and 20\% testing datasets, the cross validation procedure seeks to minimize the MSE (Mean Squared Error) defined as the average of squared deviations from the $20$ \% data held out and the built model. In low and medium $Q^2$ regimes, the MSEs dropped fast for relatively small values of $\ell$, and then stayed flat as $\ell$ increases. On the other hand, in the high $Q^2$ regime the MSEs dropped first and then increased slightly as $\ell$ increases. For very small values of $\ell$ ($\sim 0.1$) we obtained much higher MSEs in all cases, since a very small value of $\ell$ reduces the correlation of the constrained Gaussian process between neighboring points and fails to borrow information from neighbors for an accurate extrapolation. 

The MSE behavior as a function of $\ell$ is fundamentally different for the high regime in comparison with the low and medium regimes. The reason for this difference is the way the experimental points are distributed across the entire $Q^2$ range. In the Mainz data with range $25.12$ fm$^{-2}$, the $Q^2$ are collected more often with small values and only a few are collected with large values: $70\%$ of the $Q^2$ values are less than $5$ fm$^{-2}$ but only around $5\%$ of $Q^2$ are greater than $17$ fm$^{-2}$. Therefore, when re-scaling to the [0,1] interval, in the low regime the $Q^2$ values are more evenly distributed, while in the high regime most of the $Q^2$ values are concentrated around 0 and only a few are close to 1. The selected length-scale parameter $\ell$ depends on the dispersion of the $Q^2$ values, thus, in the high $Q^2$ regime the cross validation procedure tends to select a smaller value of $\ell$ so that the correlation between two points with long distance is relatively small. On the other hand, in the low $Q^2$ regime the cross validation procedure tends to select a larger value of $\ell$ that leads to a stronger correlation between any two points of $Q^2$, causing that when estimating at each point the model can borrow enough information from the neighboring points.

\begin{figure}[htbp!]

\includegraphics[scale = 1]{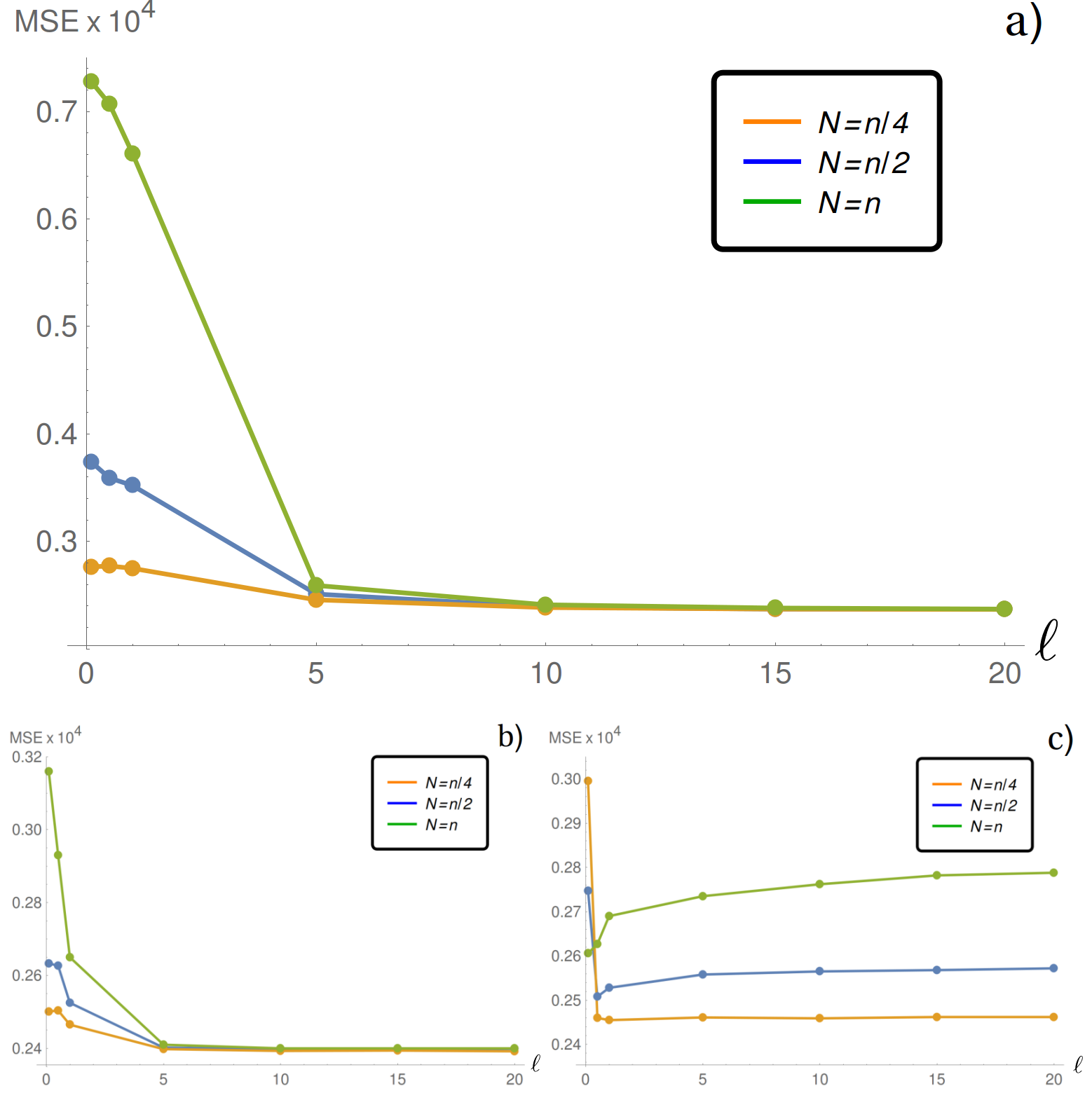}\\

\caption{MSEs versus $\ell$ for $\nu = 2.5$, $\sigma = 0.005$ and $n = 500$. (a) shows the MSEs for data in regime i), (b) for regime ii), and (c) for regime iii). In each plot, the orange line stands for the MSEs with the number of grid points $N = n/4$, the blue line stands for $N = n/2$, and the green line for $N = n$.}
\label{fig:dipole_mse}
\end{figure}

\begin{table}[htbp]
 \caption{Posterior estimates of the radius and credible interval for cGP, c$_0$GP, c$_1$GP and uGP with $N =n/4$, $\nu = 2.5$ and $\ell_\mathrm{opt}=20$ for a subset of data of size $ n=500$ in regime i).} \small
 \begin{center}
  \begin{tabular}{cclcccc} \hline 
  $\bf{ \sigma}$ & & 0   & 0.002 & 0.005 & 0.01 \\ \hline

  cGP & $r_p$ & 0.8400 & 0.8384 & 0.8402 & 0.8415 \\
  & CI$_\mathrm{l}$ & 0.8393 & 0.8340 & 0.8340 &0.8304\\
  & CI$_\mathrm{u}$ &0.8403 & 0.8429 & 0.8488 & 0.8530\\ \hline
 
 c$_0$GP &$r_p$ & 0.8400 & 0.8359 & 0.8364 & 0.8528 \\
  &CI$_\mathrm{l}$& 0.8393 & 0.8272 & 0.8243 & 0.8235\\
  &CI$_\mathrm{u}$ & 0.8402 & 0.8440 & 0.8536 & 0.8951 \\\hline

  c$_1$GP &$r_p$ & 0.8388 & 0.8391 & 0.8435 & 0.8584 \\
  &CI$_\mathrm{l}$& 0.8375 & 0.8301 & 0.8302 & 0.8337 \\
  &CI$_\mathrm{l}$ & 0.8401 & 0.8474 & 0.8612 & 0.8881\\ \hline

  uGP &$r_p$& 0.8389 &0.8363 & 0.8336 & 0.8315 \\
  &CI$_\mathrm{l}$ & 0.8376 & 0.8182 &0.8089 &0.7801 \\
  &CI$_\mathrm{u}$ & 0.8403 & 0.8555 & 0.8601 & 0.8780 \\ \hline
   
\end{tabular}%
\end{center}
\label{tab:dipole_I}%
\end{table}%

\begin{table}[htbp]
\caption{Posterior estimates of the radius and credible interval for cGP, c$_0$GP, c$_1$GP and uGP with $N =n/4$, $\nu = 2.5$ and $\ell_\mathrm{opt}=20$ for a subset of data of size $n = 500$ in regime ii).} \small
 \centering
 \begin{center}
  \begin{tabular}{cclcccc} \hline
  $\bf{ \sigma}$ & & 0   & 0.002 & 0.005 & 0.01 \\ \hline

 cGP  &$r_p$ & 0.8399 & 0.8337 & 0.8327 & 0.8319 \\
  &CI$_\mathrm{l}$ & 0.8394 & 0.8297 & 0.8276 & 0.8238 \\
   &CI$_\mathrm{u}$ & 0.8404 & 0.8381 & 0.8413 & 0.8421 \\ \hline
 
  c$_0$GP &$r_p$ & 0.8399 & 0.8376 & 0.8341 & 0.8430 \\
   &CI$_\mathrm{l}$& 0.8394 & 0.8311 & 0.8276 & 0.8217 \\
   &CI$_\mathrm{u}$& 0.8404 & 0.8459 & 0.8413 & 0.8705 \\ \hline

   c$_1$GP &$r_p$& 0.8385 & 0.8317 & 0.8247 & 0.8200 \\
  &CI$_\mathrm{l}$& 0.8373 & 0.8232 & 0.8152& 0.8074\\
   &CI$_\mathrm{u}$&  0.8399 & 0.8397 & 0.8354 & 0.8337 \\ \hline

  uGP &$r_p$& 0.8385 & 0.8344 & 0.8306 & 0.8237 \\
  &CI$_\mathrm{l}$& 0.8372 & 0.8208 & 0.8126 & 0.8027 \\
   &CI$_\mathrm{u}$ & 0.8398 & 0.8487 &0.8544 & 0.8517 \\ \hline
    \end{tabular}%
\end{center}
\label{tab:dipole_II}%
\end{table}%

\begin{table}[h]
\caption{Posterior estimates of the radius and credible interval for cGP, c$_0$GP, c$_1$GP and uGP with $N =n/4$, $\nu = 2.5$ and $\ell_\mathrm{opt} = 1$ for a subset of data of size $n = 500$ in regime iii).} \small
 \centering
 \begin{center}
  \begin{tabular}{cclcccc} \hline
  $\bf{ \sigma}$ & & 0   & 0.002 & 0.005 & 0.01 \\ \hline

  cGP  &$r_p$ & 0.8386 & 0.8312 & 0.8213 & 0.819\\
  &CI$_\mathrm{l}$ & 0.8249 & 0.8252 & 0.8131 & 0.8035 \\
   &CI$_\mathrm{u}$ & 0.8454 & 0.8373 & 0.8295 & 0.8250 \\ \hline

 c$_0$GP  &$r_p$ & 0.8169 & 0.8355 & 0.8294 & 0.8255 \\
   &CI$_\mathrm{l}$& 0.8151 & 0.8241 & 0.8126 & 0.8077\\
   &CI$_\mathrm{u}$& 0.8188 & 0.8478 & 0.8451 & 0.8458 \\ \hline

  c$_1$GP &$r_p$& 0.8303 & 0.8176 & 0.8071& 0.7939 \\
  &CI$_\mathrm{l}$&  0.8019 & 0.8077 & 0.7963 & 0.7792 \\
   &CI$_\mathrm{u}$& 0.8409 & 0.8290 & 0.8191 & 0.8088\\ \hline

  uGP &$r_p$& 0.8400 & 0.8229 & 0.8202 & 0.8118 \\
  &CI$_\mathrm{l}$& 0.8387 & 0.8030& 0.7918 & 0.7765 \\
   &CI$_\mathrm{u}$ & 0.8404 & 0.8438 & 0.8448 & 0.8526 \\ \hline
    \end{tabular}%
\end{center}
\label{tab:dipole_III}%
\end{table}%

\begin{figure}[htbp!]
\begin{center}
  \includegraphics[scale = 0.7]{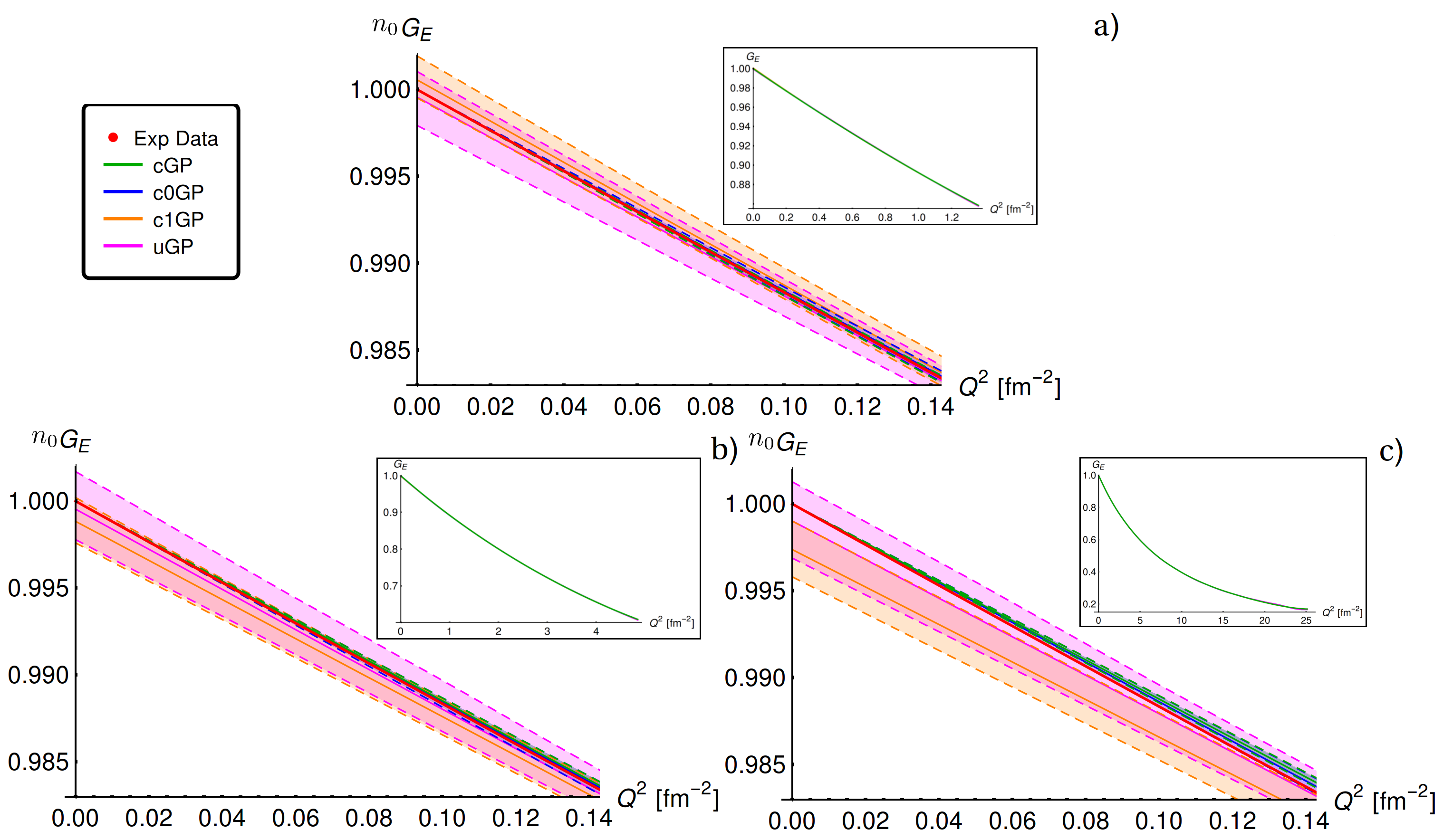}
\end{center}
\caption{Model fits in regimes i)-iii) with noise level $\sigma = 0.005$. (a) is for regime i), (b) for regime ii) and (c) for regime iii). The red line stands for the true function values; the green line stands for the cGP estimates; the blue line stands for the c$_0$GP; the orange line stands for the c$_1$GP; the purple line stands for uGP. In each case the $95\%$ point-wise credible intervals are delimited by dashed lines with the respective color.}\small
\label{fig:dipole_fit}
\end{figure}

The first column in Table \ref{tab:dipole_I}-\ref{tab:dipole_II} shows that in the no-noise setting of the low and medium regimes, both cGP and c$_0$GP estimate the radius very close to $0.84$ fm (the true value), however c$_1$GP and uGP give slightly biased estimates, consistently toward lower values. We find that without the restriction \eqref{GE0}, the estimates of the radius are drifted away from the true value even when there is no random noise in $G_E$. 

In Table \ref{tab:dipole_I}, where the $Q^2$ values are in the low regime, the cGP method estimates $r_p$ very well, and the credible interval becomes wider when the noise level increases. c$_0$GP recovers $r_p$ close to the true value when the noise level is small, but the estimates become biased as the noise level increases. Also with higher noise level, the credible intervals of c$_0$GP are much wider than those of cGP, in the same sense as the uGP credible intervals are wider than those of c1GP. This behavior indicates us that the shape constrains (Eq.~\eqref{GE1} and~\eqref{GE2}) can indeed play an important role in reducing the variability of the estimation. Nevertheless, if we only consider the shape constraints (refer to the third row in Table \ref{tab:dipole_I}), the estimates of the first derivatives are somewhat variable as the noise level increases, indicating that in the presence of noise, imposing all the physical constraints reduces the uncertainty in the estimation while maintaining accuracy. We find that uGP leads to reasonably good estimates and thin credible intervals when the noise level is small, however when the noise level is $\sigma = 0.01$, it leads to a credible interval $0.1$ fm wide.

We observed similar results from Tables \ref{tab:dipole_II} and \ref{tab:dipole_III} for the medium and the high $Q^2$ regimes. Comparing results across different regimes, we found that in the medium and the high $Q^2$ regimes, all the GP methods tend to give lower estimates at the origin as the noise level increases. In the high $Q^2$ regime we obtained slightly biased estimates of the radius and wider credible intervals even for $\sigma = 0$, especially for c$_0$GP. The fact that in the high $Q^2$ regime we do not obtained as good estimates of the radius as we obtained in the low $Q^2$ regime could be related to the smaller correlation length selected, as we explained in the previous paragraphs. Since in the high regime small $\ell$ reduces the correlations between points which are close to each other, when estimating over $Q^2 \approx 0$ the model can use less information from the data near the origin than in the low regime case. We shall investigate this topic further in a future work and propose a way to gauge the bias of our estimates and to improve the overall prediction.

We can see in Fig.\,\ref{fig:dipole_fit} that without the restriction $n_0G_E(0)=1$, the estimates of c$_1$GP and uGP are off from the truth (red line) for small values of $Q^2$. On the contrary, and as expected, cGP and c$_0$GP agree with the truth as $Q^2 \rightarrow 0$. Without the shape constraints (Eq.~\eqref{GE1}-\eqref{GE2}) we found in the high $Q^2$ regime that the estimates of c$_0$GP and uGP are not even convex toward higher values of $Q^2$.

\newpage

\subsection{Details on the Electron-Scattering Data Analysis}\label{DetailsElec}

Fig.\,\ref{fig:MSE_Combined} (a) shows the 5-fold cross-validation MSEs for the high regime over $\ell \in \{ 0.05, 0.1, 0.5, 1\}$ of cGP for $n=250$ with $N= \{n/4, n/2, n\}$. Since it is evident that smaller values of $\ell$ ($\leq 0.5$) causes the MSE to increase, we focused on $\ell \ge 0.5$ in the subsequent analysis, following also our observations from the MSEs results in the high $Q^2$ regime of pseudo generated data. From the results used to plot Fig.\,\ref{fig:MSE_Combined} (a) and Fig.~\ref{fig:dipole_mse}, we noted that choosing a smaller number of grid points leads to more accurate predictions in terms of MSE when $\ell>0.5$. We therefore chose the number of grid points $N = n/4$ and considered different smoothness parameters $\nu = \{2.5, 3, 3.5\} $ to perform the finer cross validation procedure to select the optimal value of $\ell$ in the grid $\ell \in\{0.5, 1, 1.5, 2, 2.5\}$ (Fig.\,\ref{fig:MSE_Combined} (a) (Inset)). We saw that the MSEs increased as $\ell$ increased from $0.5$ to $2.5$ for all $\nu$ and $\nu=2.5$ gave relatively lower MSEs in this case. The results of our analysis guide us to chose $\ell_\mathrm{opt}=0.5$ for the full data set analysis.

Fig.\,\ref{fig:MSE_Combined} (b) shows the 5-fold cross-validation MSEs for the low regime. Again based on the results in the pseudo data analysis, that a larger value of $\ell$ is preferred, we conducted the cross validation for cGP model over the parameter set $ \ell \in \{1, 5, 10, 15, 20 \}$. Fig.\,\ref{fig:MSE_Combined} (b) shows that as $\ell$ increases, the MSEs drops fast first and then stays stable for large values of $\ell$ ($\ge 10$). Also, in the low regime cGP with smaller number of grid points ($N=n/4$) gave lower MSEs, which is similar to the full dataset case. This analysis leads us to chose $\ell_\mathrm{opt}=10$ on the low $Q^2$ set.

\begin{figure}[h!]
\begin{center}
  \includegraphics[scale = 0.7]{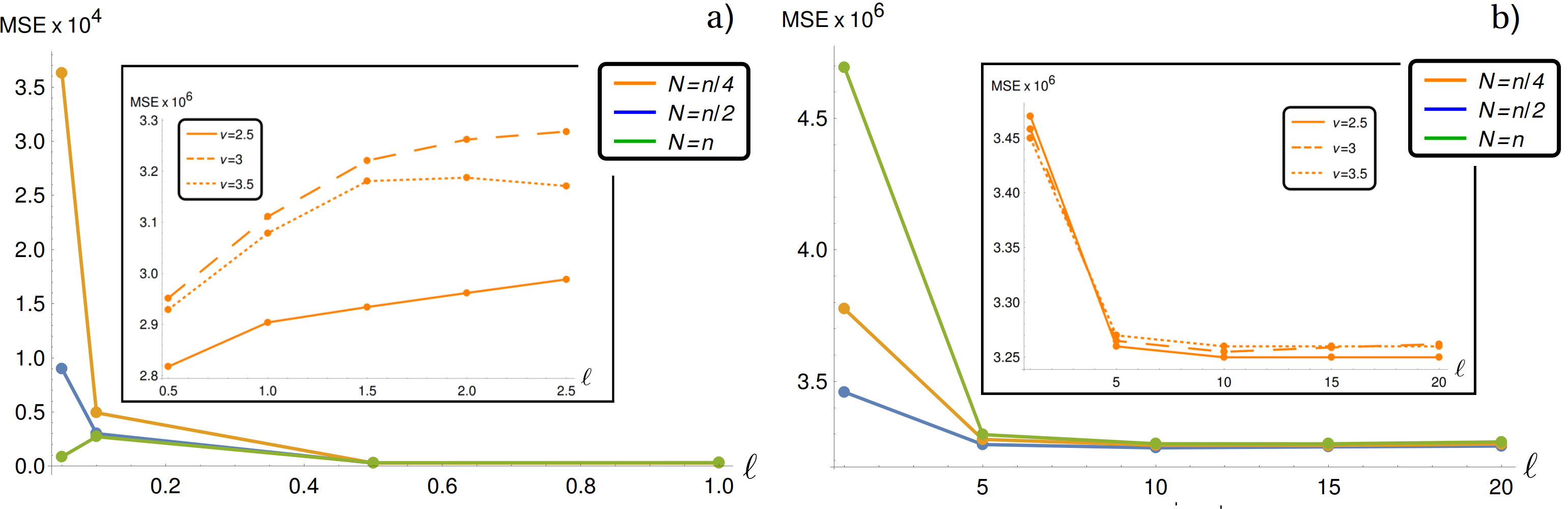}

 \end{center}  
\caption{(a) MSEs over the range $\ell \in [0,1]$ for cGP model in the high regime with $n=250$, $N = n/4$, $n/2$, $n$. (b) MSEs over $\ell$ for cGP model in the low regime with $N = \{n/4, n/2, n\}$, $n=500$ and $\nu = 2.5$. In both graphs the green line stands for the case with $N = n$, the blue line stands for $N = n/2$, and the orange line stands for $N= n/4$. (a) (Inset) MSEs over $\ell \in [0.5,2.5]$ for cGP model with $N=n/4$ on full dataset. (b) (Inset) MSEs over $\ell \in [1,20]$ for cGP model with $N=n/4$ on the low regime. In both insets $\nu = 2.5$ (solid), $\nu = 3$ (dashed), $\nu = 3.5$ (dot-dashed).}
\label{fig:MSE_Combined}
\end{figure}

The following figures show the detailed histograms for the 400 MCMC samples for the four models discussed in the real data analysis section. In all cases $\nu=2.5$ was used. Fig.\,\ref{fig:highhist_n4} and \ref{fig:highhist_n} show the results for the high regime with $N=n/4$ and $N=n$ respectively. Fig.\,\ref{fig:lowhist_n4} and \ref{fig:lowhist_n} show the results for the low regime with $N=n/4$ and $N=n$ respectively. Fig.\,\ref{fig:highhist_f0} and \ref{fig:lowhist_f0} show the samples of $\xi_1$ ($n_0G_E(0)$) for c$_1$GP and uGP for high and low regimes respectively.

\begin{figure}
\begin{center}
  \includegraphics[scale = 0.55]{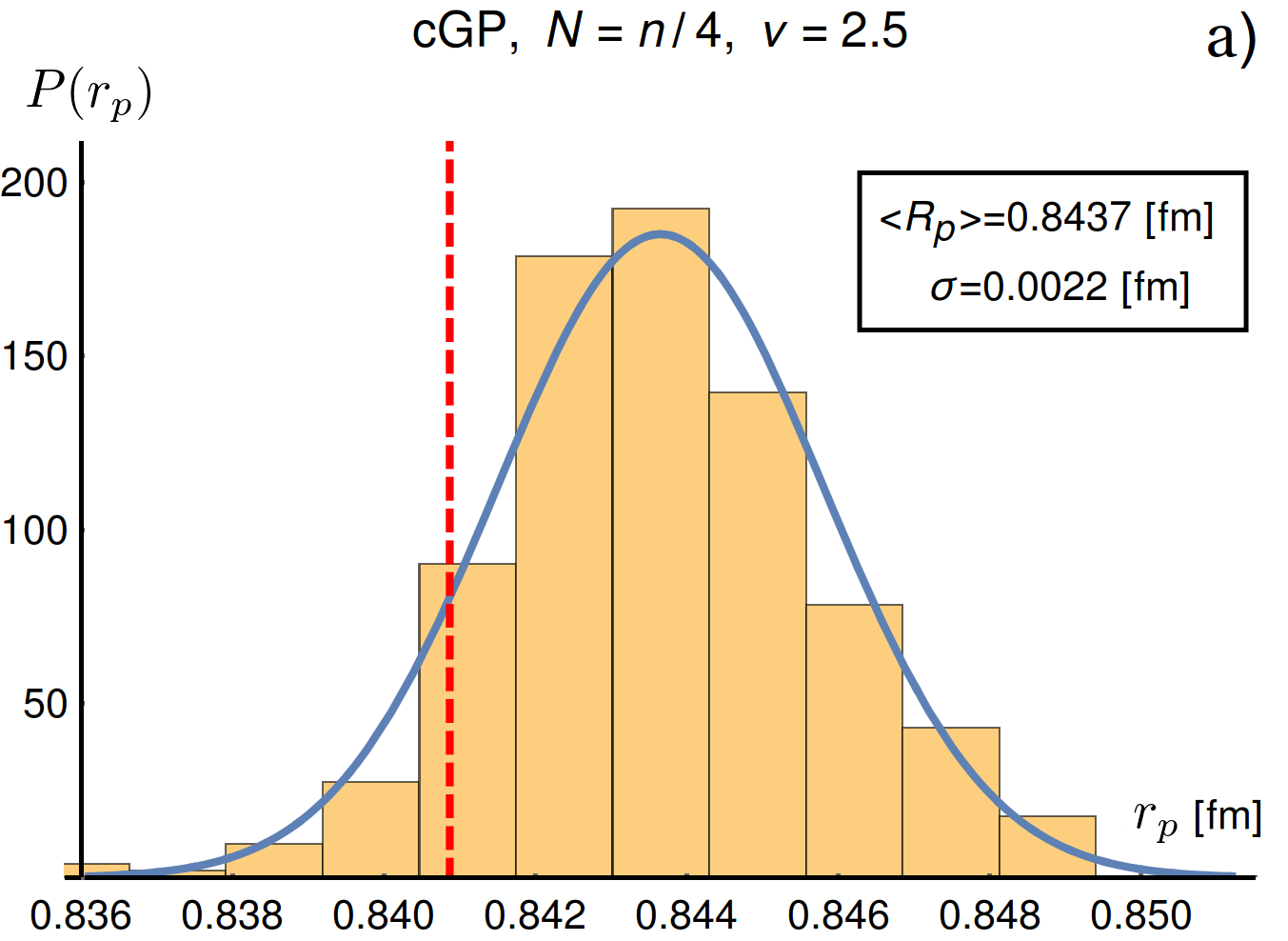}
  \includegraphics[scale = 0.55]{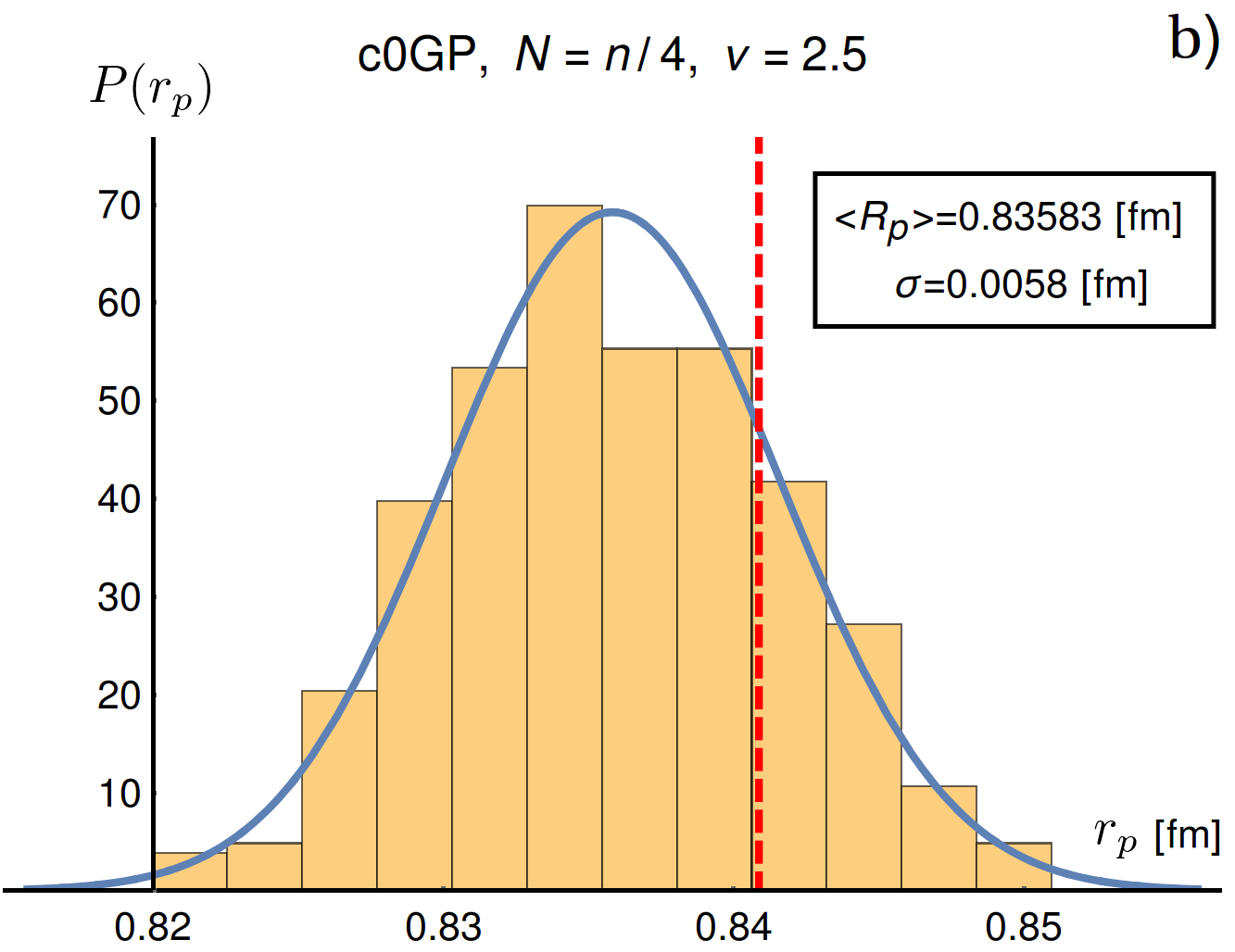}
  \\
  \includegraphics[scale = 0.55]{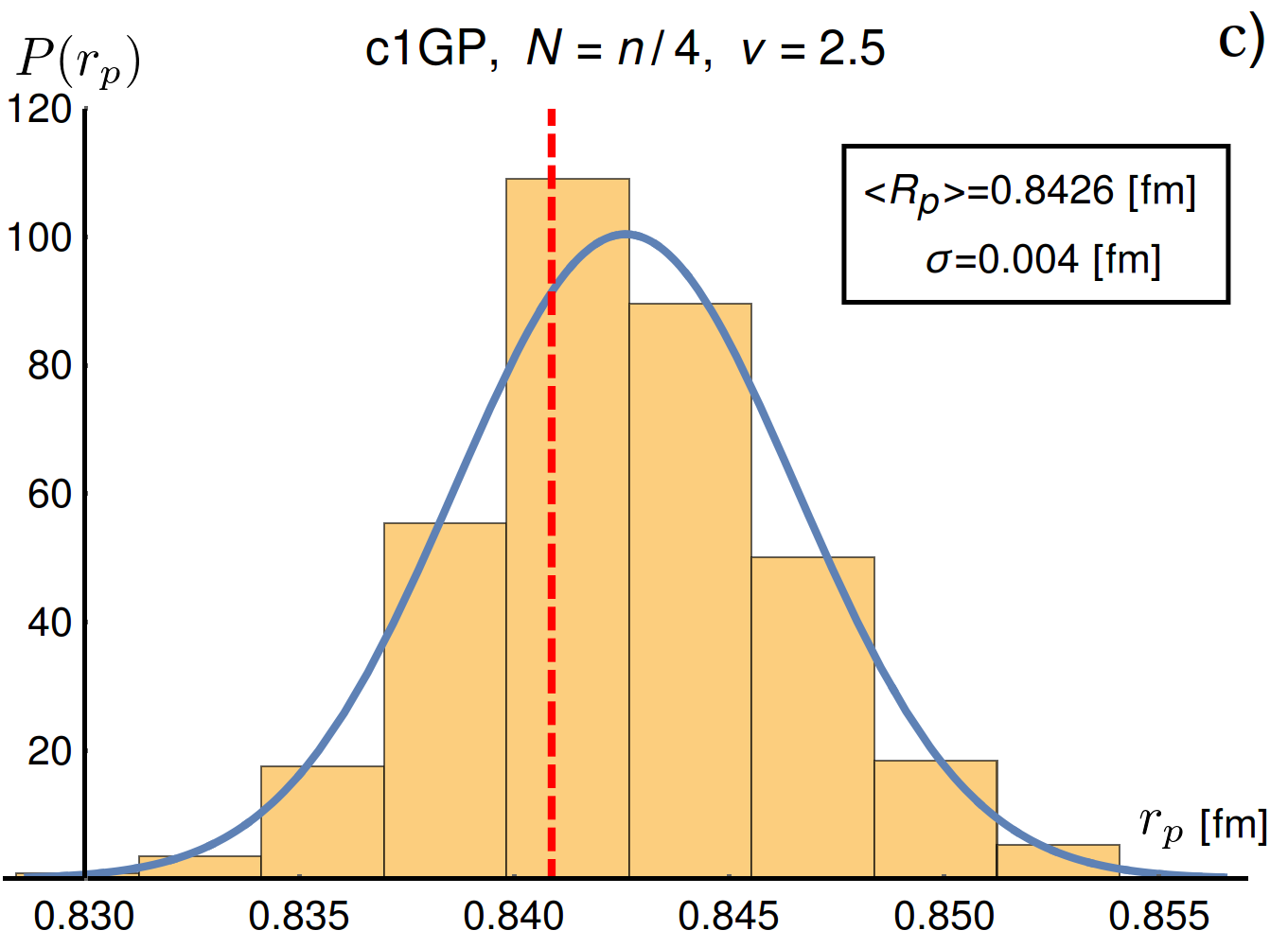}
  \includegraphics[scale = 0.55]{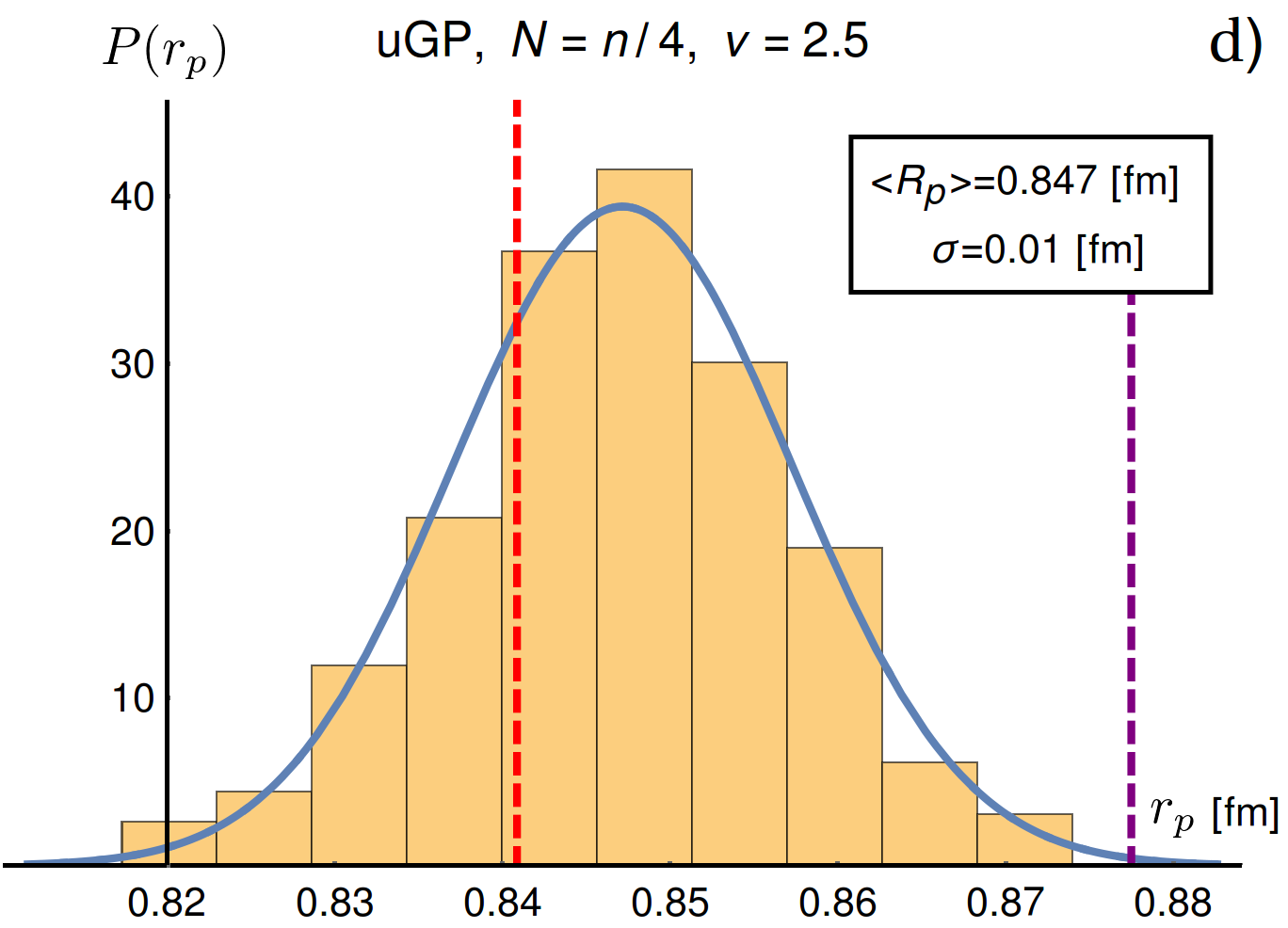}

  \end{center}
\caption{Histogram plots of MCMC samples of the radius $r_p$ for cGP (a), c$_0$GP (b), c$_1$GP (c) and uGP (d) with $N = n/4$ and $\nu = 2.5$ for the full dataset. The red and purple vertical dashed lines indicate the values of 0.84087 fm and 0.8775 fm respectively.}
\label{fig:highhist_n4}
\end{figure}

\begin{figure}
\begin{center}
   \includegraphics[scale = 0.55]{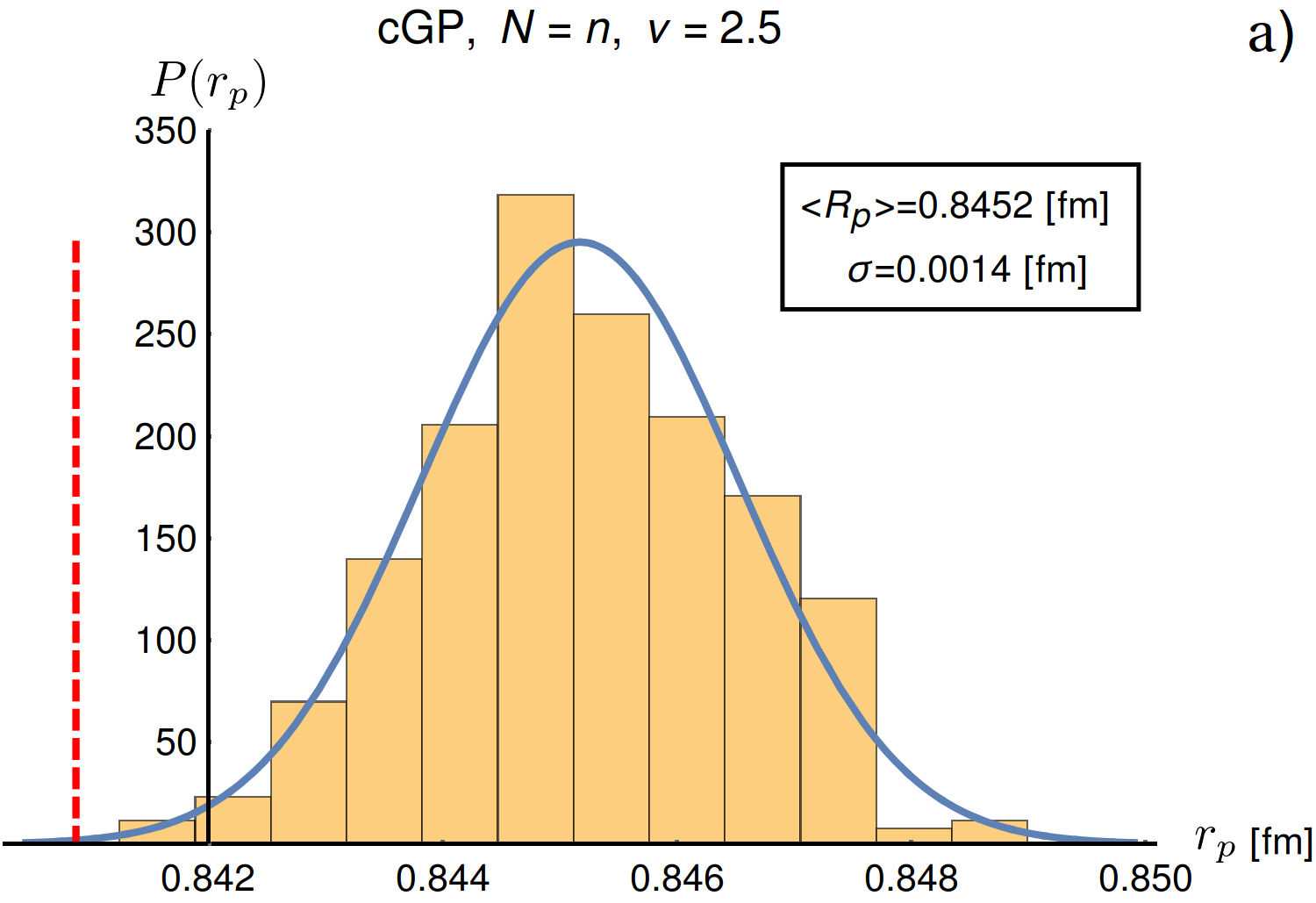}
  \includegraphics[scale = 0.55]{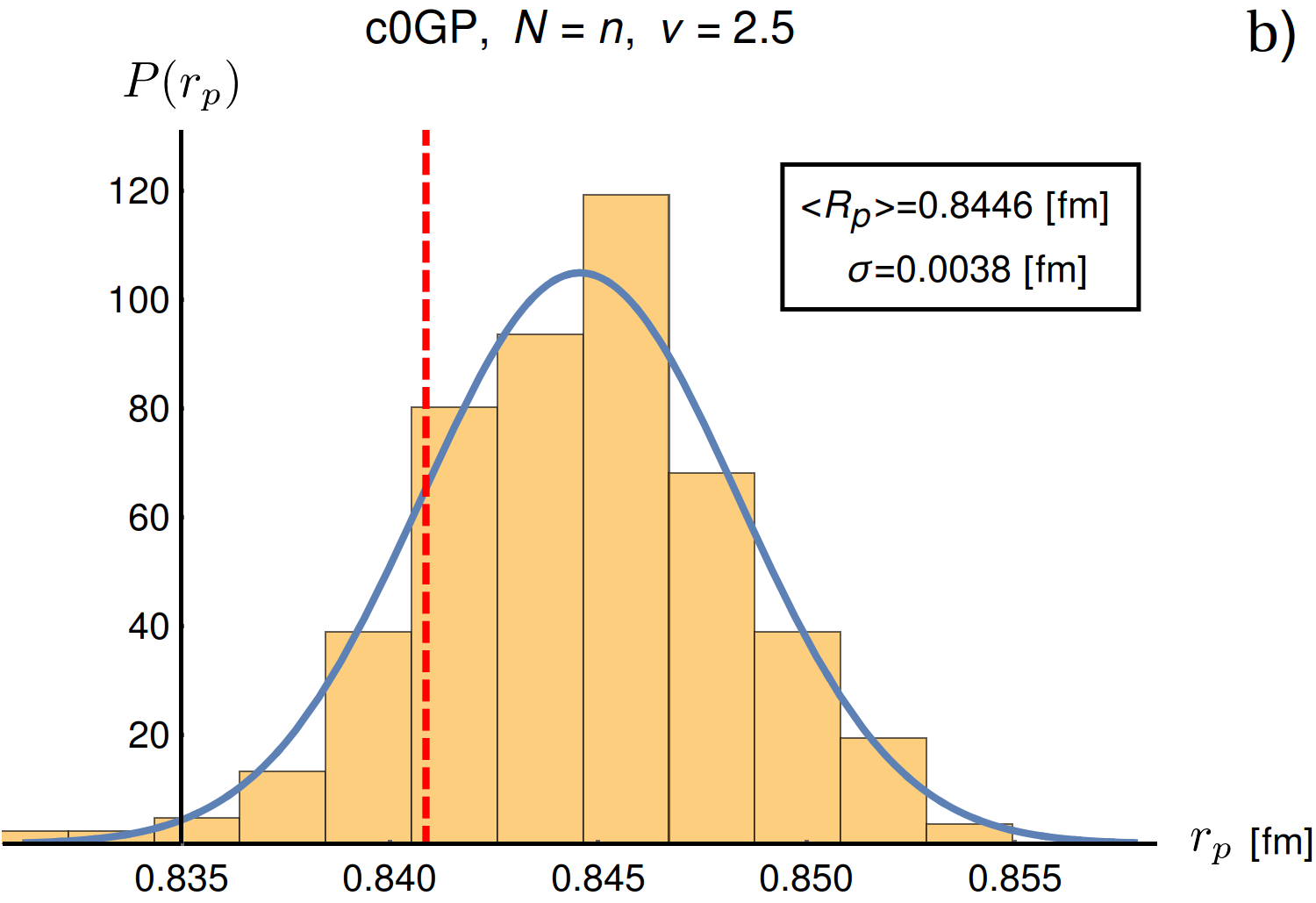}
  \\
   \includegraphics[scale = 0.55]{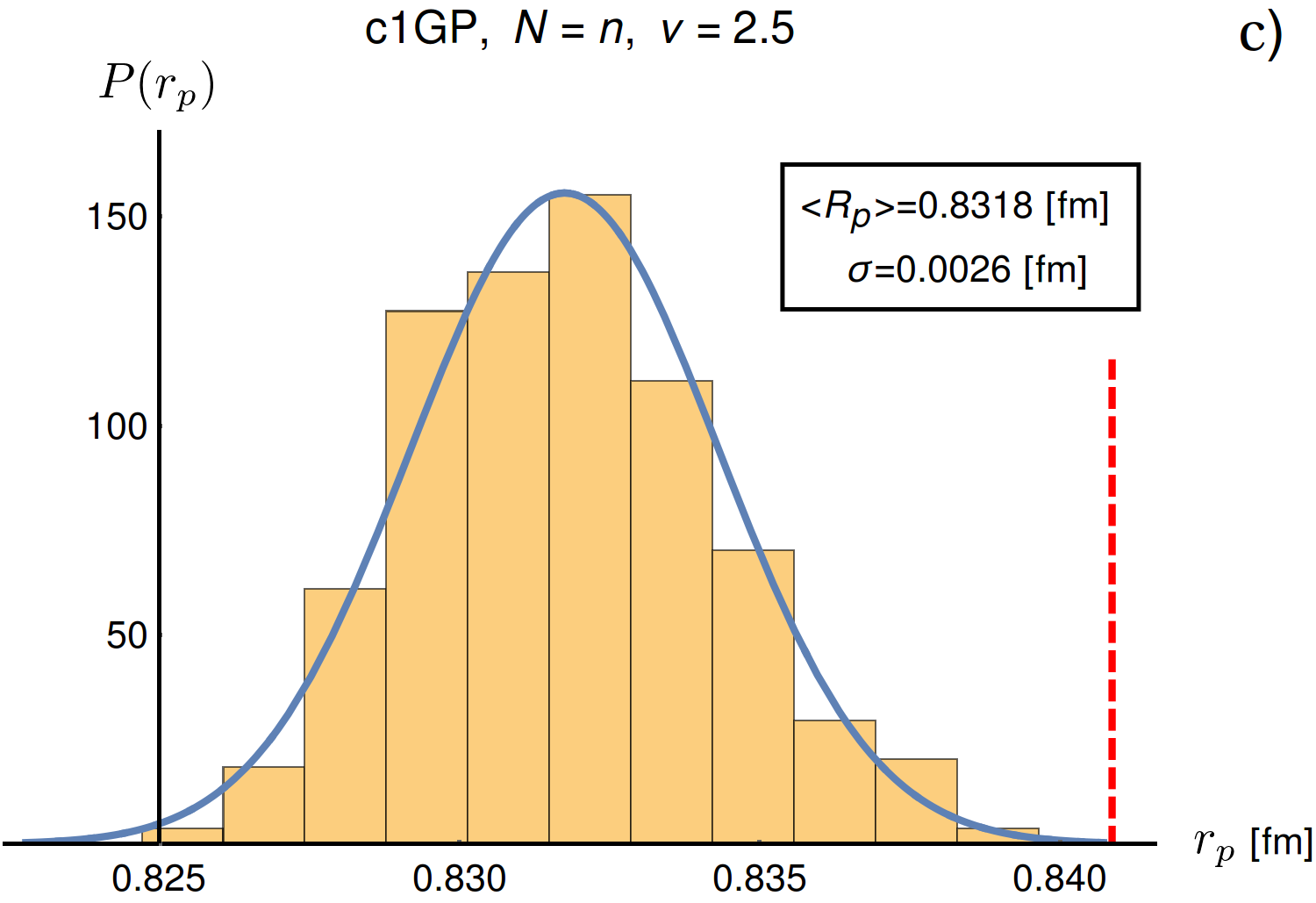}
  \includegraphics[scale = 0.55]{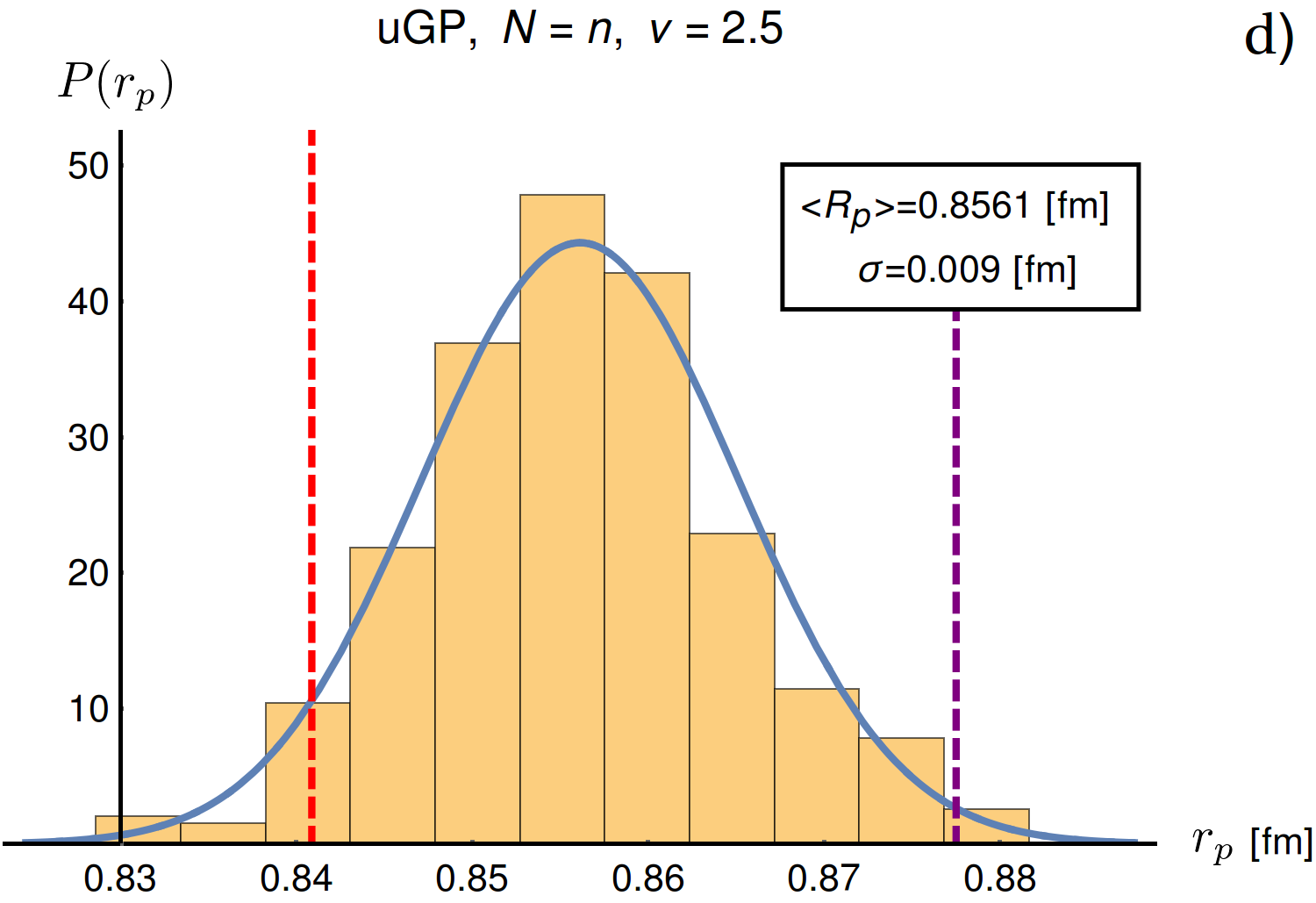}
  \end{center}
\caption{Histogram plots of MCMC samples of the radius $r_p$ for cGP (a), c$_0$GP (b), c$_1$GP (c) and uGP (d) with $N = n$ and $\nu = 2.5$ for the full dataset. The red and purple vertical dashed lines indicate the values of 0.84087 fm and 0.8775 fm respectively.}
\label{fig:highhist_n}
\end{figure}

\begin{figure}
\begin{center}
  \includegraphics[scale = 0.55]{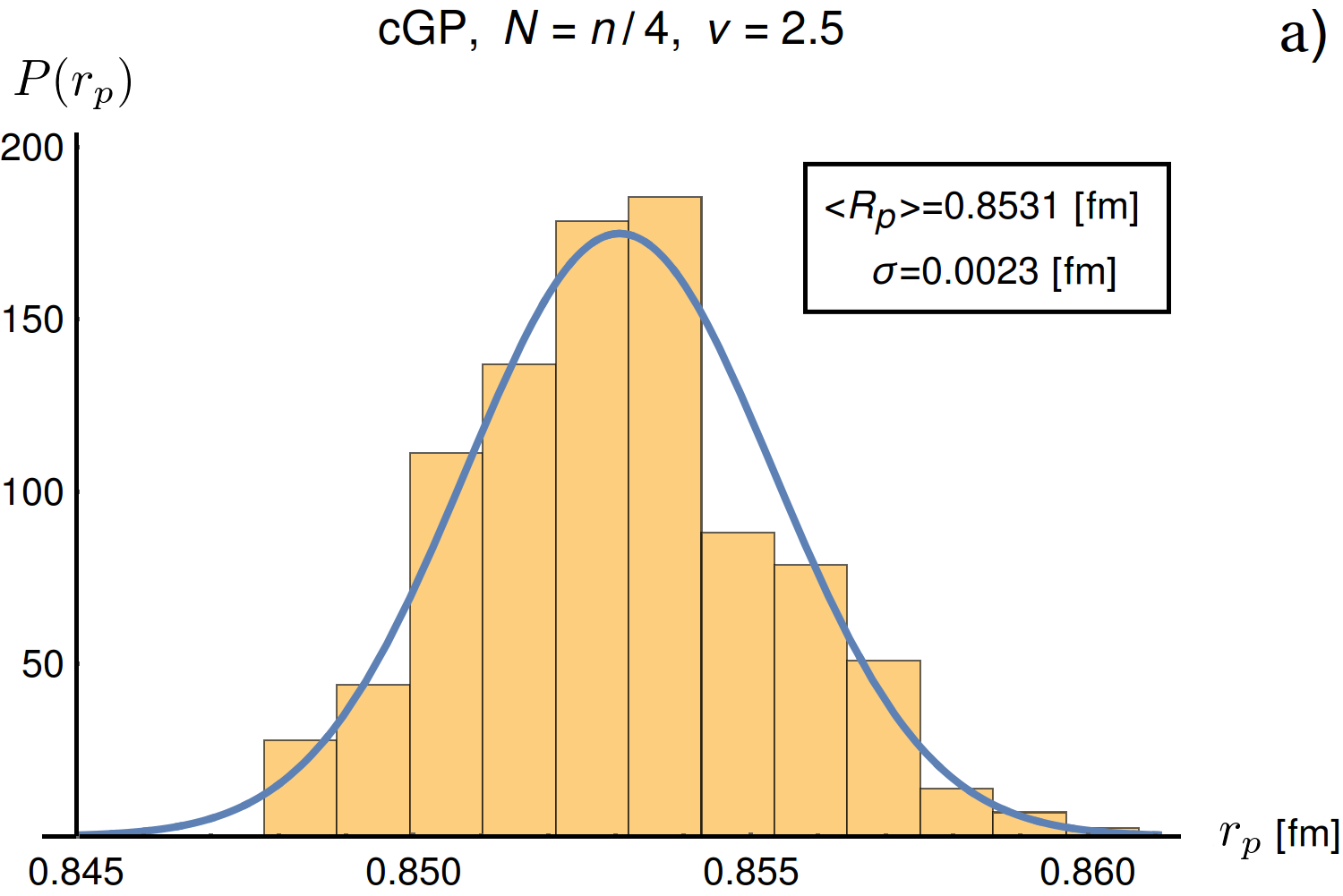}
  \includegraphics[scale = 0.55]{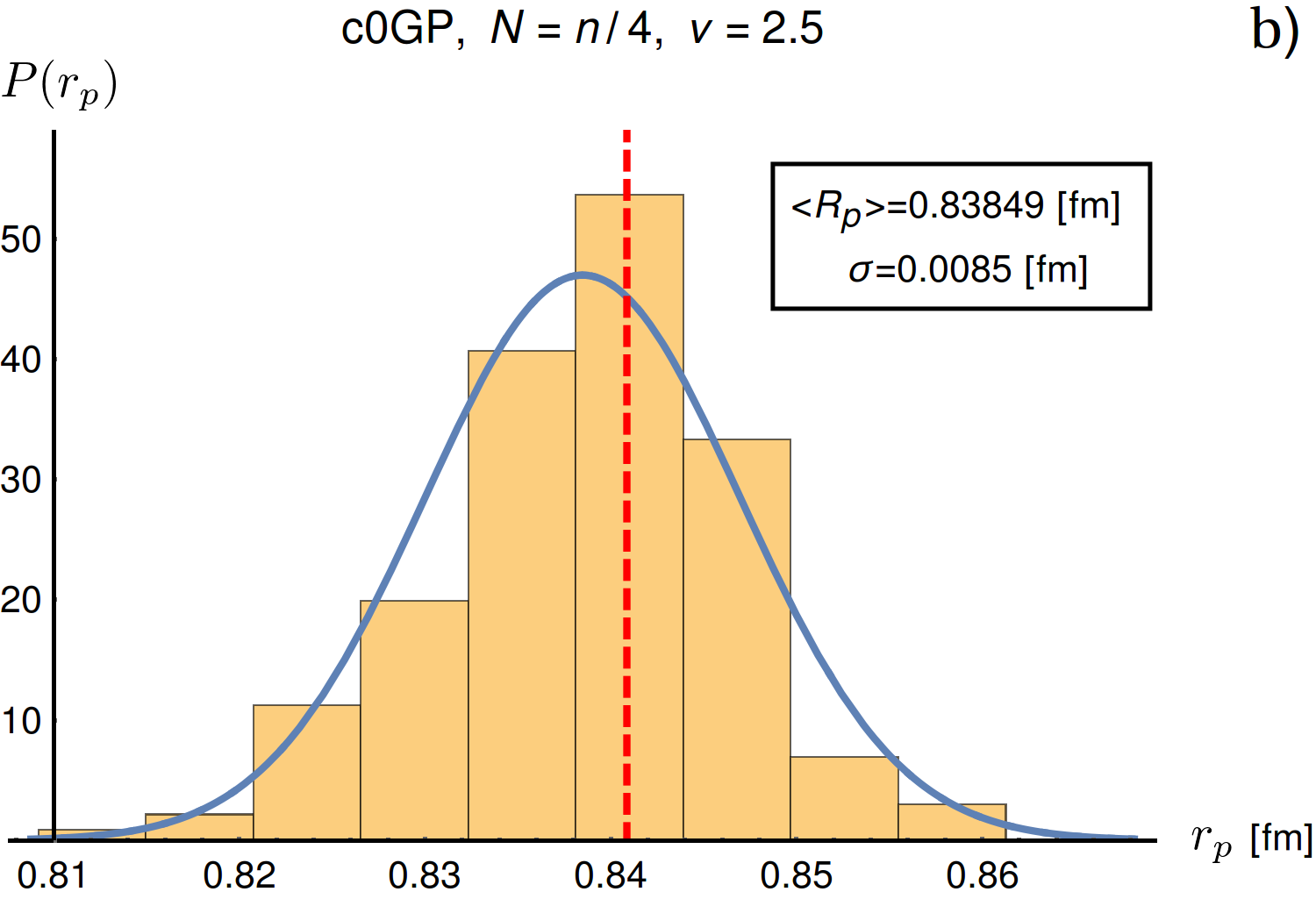}
  \\
  \includegraphics[scale = 0.55]{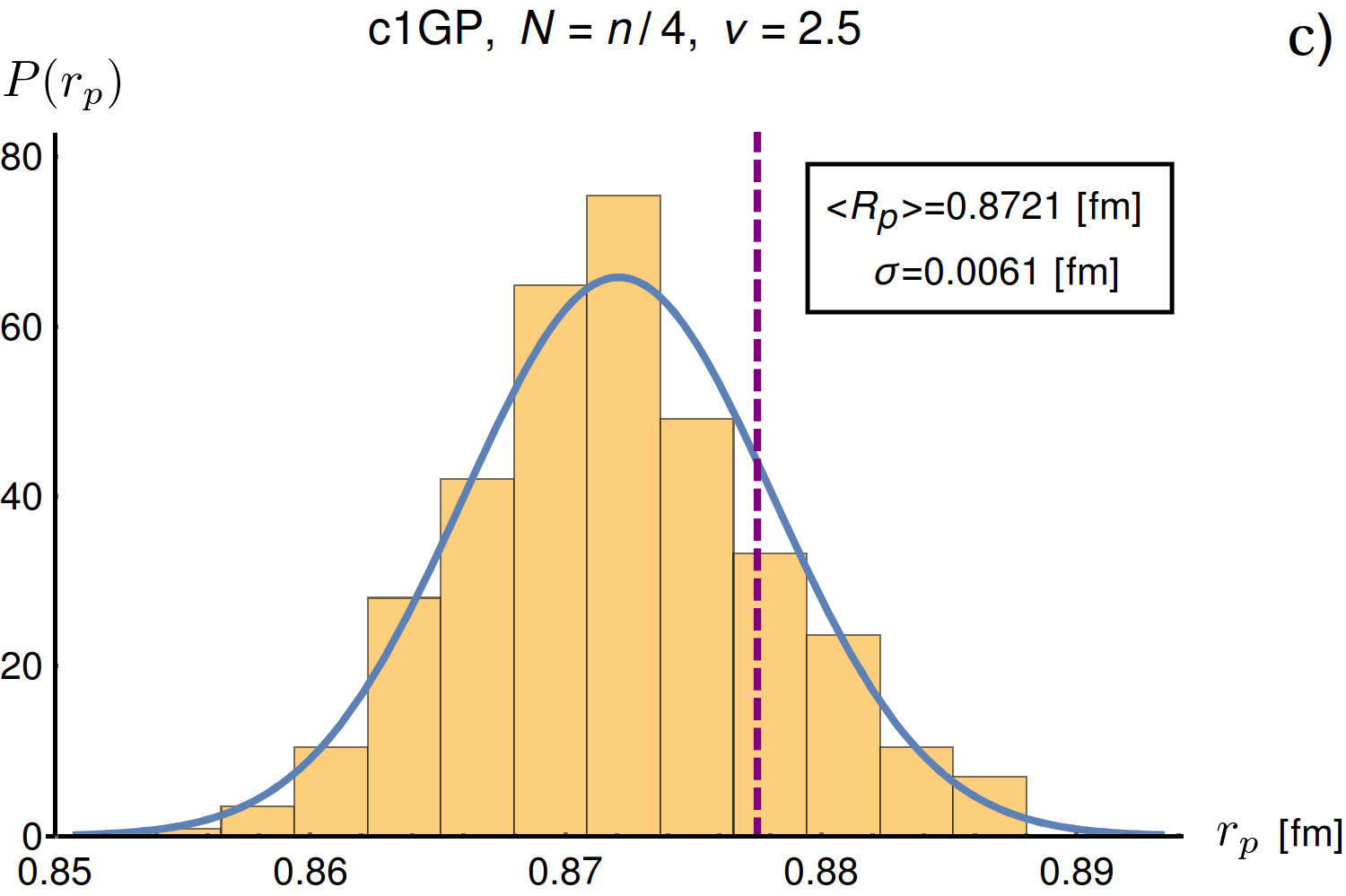}
  \includegraphics[scale = 0.55]{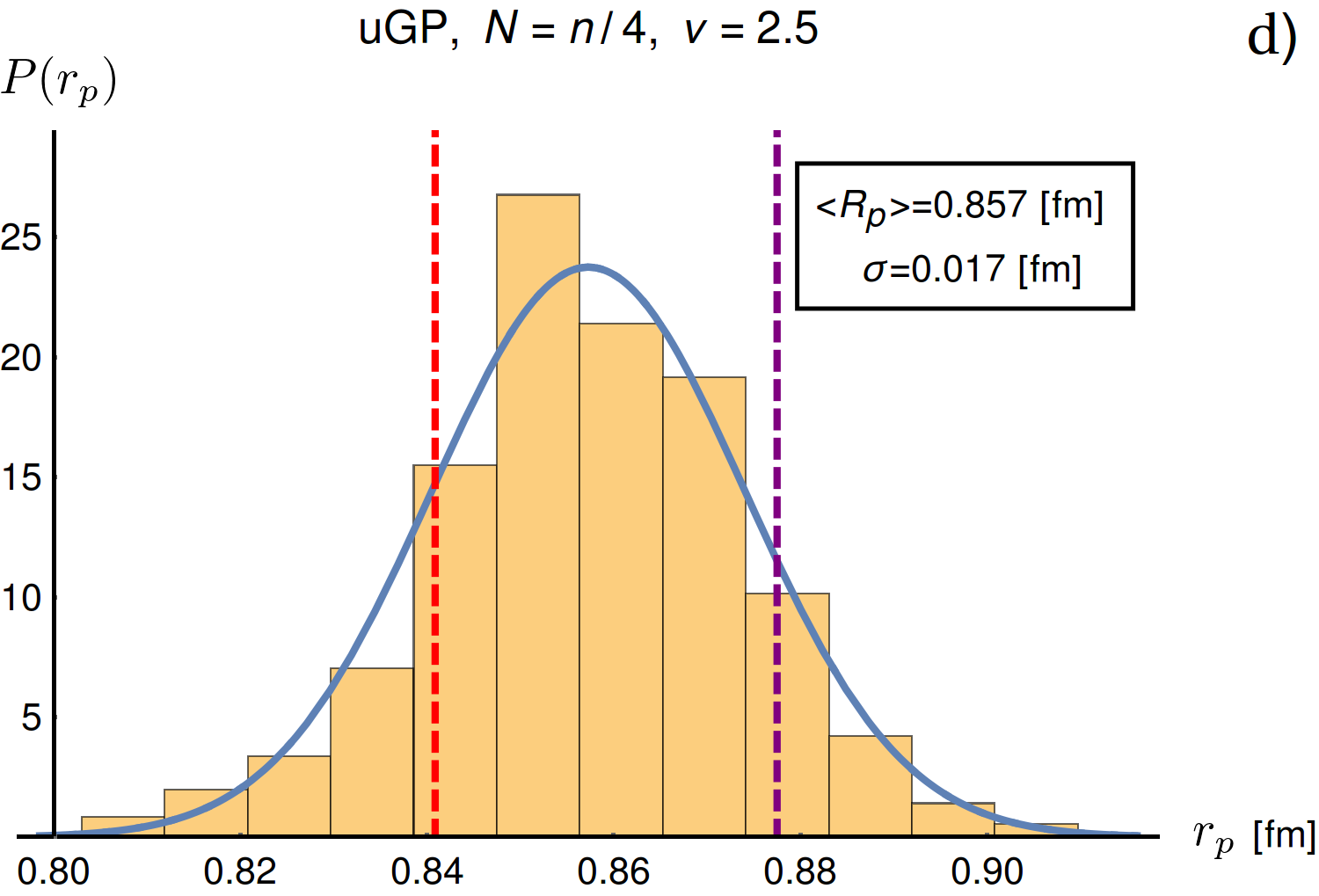}

  \end{center}
\caption{Histogram plots of MCMC samples of the radius $r_p$ for cGP (a), c$_0$GP (b), c$_1$GP (c) and uGP (d) with $N = n/4$ and $\nu = 2.5$ in the low regime case. The red and purple vertical dashed lines indicate the values of 0.84087 fm and 0.8775 fm respectively.}
\label{fig:lowhist_n4}
\end{figure}

\begin{figure}
\begin{center}
   \includegraphics[scale = 0.55]{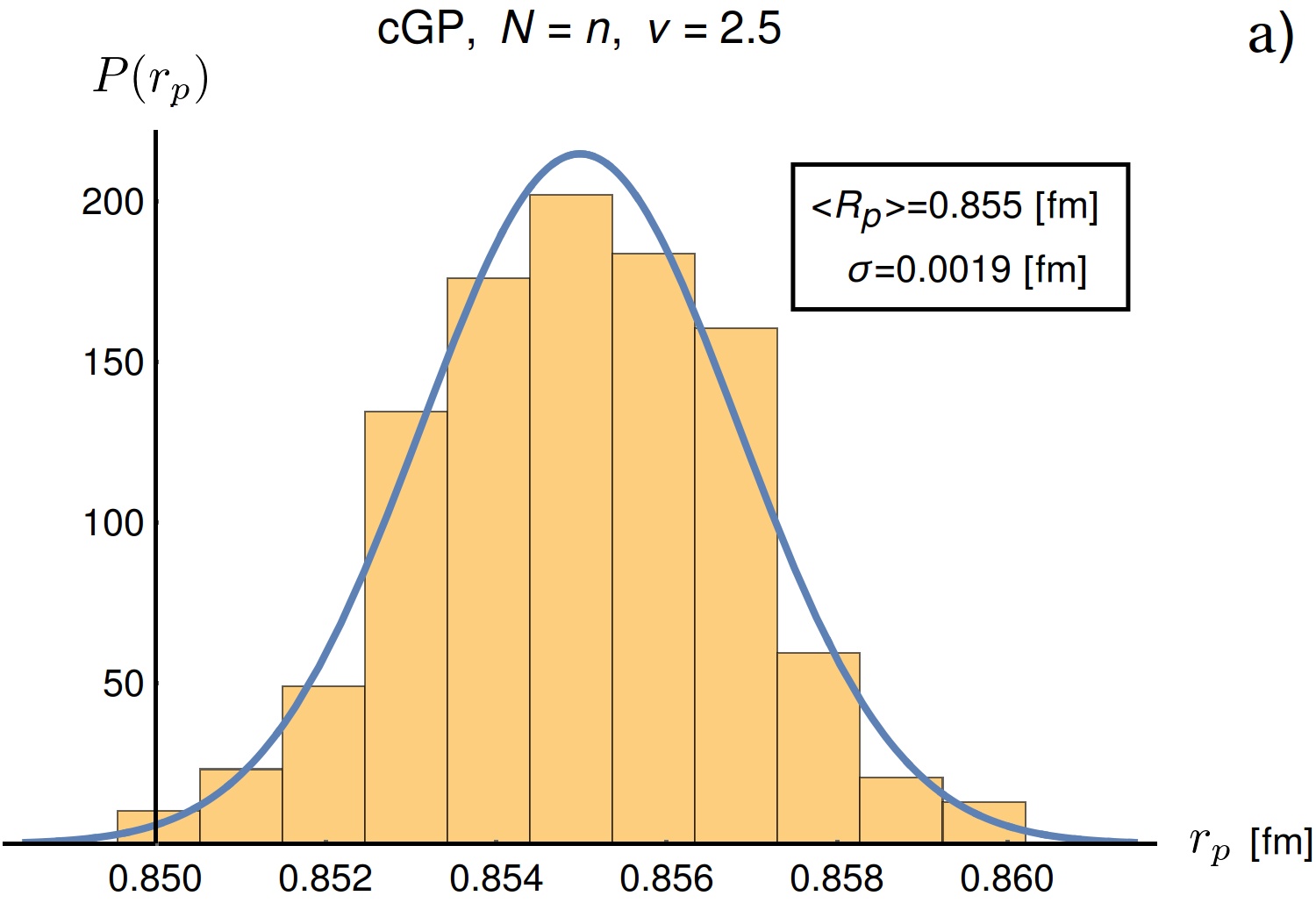}
  \includegraphics[scale = 0.55]{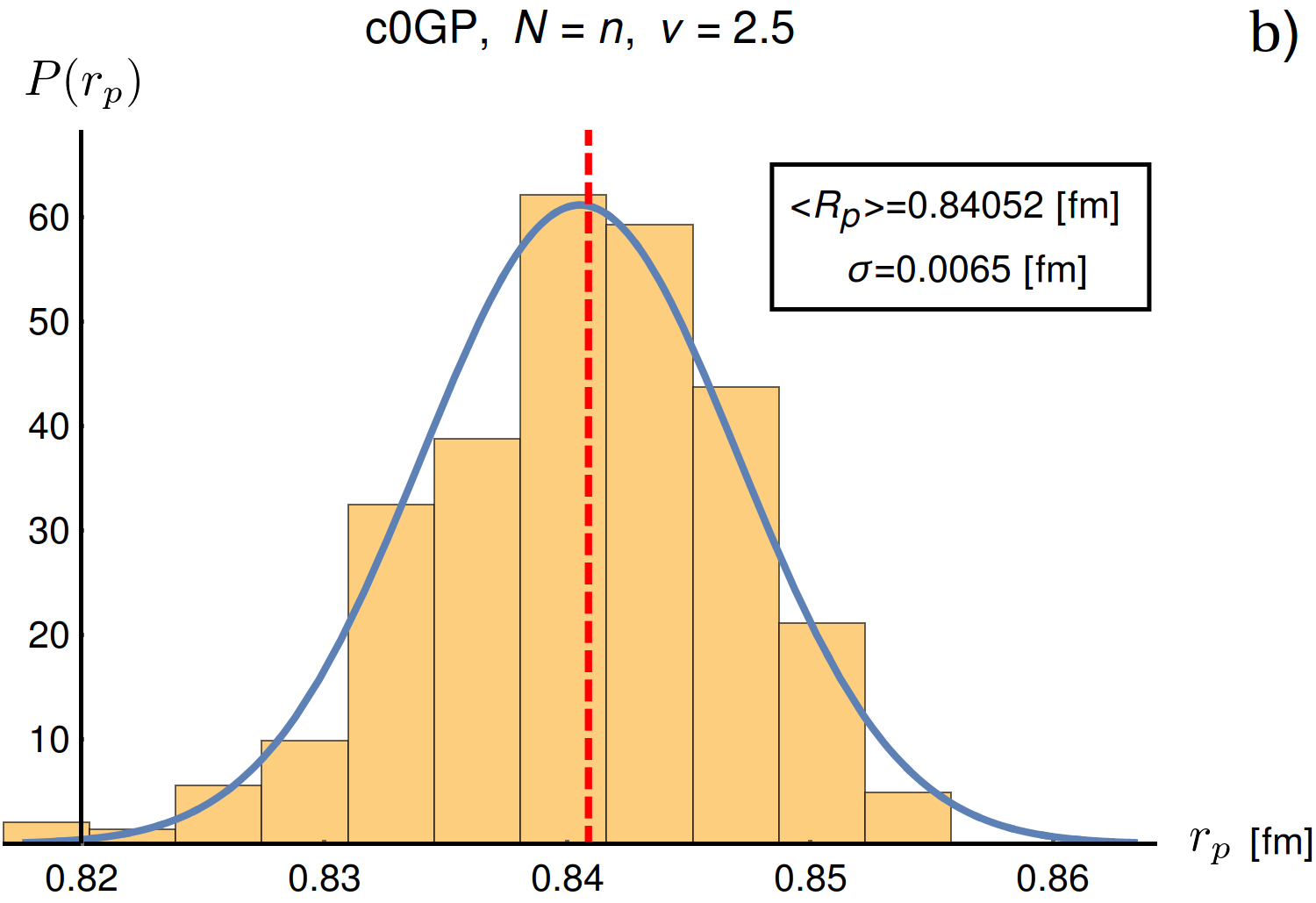}
  \\
   \includegraphics[scale = 0.55]{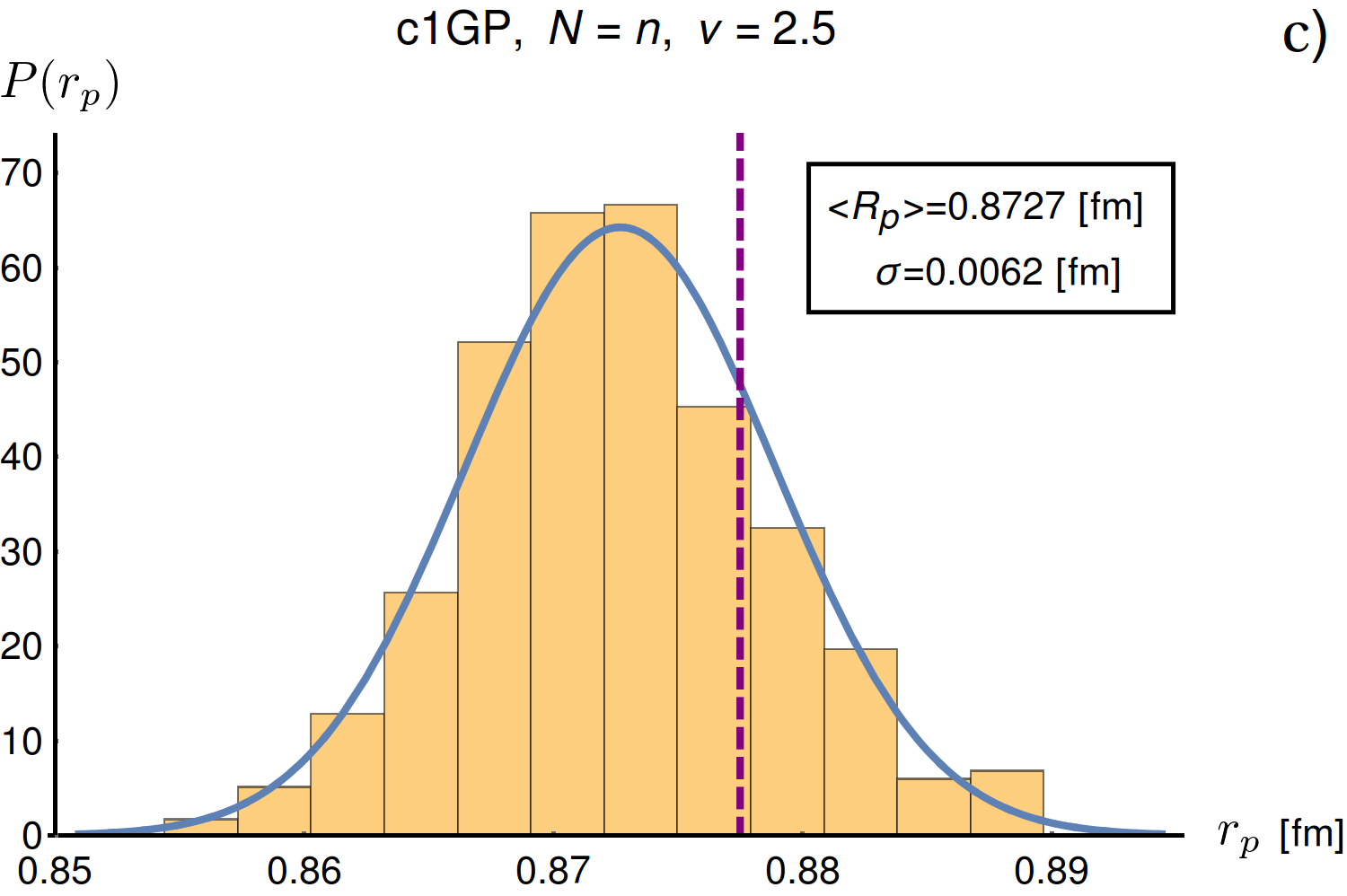}
  \includegraphics[scale = 0.55]{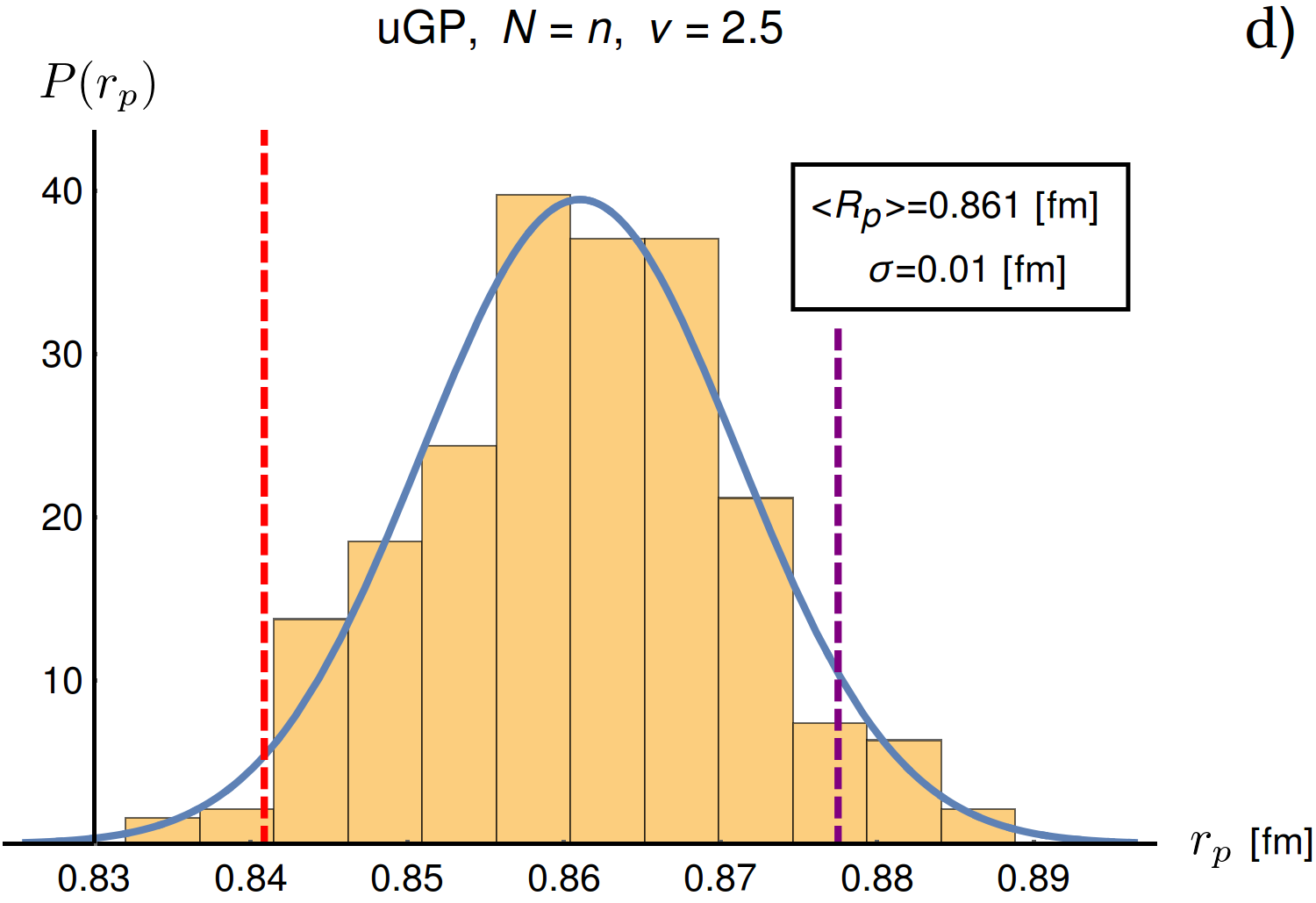}
  \end{center}
\caption{Histogram plots of MCMC samples of the radius $r_p$ for cGP (a), c$_0$GP (b), c$_1$GP (c) and uGP (d) with $N = n$ and $\nu = 2.5$ in the low regime. The red and purple vertical dashed lines indicate the values of 0.84087 fm and 0.8775 fm respectively.}
\label{fig:lowhist_n}
\end{figure}

\begin{figure}
\begin{center}
   \includegraphics[scale = 0.5]{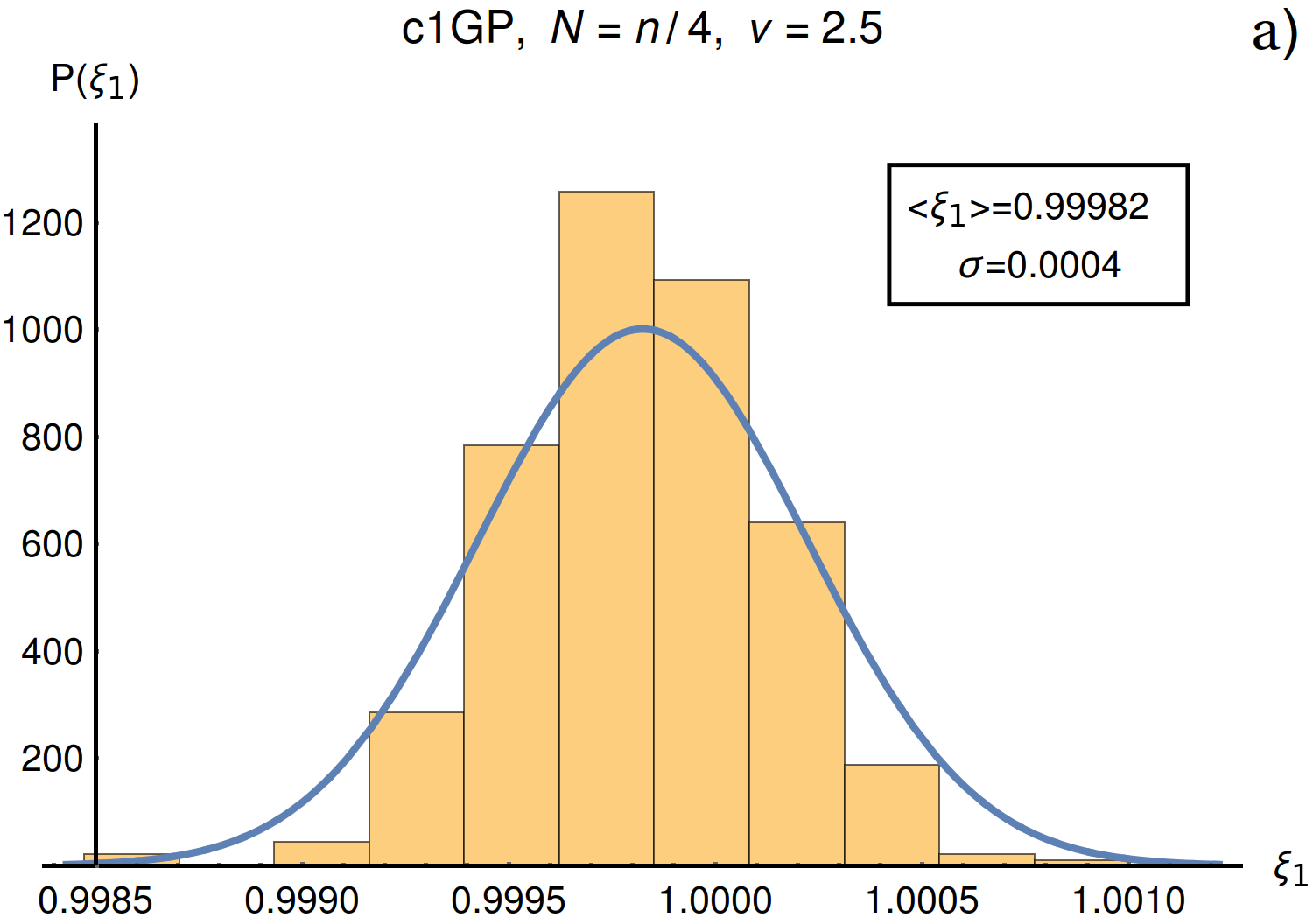}
  \includegraphics[scale = 0.5]{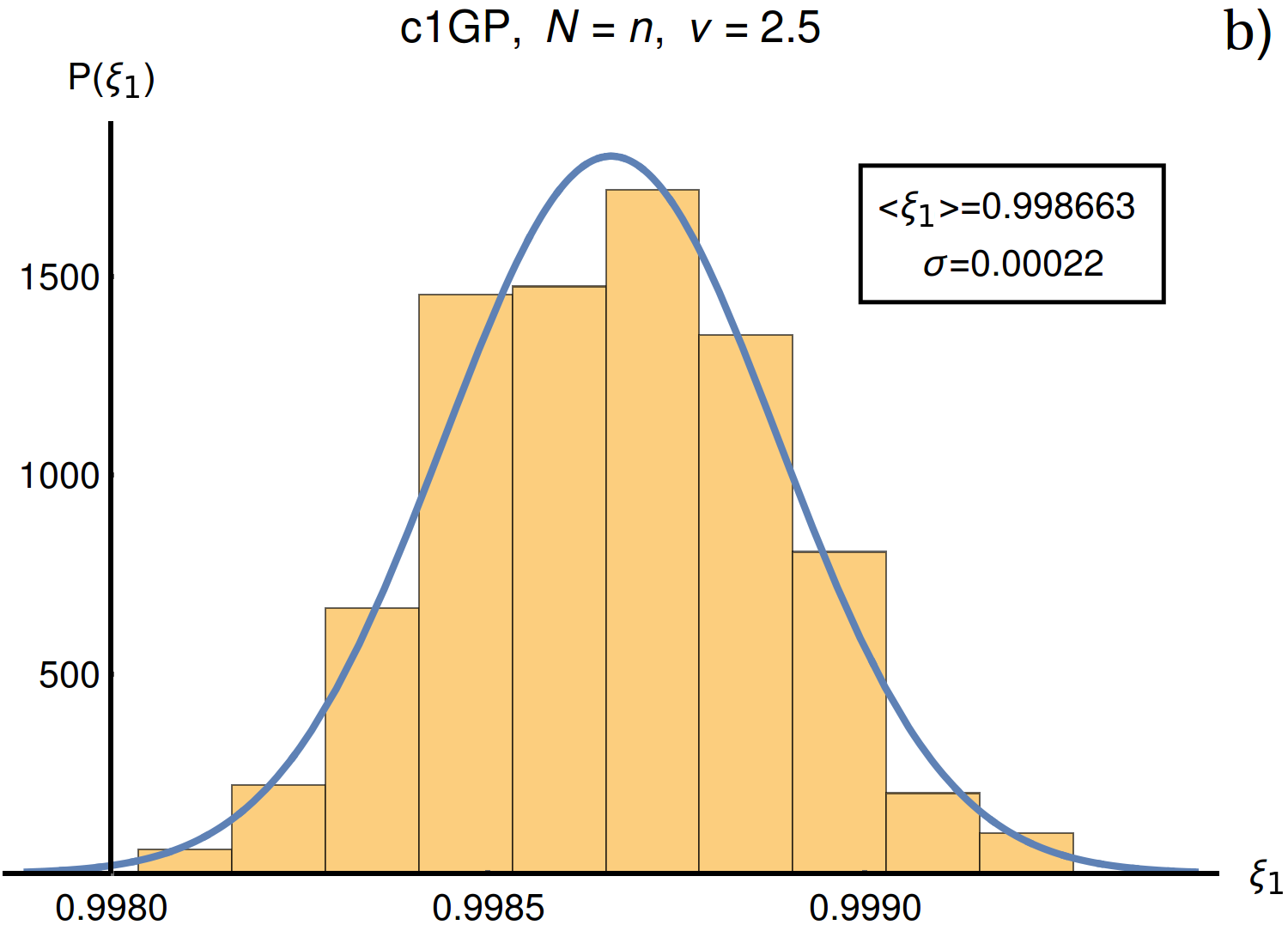}
  \\
   \includegraphics[scale = 0.5]{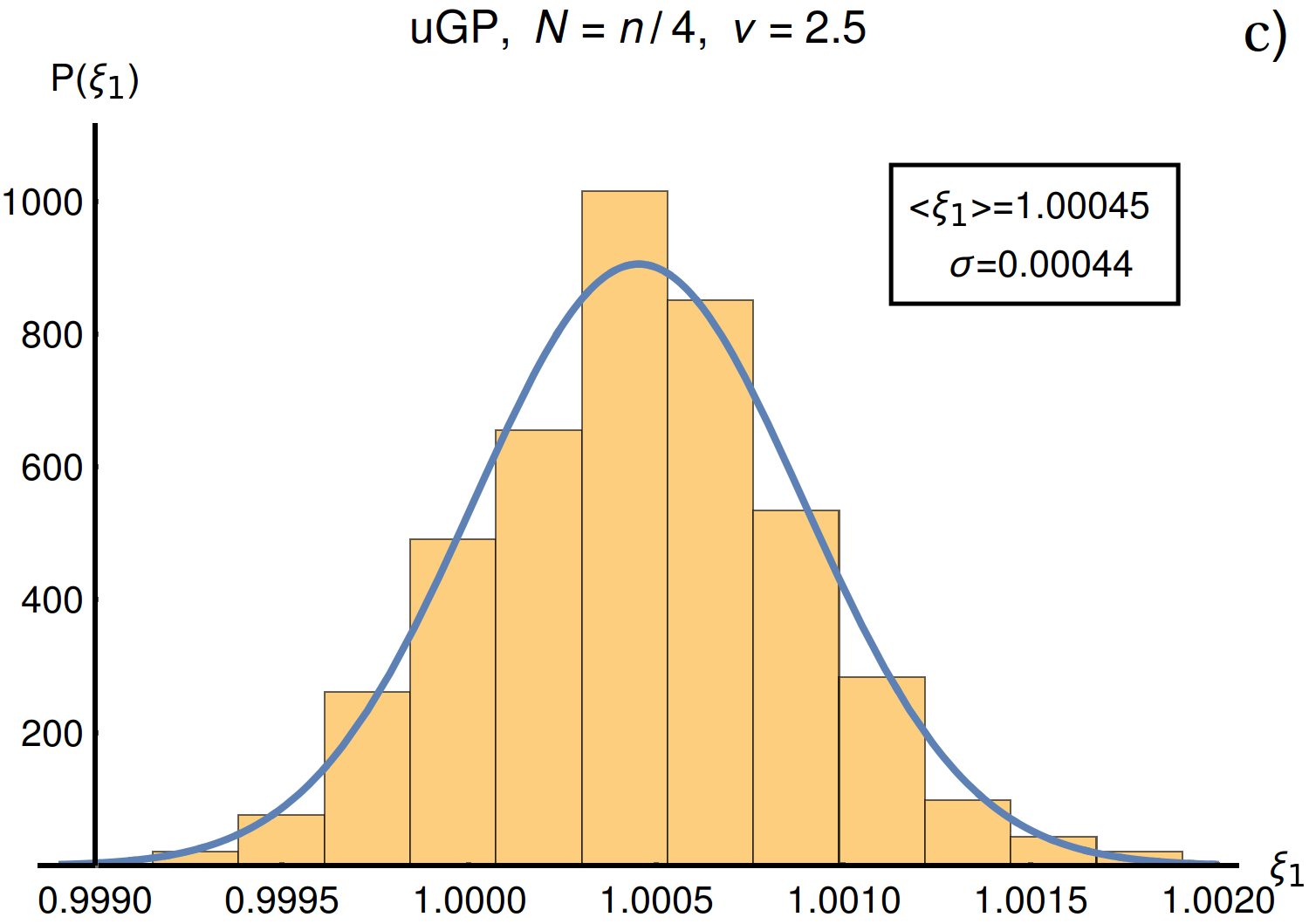}
  \includegraphics[scale = 0.5]{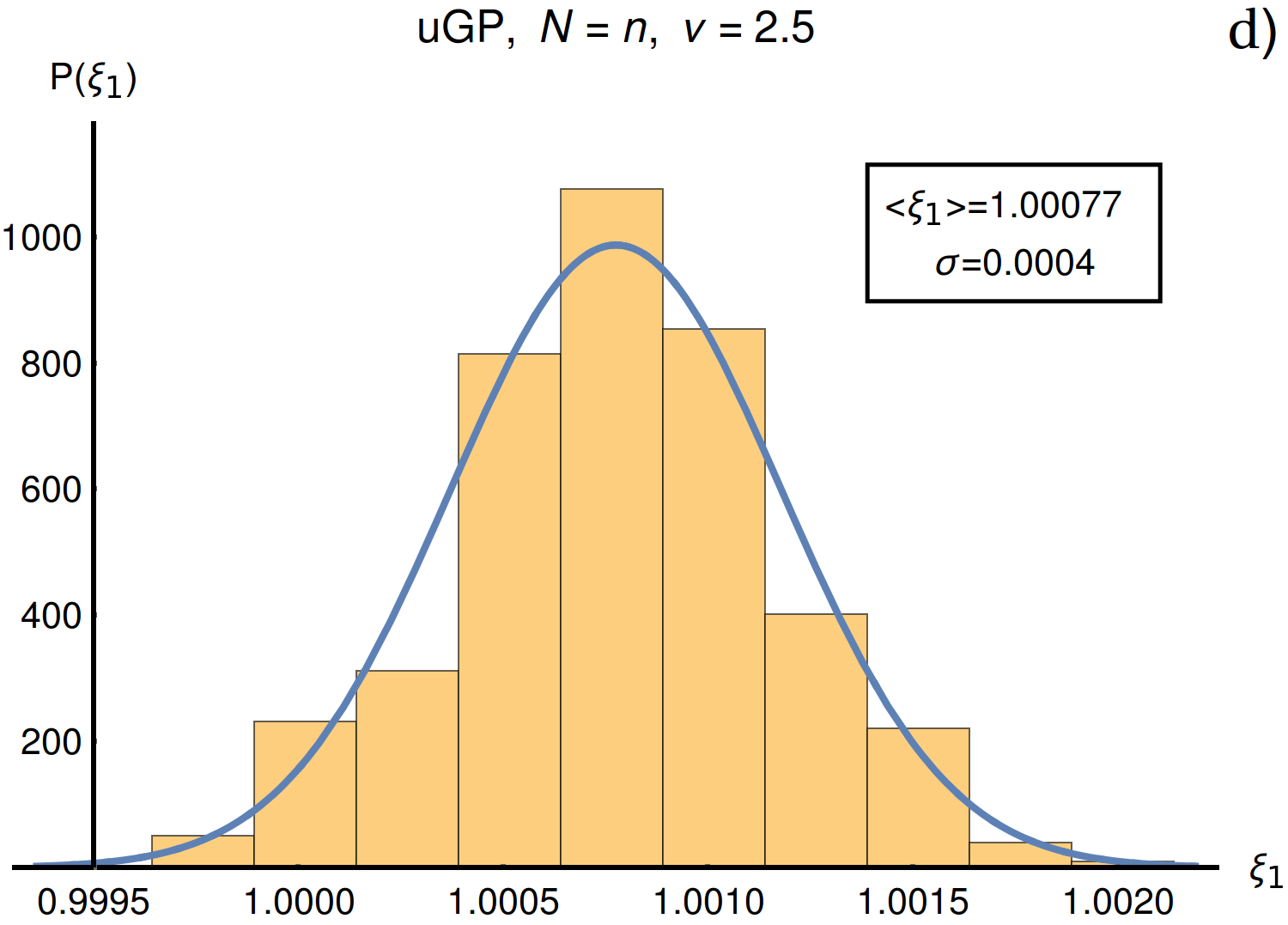}
  \end{center}
\caption{Histogram plots of MCMC samples of $\xi_1$ ($n_0G_E(0)$) for c$1$GP and uGP with $N = \{n/4, n\}$ and $\nu = 2.5$ for the full dataset. (a) ($N=n/4$) and (b) ($N=n$) show the results of c$_1$GP, (c) ($N=n/4$) and (d) ($N=n$) show the results of uGP.}
\label{fig:highhist_f0}
\end{figure}

\begin{figure}
\begin{center}
   \includegraphics[scale = 0.5]{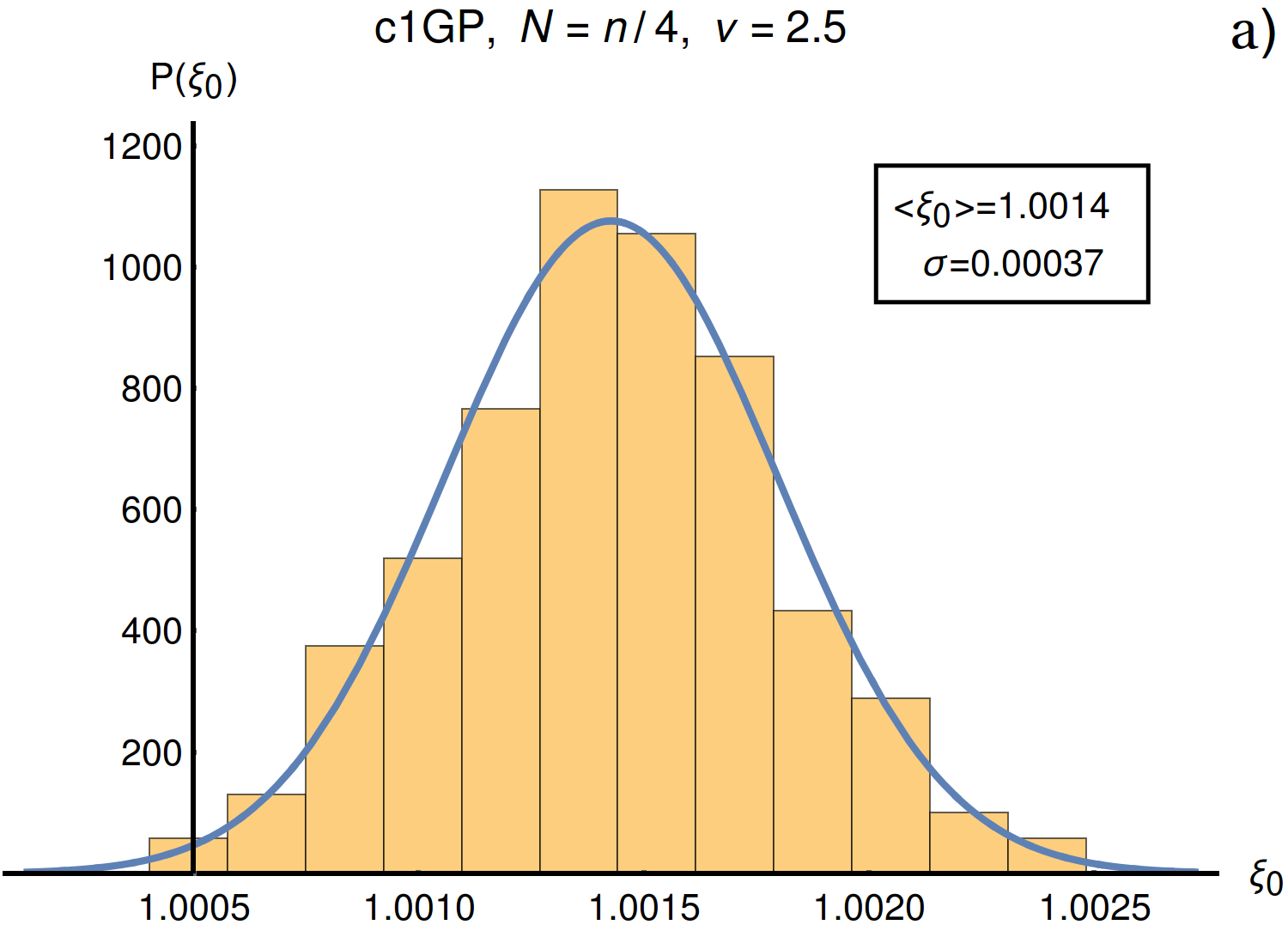}
  \includegraphics[scale = 0.5]{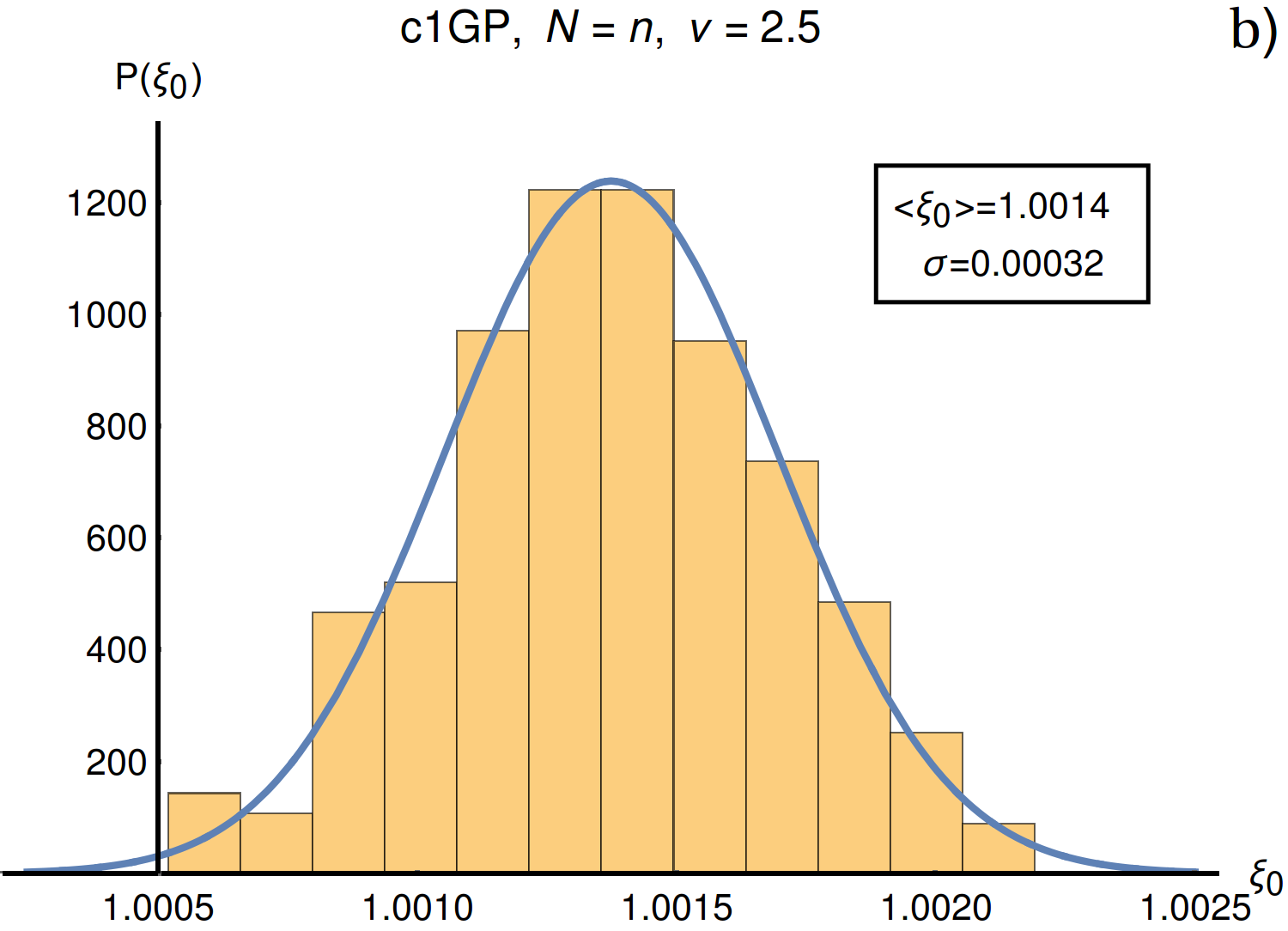}
  \\
   \includegraphics[scale = 0.5]{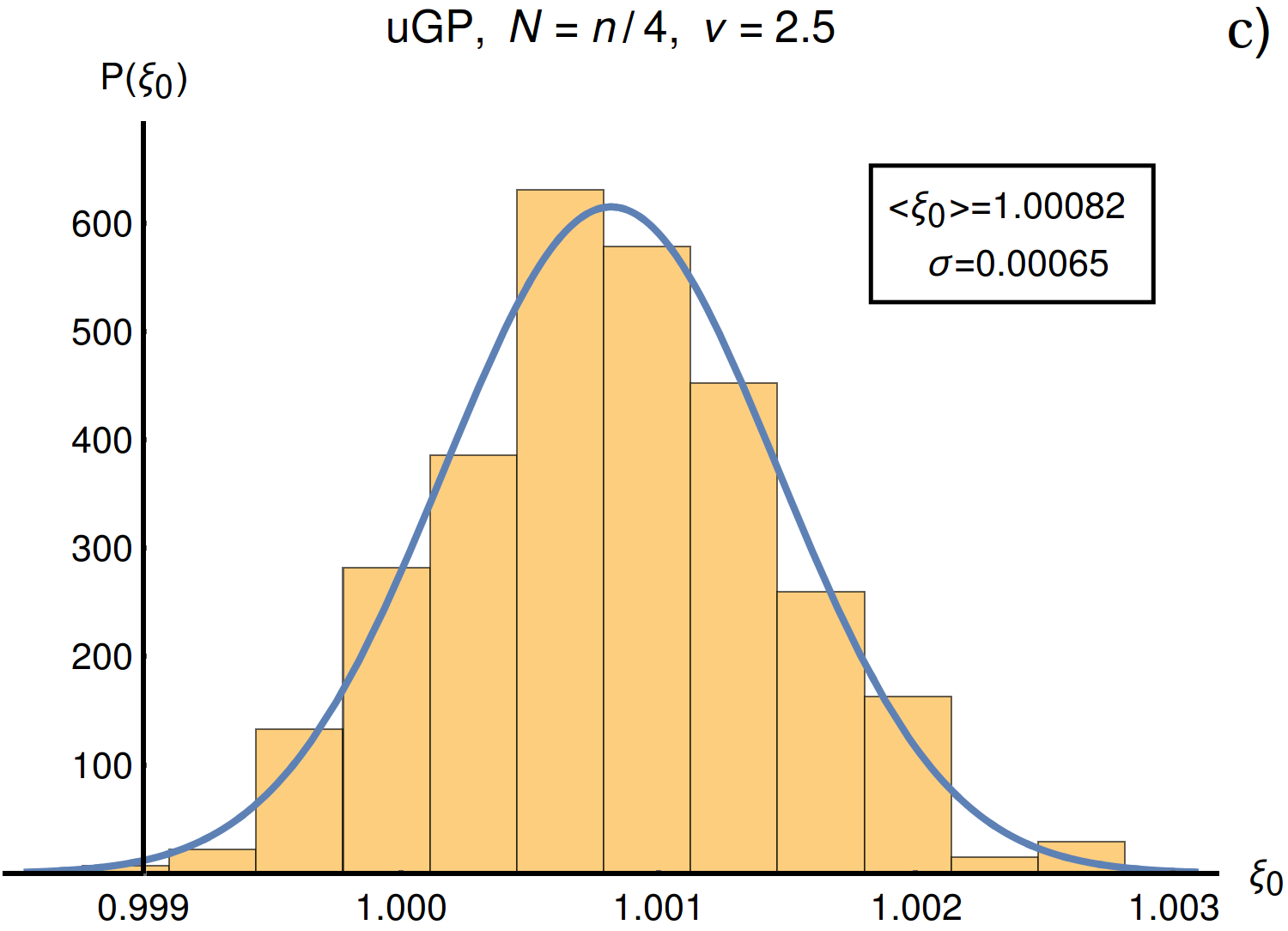}
  \includegraphics[scale = 0.5]{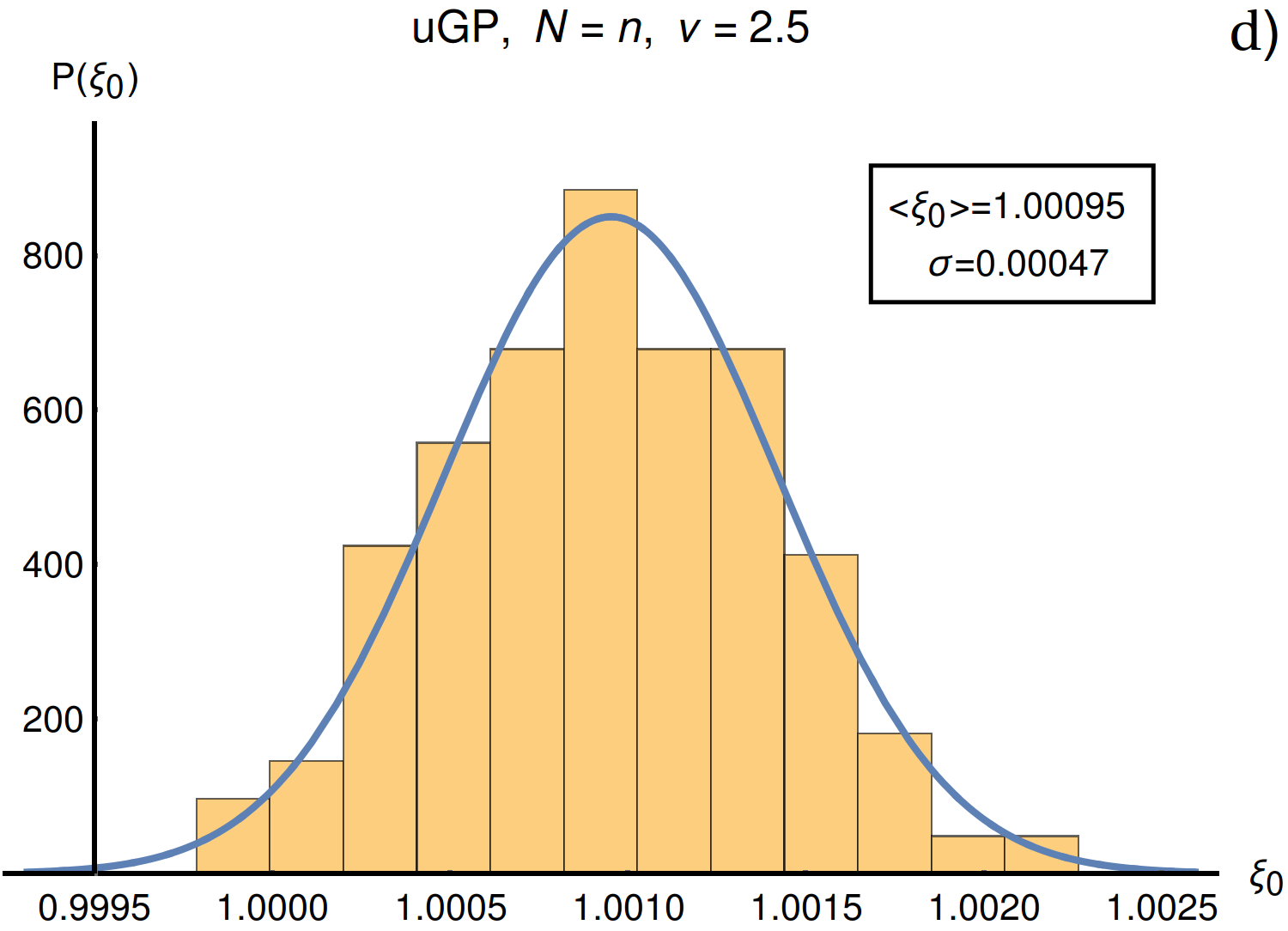}
  \end{center}
\caption{Histogram plots of MCMC samples of $\xi_1$ ($n_0G_E(0)$) for c$1$GP and uGP with $N = \{n/4, n\}$ and $\nu = 2.5$ for the low regime. (a) ($N=n/4$) and (b) ($N=n$) show the results of c$_1$GP, (c) ($N=n/4$) and (d) ($N=n$) show the results of uGP.}
\label{fig:lowhist_f0}
\end{figure}

\newpage
\clearpage

\bibliography{proton_final.bbl}
\end{document}